 \documentclass[11pt]{article}
 \usepackage{fullpage}
 
\usepackage{authblk} 
\usepackage{hyperref}
\usepackage{enumitem}
\usepackage{graphicx} 
\usepackage{xspace} 
\usepackage{amsmath,amsfonts,amssymb,amsthm} 
\usepackage{todonotes}
\usepackage{thm-restate}
\usepackage{cleveref}
\usepackage{physics}
\usepackage{tikz}
\usepackage{mathdots}
\usepackage{yhmath}
\usepackage{cancel}
\usepackage{color}
\usepackage{array}
\usepackage{multirow}
\usepackage{amssymb}
\usepackage{tabularx}
\usepackage{extarrows}
\usepackage{booktabs}
\usetikzlibrary{fadings}
\usetikzlibrary{patterns}
\usetikzlibrary{shadows.blur}
\usetikzlibrary{shapes}
\usetikzlibrary{shapes.geometric, arrows, positioning, backgrounds}
\usepackage[most]{tcolorbox}

\newtheorem{theorem}{Theorem}[section]

\newtheorem{lemma}[theorem]{Lemma}
\newtheorem{proposition}[theorem]{Proposition}
\newtheorem{definition}[theorem]{Definition}

\newtheorem{problem}[theorem]{Problem}

\newtheorem{question}[theorem]{Question}

\newcommand{\BCONGEST}{\ensuremath{\textsf{BROADCAST CONGEST}}\xspace}
\newcommand{\CONGEST}{\ensuremath{\textsf{CONGEST}}\xspace}
\newcommand{\LOCAL}{\ensuremath{\textsf{LOCAL}}\xspace}
\newcommand{\OPT}{\ensuremath{\textsf{OPT}}\xspace}
\newcommand{\cnt}{\ensuremath{\textsf{count}}\xspace}

\newcommand{\FPT}{\ensuremath{\text{FPT}}\xspace}
\newcommand{\W}[1]{\ensuremath{\text{W}[#1]}}
\newcommand{\cO}{\mathcal{O}}

\newcommand{\N}{\mathbb{N}}
\newcommand{\ball}{\mathrm{ball}}
\newcommand{\polylog}{\mathrm{polylog}}

\def\true{\textsf{true}}

\def\id{\textsf{id}}
\def\adj{\textsf{adj}}
\def\lab{\textsf{lab}}
\def\depth{\textsf{depth}}
\def\FO{\mathrm{FO}}
\def\parent{\textsf{parent}}
\def\lca{\textsf{lca}}
\def\root{\textsf{root}}

\def\depth{\textsf{depth}}
\def\type{\textsf{type}}
\def\Type{\textsf{Type}}
\def\good{\textsf{good}}
\def\val{\textsf{val}}
\def\levels{\textsf{Levels}}
\def\level{\textsf{level}}
\def\Col{\textsf{Col}}
\def\col{\textsf{color}}

\def\parent{\textsf{parent}}
\def\depth{\textsf{depth}}

\def\Cert{\textsf{Cert}}

\def\ntcol{\textsf{Non-3-Coloring}}
\def\Disj{\textsf{Disj}}
\def\cross{\textsf{cr}}
\def\csixfree{C_6\textsf{-free}}

\tikzstyle{class} = [rectangle, rounded corners, draw=black, fill=white, text centered, minimum width=3cm, minimum height=1cm]
\tikzstyle{arrow} = [thick,->,>=stealth]
\tikzstyle{dottedline} = [dotted, thick]

\title{What Can Be Computed Locally Revisited\\
\Large{--- First-Order Logic on Sparse Graphs in Distributed Computing\footnote{The research leading to these results has been supported by the Research Council of Norway via the project BWCA (314528), and by Université Paris Cité via the Emergence project METALG. F.~Fomin partially supported by European Research Council (ERC) via grant NewPC, reference 101199930. P.~Fraigniaud was partially supported by the ANR projects DUCAT (ANR-20-CE48-0006), PREDICTIONS (ANR-23-CE48-0010), ENEDISC (ANR-24-CE48-7768-01), and the French PEPR integrated project EPiQ (ANR-22-PETQ-0007). P.~Montealegre was supported by ECOS-ANID  ECOS240020, STIC-AMSUD ECODIST AMSUD240005,  Centro de Modelamiento Matemático (CMM), FB210005, BASAL funds
for centers of excellence from ANID-Chile, FONDECYT 1230599 and ANID-MILENIO-NCN2024\_103. I.~Rapaport has received funding from FB210005, BASAL funds for centers of excellence from ANID-Chile, and FONDECYT 1220142.} ---}} 

\author[1]{Lélia Blin}
\author[2]{Fedor V. Fomin} 
\author[1]{Pierre Fraigniaud}
\author[1,6]{Sylvain Gay}
\author[2]{Petr Golovach}
\author[3]{\\Pedro Montealegre}
\author[4]{Ivan Rapaport}
\author[5]{Ioan Todinca}

\affil[1]{Université Paris Cité, CNRS, IRIF, Paris, France}
\affil[2]{University of Bergen, Bergen, Norway}
\affil[3]{Universidad Adolfo Ibañez, Santiago, Chile}
\affil[4]{DIM-CMM, Universidad de Chile, Santiago, Chile}
\affil[5]{LIFO, Université d'Orléans and INSA Centre-Val de Loire, Orléans, France}
\affil[6]{École Normale Supérieure, Paris, France}
\date{}

\begin{document}

\maketitle
\thispagestyle{empty}
\begin{abstract}
The question of ``what can be computed locally?'' lies at the heart of \emph{distributed computing in networks}. As established in Naor and Stockmeyer's seminal paper (STOC 1993, Edsger W.~Dijkstra Prize in Distributed Computing 2025), this question is undecidable, even for graph problems whose solutions can be checked locally. In this paper, we adopt a novel perspective on the question, by asking for which classes $\Pi$ of problems, and for which classes $\mathcal{G}$ of graphs, \emph{all} problems in~$\Pi$ can be solved efficiently in a distributed manner in \emph{all} graphs of~$\mathcal{G}$. This paper focuses on two natural candidates for such an approach, namely the class of problems expressible in \emph{first-order logic (FO)}, because they possess an intrinsic form of locality thanks to Gaifman's theorem, and the class  of graphs with \emph{bounded expansion}, because they form a large class of graphs encompassing, e.g., planar, bounded-genus,  bounded-treewidth, and bounded-degree graphs, as well as graphs excluding a fixed minor or topological minor, sparse Erd\H{o}s--R\'enyi graphs (a.a.s.), and several network models such as stochastic block models for suitable parameter ranges.

The starting point of our work is the decade-old open question of Ne\v{s}et\v{r}il and Ossona de Mendez (Distributed Computing~2016) on the distributed complexity of \emph{local} FO formula on graphs of bounded expansion, in the standard \CONGEST model of distributed computing. Recall that a formula $\varphi(x)$ is local if the satisfaction of $\varphi(x)$ depends only on the $r$-neighborhood of its free variable~$x$, for some fixed~$r$. For instance, the formula ``$x$ belongs to a triangle'' is local. We resolve the open problem of Ne\v{s}et\v{r}il and Ossona de Mendez positively by showing that, for every local FO formula~$\varphi(x)$, and for every graph class $\mathcal{G}$ of bounded expansion, there exists a deterministic algorithm that identifies, for every $n$-vertex graph $G\in \mathcal{G}$,  all vertices $v$ of $G$ such that $G\models \varphi(v)$, in $\cO(\log n)$ rounds. The requirement of locality is unavoidable, as even the simple FO formula ``there exist two vertices of degree~3'' requires $\Omega(D)$ rounds in \CONGEST, even on trees of diameter~$D$. Nevertheless, we establish a second result, which goes beyond the question of Ne\v{s}et\v{r}il and Ossona de Mendez. We show that $\cO(D+\log n)$ rounds are sufficient for deciding \emph{any} FO formula~$\varphi$ on graphs of bounded expansion. That is, the overhead to be paid over the diameter is just $\cO(\log n)$. We underline that the techniques behind our two distributed ``meta-theorems'' extend to distributed \emph{counting}, \emph{optimization}, and \emph{certification} problems.  

Our results are tight in several ways. Regarding the choice of the graph class~$\mathcal{G}$, we show that deciding FO formulas may have high round complexity in \CONGEST\ on larger classes of graphs, even if they remain \emph{sparse}. For instance, the simple local FO formula expressing $C_6$-freeness requires $\tilde{\Omega}(\sqrt{n})$ rounds to be decided in graphs of degeneracy~2 with constant diameter. Regarding the choice of the class~$\Pi$ of problems, we show that deciding problems expressible in monadic second-order (MSO) logic may have high round complexity in \CONGEST, even in classes of graphs with bounded expansion. For example, deciding non-3-colorability requires $\tilde{\Omega}(n)$ rounds in bounded-degree graphs with logarithmic diameter. 



\medskip

\end{abstract}

\thispagestyle{empty}

\newpage

\setcounter{tocdepth}{2}
\tableofcontents
\thispagestyle{empty}

\newpage

\setcounter{page}{1}

\section{Introduction}

\subsection{Context and Objective}


Distributed computing in networks is primarily concerned with identifying which graph problems can be efficiently solved “by the graphs themselves.” That is, each vertex of the graph is assumed to be a computing device, referred to as a \emph{node}, and nodes communicate by exchanging messages along the edges of the graph. Initially, each node knows only local information (e.g., its identifier, its degree, the weights of its incident edges if any, etc.), and it must compute its own output (e.g., whether it belongs to the computed minimum dominating set, or which of its incident edges belong to the computed minimum-weight spanning tree). Each node must do so autonomously, based only on the information obtained by exchanging data with its neighbors.

In this framework, the nodes are typically subject to two constraints: \emph{distance} and \emph{bandwidth}. The former merely expresses that exchanging information between two nodes at distance $d$ in the graph requires at least $d$ rounds of communication, for traversing a path of length $d$ between these two nodes. The latter is motivated by the fact that network links may have limited capacity. In $n$-node graphs, this is usually captured by limiting  communication to transmitting a single message of $\cO(\log n)$  bit per edge in each direction at each round. To sum up, distributed computation in networks proceeds as a sequence of synchronous rounds where, at each round, every node (1)~sends  $\cO(\log n)$-bit messages to each of its neighbors, (2)~receives the messages sent by its neighbors, and (3)~performs some individual computation. This model is commonly referred to as \CONGEST in the literature~\cite{Peleg2000}. In \CONGEST, each node is given an $\cO(\log n)$-bit identifier, which is unique in the network. 

\paragraph{Local and Global Problems.}

Two categories of graph problems can be identified with respect to their solvability in \CONGEST. \emph{Global problems} are problems that require $\Omega(D)$ rounds to be solved in graphs of diameter~$D$. This includes, e.g., minimum-weight spanning tree (MST), as deciding whether an edge belong to the tree may depend on the weight of another edge, at distance~$\Omega(D)$. In contrast, \emph{local problems} are problems that can be solved in $\polylog(n)$ rounds, even in graphs with 
 diameter $\Omega(n)$. 

\begin{itemize}
\item For a global problem, the typical objective is  to express its round complexity as $\cO(D+h(n))$ rounds, and to minimize the overhead function~$h(\cdot)$. For instance, the overhead for MST is $\tilde{\cO}(\sqrt{n})$ rounds\footnote{The $\tilde{\cO}$ and $\tilde{\Omega}$ notation ignores $\polylog(n)$ factors.} \cite{garay1993sub}, and this is optimal~\cite{PelegR00}. For other examples, such as SSSP, Min-Cut, routing, min-weight cycles, etc., see, e.g., \cite{DoryEMN21,LenzenP13,ManoharanR24,Nanongkai14}, and see~\cite{SarmaHKKNPPW12} for lower bounds. 

\item For local problems, the challenging question of ``what can be computed locally?'', to quote Naor and Stockmeyer's seminal paper~\cite{NaorS95}, has been mostly addressed by considering \emph{locally checkable labeling} (LCL) problems\footnote{The paper~\cite{NaorS95} actually considered the \LOCAL  model~\cite{Linial92}, which is \CONGEST without any restriction on the size of the messages. Nevertheless, since then, the question as been addressed in both \LOCAL and \CONGEST in the literature.}
  Roughly, LCL problems are problems for which an assignment of labels to the nodes or edges is a correct solution if and only if this assignment  is locally correct at each node. This is for instance the case of proper $k$-coloring, maximal independent set (MIS), maximal matching, etc. In other words, LCL problems are problems whose solutions can be checked at each node by inspecting its ball of constant radius around it. Determining which LCL problems are locally solvable (i.e., in $\polylog(n)$ rounds) has given rise of an enormous amount of work during the last two decades (see, e.g., \cite{Faour0GKR23,FischerGK17,GhaffariK21}). 
\end{itemize}

In this paper, we tackle local and global problems from a different perspective, which has the side effect of placing these two types of problems under the same umbrella whereas they are so far somehow treated separately in the literature (see above). 

\paragraph{The Novel ``Logical Approach'' of Distributed Locality.}

Another line of research aiming at identifying which problems can be solved locally in \CONGEST has  emerged in the years 2010s~\cite{NesetrilM16}, motivated by Gaifman's locality theorem~\cite{gaifman1982local}. Recall that Gaifman's theorem states that every first-order (FO) sentence is equivalent to a Boolean combination of local sentences. That is, in graphs, every FO sentence is equivalent to a sentence saying that there exist nodes that are far apart from one another in the graph, and each satisfies some local condition described by an FO formula referring to the vicinity of the node. 

Considering graph problems expressible in FO logic opens a completely new approach of locality in distributed computing. Note that the class of problems expressible in FO is quite rich. In particular, it includes problems as diverse as  subgraph isomorphism  (induced or not), $k$-dominating set, $k$-independent set, diameter at most~$k$, etc. For instance,  triangle-freeness can be expressed as 
\begin{equation}\label{eq:FO-triangle-freeness}
\neg\Big(\exists x \, \exists y \, \exists z \; \big(\adj(x,y) \land \adj(y,z) \land \adj(z,x)\big)\Big), 
\end{equation}
where $\adj$ denotes the adjacency predicate. Clearly,  every FO problem expressed as the combination of at least two local sentences cannot be decided in less than $\Omega(D)$ rounds in graphs of diameter~$D$, as the two local sentences may be satisfied by far away nodes in the network. For instance, the mere problem ``there exist at least two nodes of degree~3'' requires $\Omega(D)$ rounds to be checked in \CONGEST, even in trees. On the other hand, FO problems expressed as a unique local sentence may or may not be decidable in $\polylog(n)$ rounds. For instance, deciding\footnote{In distributed decision, a positive instance, i.e., an instance satisfying the predicate, must be accepted by all nodes, whereas a negative instance, i.e., an instance not satisfying the predicate, must be rejected by at least one node.}  whether the graph does not contain a given tree~$T$ as a subgraph (i.e., deciding $T$-freeness) can be done in a constant number of rounds~\cite{EvenFFGLMMOORT17}.  However, deciding $C_4$-freeness cannot be decided in less than $\widetilde{\Omega}(\sqrt{n})$ rounds~\cite{DruckerKO13}. This gives rise to a new question, complementary to the aforementioned  question by Naor and Stockmeyer. 

\begin{center}
\begin{tcolorbox}[colback=blue!5!white, colframe=blue!40!black,
                  width=1\textwidth, boxrule=0.6pt,
                  arc=3pt, auto outer arc,
                  left=8pt, right=8pt, top=6pt, bottom=6pt]
Instead of asking which problems can be solved locally in \emph{all} graphs, 
we focus on a broad class of problems, namely those expressible in FO, 
and ask for which classes of graphs \emph{all such problems} 
can be solved efficiently.
\end{tcolorbox}
\end{center}

%

More precisely, the paper addresses the following problem. 

\begin{problem}\label{pb:which-graph-classes}
Which graph classes $\mathcal{G}$ satisfy the following two conditions?
\begin{itemize}
\item Every \emph{local} FO property can be decided in $\cO(\polylog(n))$ rounds in the \CONGEST\ model for every $n$-node graph belonging to the class $\mathcal{G}$;
\item Every FO property can be decided in $\cO(D+\polylog(n))$ rounds in the \CONGEST\ model for every $n$-node graph of diameter~$D$ belonging to the class $\mathcal{G}$. 
\end{itemize}
\end{problem}

For instance, do the above two conditions hold for the class of planar graphs? 
Or for random graphs (with some statistical guarantees)? 
Or for graphs with constant maximum degree~$\Delta$? 
Nešetřil and Ossona de Mendez~\cite{NesetrilM16} partially addressed this question in the more general setting of graph classes with \emph{bounded expansion}, as part of the development of their \emph{sparsity theory}~\cite{NesetrilM06,NesetrilOdM12}. 
Their findings have important applications to the subgraph isomorphism problem.

\paragraph{Deciding Subgraph Isomorphism in Graphs of Bounded Expansion.}

Graph classes with bounded expansion form a rich graph family (see Figure~\ref{fig:graphclasse}). For instance, the class of planar graphs, the class of graphs with bounded genus, the class of graphs of constant maximum degree~$\Delta$, and the class of minor-free graphs and topologically minor-free graphs are all of bounded expansion, as well as sparse Erdős–Rényi graphs (a.a.s.).
Moreover, \cite{DemaineRRVSS19} demonstrated that many network models (e.g., stochastic block models) exhibit bounded expansion for some ranges of parameters, and these findings are even supported with empirical measurements on a corpus of real-world networks. All these points make graph classes of bounded expansion quite attractive from both conceptual and practical perspectives.

\begin{figure}
\begin{center}

%
\begin{tikzpicture}[node distance=1.5cm, scale=0.8, every node/.style={transform shape}]


\node[class] at (0,-1.5) (boundexp) {Bounded expansion};
\node[class, below of=boundexp] (excludedtopo) {Excluded topological subgraph};
\node[class, below left=1cm and 1cm of excludedtopo] (excludedminor) {Excluded minor};
\node[class, below right=1cm and 1cm of excludedtopo] (boundeddeg) {Bounded degree};
\node[class, below of=excludedminor] (boundedgenus) {Bounded genus};
\node[class, right=2cm of boundedgenus] (boundedtreewidth) {Bounded treewidth};
\node[class, below of=boundedgenus] (planar) {Planar};
\node[class, below of=boundedtreewidth] (boundtreedepth) {Bounded treedepth};




\draw[arrow] (boundexp) -- (excludedtopo);
\draw[arrow] (excludedtopo) -- (excludedminor);
\draw[arrow] (excludedtopo) -- (boundeddeg);
\draw[arrow] (excludedminor) -- (boundedgenus);
\draw[arrow] (excludedminor) -- (boundedtreewidth);
\draw[arrow] (boundedtreewidth) -- (boundtreedepth);
\draw[arrow] (boundedgenus) -- (planar);

\end{tikzpicture}
\caption{Different classes of graphs, with their mutual inclusion relations. An arrow from a class $\mathcal{G}$ to a class $\mathcal{G}'$ indicates that $\mathcal{G}'$ is (strictly) included in $\mathcal{G}$.  
}
\label{fig:graphclasse}
\end{center}
\end{figure}

In particular, the study of graphs with bounded expansion has already driven major advances in model checking~\cite{DvorakKT13}, parameterized complexity~\cite{NesetrilM06}, kernelization \cite{drange2016kernelization,eickmeyer2017neighborhood} and approximation algorithms~\cite{Har-PeledQ17}. In the context of distributed computing, the subgraph isomorphism  can be efficiently solved in graphs of bounded expansion. Specifically, the proposition below states that, for any fixed connected graph~$H$ (e.g., a triangle~$C_3$), $H$-freeness can be decided in $\cO(\log n)$ rounds in any $n$-node graph $G$ picked in a class of graphs with bounded expansion. That is, if $G$ contains $H$ as a subgraph, at least one node rejects, otherwise all nodes accept. 

\begin{proposition}[Nešetřil and Ossona de Mendez~\cite{NesetrilM16}]
For every graph class $\mathcal{G}$ of bounded expansion, and for every connected graph $H$, there exists a distributed algorithm that, for every $n$-node graph $G\in \mathcal{G}$, decides whether $G$ does not contain $H$ as a subgraph in $\cO(\log n)$ rounds in \CONGEST.
\end{proposition}

Actually, the algorithm in~\cite{NesetrilM16} does more than just deciding $H$-freeness. It actually marks all nodes of the input graph $G$ belonging to an isomorphic copy of~$H$ (these nodes are the ones that reject whenever $G$ is not $H$-free). Note that it is known~\cite{Censor-Hillel21,DruckerKO13} that, for some~$H$ (e.g., $H=C_{2k+1}$ for $k>1$), deciding $H$-freeness requires $\widetilde{\Omega}(n)$ rounds in arbitrary graphs. Instead, for all graphs~$H$, $H$-freeness can be decided in just $\cO(\log n)$ rounds in graphs of bounded expansion, and thus in particular in planar graphs. 

\paragraph{Intuition for Planar Graphs.}
The algorithm of Nešetřil and Ossona de Mendez for subgraph isomorphism~\cite{NesetrilM16} is important for our work, so we devote some space to highlight their approach. Their algorithm works roughly as follows. We provide the intuition for the case of planar graphs. It is known~\cite{NesetrilOdM12}  that there exists a function $f:\mathbb{N}\to\mathbb{N}$ such that, for every integer $k\geq 1$, every planar graph $G$ can be properly colored with $f(k)$ colors so that, for every set $S$ of at most $k$ colors, the subgraph of $G$ induced by the nodes colored with a color in $S$ has \emph{treedepth} at most~$|S|$. Intuitively, treedepth measures how far a graph is from being a forest of bounded depth. (More formally,  a graph has treedepth $d$ if there exists a rooted forest $F$ on the same set of vertices as the graph, in which all trees have depth at most~$d$, such that every edge of $G$ is between a node and one of its ancestors in~$F$.)  Graphs of bounded treedepth are therefore appealing from an algorithmic perspective.  As a side remark, the existence of such a function $f$ leading to colored subgraphs of bounded treedepth characterizes the graph classes of bounded expansion, but we focus on planar graphs here for concreteness.

Let us fix $k=|V(H)|$, i.e., $k$ is the number of nodes of~$H$. By coloring the input graph $G$ with $f(k)$ colors, there must exist a set $S$ of $k$ colors such that, if $H$ is contained in $G$ as a subgraph, then at least one copy of $H$ appears in the subgraph $G[S]$, where $G[S]$ denotes the subgraph of $G$ induced by the nodes with colors in~$S$. Without going into details, since $G[S]$ has bounded treedepth, deciding $H$-freeness in $G[S]$  can actually be done in $O(1)$ rounds~\cite{FominFMRT24}. It is therefore sufficient to test all choices of $k$ colors among the $f(k)$ colors of the palette, and for a node to reject if it belongs to an isomorphic copy of~$H$ for one of these ${f(k) \choose k}$ choices. The main contribution in~\cite{NesetrilM16} is in fact the design of a distributed algorithm for \CONGEST that, given any positive integer~$k$, computes an appropriate $g(k)$-coloring of any $n$-node planar graph~$G$ in $\cO(\log n)$ rounds, for some function $g:\mathbb{N}\to\mathbb{N}$. In fact, the algorithm works not only for planar graphs, but also for any graph class $\mathcal{G}$ of expansion $f:\mathbb{N}\to\mathbb{N}$. For a given $k$, the number of colors $g(k)$ used by the algorithm may be larger than $f(k)$, but the crucial point is that the subgraph induced by any set of $k$ colors among the $g(k)$ colors has bounded treedepth.

\paragraph{Beyond Subgraph Isomorphism.}
At this point, it is worth noticing that $H$-freeness for a connected graph $H$ is expressible by an FO formula of a very specific form: it is merely the negation of an existential formula; see, e.g., Eq.~\eqref{eq:FO-triangle-freeness}. Therefore, checking whether such a formula holds boils down to searching for $k=|V(H)|$ nodes satisfying the adjacency conditions specified by~$H$. Since these $k$ nodes must appear in one of the ${f(k) \choose k}$ subgraphs $G[S]$ of bounded treedepth, this is (to some extent) straightforward. However, the situation becomes immediately more challenging once universal quantifiers are introduced. 
For instance, consider checking the existence of \emph{twins} in planar graphs, that is, the existence of two nodes $v$ and $v'$ with the same neighborhood, i.e., satisfying $N_G(v)\smallsetminus \{v'\} = N_G(v')\smallsetminus \{v\}$. The existence of twin nodes can be written as (assuming all variables refer to distinct vertices): 
\[
\exists x \, \exists y \, \forall z \; \Big( \big(\adj(x,z) \land \adj(y,z) \big) \; \lor \; \big(\neg\adj(x,z) \land \neg\adj(y,z) \big)\Big).
\]
The formula above has only three variables, as is the case for Eq.~\eqref{eq:FO-triangle-freeness}. However, here the variable $z$ is universally quantified, whereas all variables in Eq.~\eqref{eq:FO-triangle-freeness} are existentially quantified. As a consequence, it is no longer true that it suffices to search for three nodes in each of the ${f(3) \choose 3}$ subgraphs $G[S]$ of bounded treedepth in order to determine whether there are twin nodes in a planar graph~$G$. Indeed, two twins may each have up to $\Omega(n)$ neighbors, which may not appear together in any subgraph of $G$ induced by a constant number of colors.

Ne\v{s}et\v{r}il and Ossona de Mendez~\cite{NesetrilM16} therefore posed a broader and more challenging question, later reiterated by Pilipczuk, Siebertz, and Toruńczyk~\cite{PilipczukST18}, and by Fomin, Fraigniaud, Montealegre, Rapaport, and Todinca~\cite{FominFMRT24}. Recall that an FO formula $\varphi(x)$ is \emph{local} if the satisfaction of $\varphi(x)$ depends only on the $r$-neighborhood of its free variable~$x$, for some fixed~$r$. For example, in graphs, the formula ``node $x$ has a twin'' is local. Another example of a local formula $\varphi(x)$ is the following, for fixed positive integers~$d$ and~$k$: ``the distance-$d$ neighborhood of~$x$ in $G$ has a dominating set of size at most~$k$''. The question in~\cite{NesetrilM16} can be formulated as follows: 

\begin{question}[Ne\v{s}et\v{r}il and Ossona de Mendez~\cite{NesetrilM16}]
\label{prob:NM}
Is there, for every local first-order (FO) formula~$\varphi(x)$ and for every graph class $\mathcal{G}$ of bounded expansion, a distributed algorithm that, for every $n$-node graph $G\in \mathcal{G}$, marks all vertices $v\in V(G)$ such that $G\models \varphi(v)$ in $\cO(\log n)$ rounds in the \CONGEST\ model?
\end{question}

Note that a positive answer to this question implies that all graph classes of bounded expansion satisfy the first item of  Problem~\ref{pb:which-graph-classes}.

\subsection{Our Results}
\label{ss:results}

Our first main result resolves \Cref{prob:NM}. 

\begin{theorem}
\label{thm:thmFOlocal-informal}
For every local FO formula $\varphi(x)$ on graphs, and for every class of graphs $\mathcal{G}$ of bounded expansion, there exists a deterministic algorithm that, for every $n$-node network $G\in \mathcal{G}$, marks all vertices $v$ of $G$ satisfying $G\models \varphi(v)$, in $\cO(\log n)$ rounds in the \CONGEST\ model.
\end{theorem}

In fact, the result above also holds for FO formulas on \emph{labeled} graphs, i.e., graphs in which every node is provided with labels drawn from a finite set. In other words, formulas may also include \emph{unary} predicates on node labels (see Theorem~\ref{thm:thmFOlocal} for the formal statement). Before summarizing the techniques used to prove \Cref{thm:thmFOlocal-informal}, it is worth mentioning that we extend this result in several directions.

\subsubsection{Extension 1: Deciding FO formulas}

Our second result concern general FO formulas, for which the dependence on the diameter~$D$ is unavoidable. Again, the result below holds for labeled graphs too. For the sake of clarity, we state it here for unlabeled graphs only (see Theorem~\ref{thm:thmFOexpCONG} for the more general statement). 

\begin{theorem}
\label{thm:thmFOexpCONG-informal}
    For every FO formula $\varphi$ on graphs, and for every class of graphs $\mathcal{G}$ of bounded expansion, there exists a distributed algorithm that,  for every $n$-node network $G\in \mathcal{G}$  of diameter~$D$, decides whether $G\models\varphi$ in $\cO(D+\log n)$ rounds under the \CONGEST model.
\end{theorem}

Note that bounded degeneracy does not suffice to establish the same result as above\footnote{Any class of graphs with bounded expansion has bounded degeneracy, but the reciprocal does not hold. For instance, the class of cliques with all edges subdivided once has degeneracy~2, but is not of bounded expansion.}. Indeed, we show that the simple FO formula expressing $C_6$-freeness requires $\tilde{\Omega}(\sqrt{n})$ rounds to be checked in \CONGEST, even in graphs of degeneracy~2 and constant diameter (see \Cref{thm:lowerboundC6degeneracy}). Moreover, our result is optimal with respect to the class of formulas, in the sense that it cannot be extended to, e.g., MSO. In particular, we show that non-3-colorability requires $\tilde{\Omega}(n)$ rounds to be checked in bounded-degree graphs with logarithmic diameter (see \Cref{thm:lowerboundMSOinEXP}).

\paragraph{Take Away Message.}

Theorems~\ref{thm:thmFOlocal-informal} and~\ref{thm:thmFOexpCONG-informal} state that there are large classes of graphs satisfying the two conditions of Problem~\ref{pb:which-graph-classes}, namely, the graph classes with bounded expansion.

\begin{tcolorbox}[colback=blue!5!white, colframe=blue!40!black,
                  width=1\textwidth, boxrule=0.6pt,
                  arc=3pt, auto outer arc,
                  left=8pt, right=8pt, top=6pt, bottom=6pt]

In the \CONGEST\ model, for every graph class $\mathcal{G}$ of bounded expansion,  
\begin{itemize}
\item every local FO property can be decided in $\cO(\log n)$ rounds, and 
\item every FO property can be decided in $\cO(D+\log n)$ rounds.
\end{itemize}
\end{tcolorbox}

\noindent While the tight connections between problems expressible in FO logic and graph classes of bounded expansion in the context of sequential model checking are well established and well understood (see \Cref{sect:relatedwork}), prior to our work this was not the case for distributed computing. In the latter setting, our results demonstrate that:
\begin{enumerate}
    \item Locality of FO formulas implies locality of distributed decision, within just $\cO(\log n)$ rounds;
    \item Gaifman’s locality theorem for FO formulas implies only a small $\cO(\log n)$-round overhead beyond the unavoidable round complexity~$D$.
\end{enumerate}
Moreover, our contribution goes even beyond distributed decision, as detailed next.

\subsubsection{Extension 2: Counting, Optimization, and Certification}

\paragraph{Counting Solutions.}

We consider the mere problem of \emph{counting} solutions of FO formulas with free variables. We actually consider two forms of counting: 
\begin{itemize}
    \item \emph{Local counting}: Given a local formula $\varphi(x_1,\dots,x_k)$ with $k\geq 1$ free variables (e.g., $\adj(x_1,x_2)\land \adj(x_2,x_3) \land \adj(x_3,x_1)$, i.e., $x_1,x_2,x_3$ form a triangle), each node $v$ of the considered graph $G$ outputs a (possibly negative)
    integer $\nu(v)$ satisfying that the number of tuples $v_1,\dots,v_k$ for which $G\models \varphi(v_1,\dots,v_k)$ equals $\sum_{v\in V}\nu(v)$.
    
    \item \emph{Global counting}: Given a general formula $\varphi(x_1,\dots,x_k)$ with $k\geq 1$ free variables,  all nodes  output a same value equal to the number of tuples $v_1,\dots,v_k$ for which $G\models \varphi(v_1,\dots,v_k)$. 
\end{itemize}
The next theorem extends  \Cref{thm:thmFOlocal-informal,thm:thmFOexpCONG-informal} to local and global counting. Again, the results below also hold for labeled graphs (cf. \Cref{thm:thmFOexpCNT}).

\begin{theorem}\label{thm:thmFOexpCNT-informal}
For every FO formula $\varphi(x_1,\dots,x_k)$ on graphs with $k\geq 1$ free variables $x_1,\dots,x_k$, and for every class of graphs $\mathcal{G}$ of bounded expansion:
\begin{itemize}
    \item There exists a distributed algorithm that, for every $n$-node network $G\in \mathcal{G}$ of diameter~$D$, performs global counting in ${O(D+\log n)}$ rounds in \CONGEST;
    \item If the formula $\varphi(x_1,\dots,x_k)$  is local, then there exists a distributed algorithm that performs local counting in $\cO(\log n)$ rounds in \CONGEST. 
\end{itemize}
\end{theorem}

Note that, in contrast to our~\Cref{thm:thmFOexpCNT-informal}, the best known algorithm for (locally) counting triangles in general graphs takes $\tilde{\cO}(n^{1/3})$ rounds \cite{ChangPTZ21}.

\paragraph{Optimization Problems.}

Some optimization problems can also be expressed as finding optimal solutions for FO formulas with free variables, over weighted graphs, e.g., finding a maximum-weight triangle. (If no triangles exist in the graph, then all nodes reject.) Note that, as usual in the literature on \CONGEST~\cite{Peleg2000}, the weights assigned to the nodes or edges  are assumed to be polynomial in the network size so that any weight can be transmitted between neighbors in a single round of communication. 
Our next theorem extends  \Cref{thm:thmFOexpCONG-informal} to   weighted problems. For the version of labeled graphs, see \Cref{thm:thmFOexpOPT}.

\begin{theorem}\label{thm:thmFOexpOPT-informal}
For every FO formula $\varphi(x_1,\dots,x_k)$ on graphs with $k\geq 1$ free variables $x_1,\dots,x_k$, and for every class of graphs $\mathcal{G}$ of bounded expansion, there exists a distributed algorithm that, for every $n$-node network $G=(V,E)\in \mathcal{G}$ of diameter~$D$, and  every weight function $\omega:V\to \mathbb{N}$, computes a $k$-tuple of vertices $(v_1,\dots,v_k)$ such that 
\[
G\models\varphi(v_1,\dots,v_k), \;\; \text{and} \;\; 
\sum_{i=1}^k\omega(v_i) \; \mbox{is maximum,}
\]
in $\cO(D+\log n)$ rounds  under the \CONGEST model. (The same holds if replacing maximum by minimum.) If no tuples $(v_1,\dots,v_k)$ exist such that $G\models\varphi(v_1,\dots,v_k)$, then all nodes reject. 
\end{theorem}

\Cref{thm:thmFOexpOPT-informal} allows to find, for example, a monochromatic triangle of minimum  weight, or a multi-colored $k$-independent set of maximum weight, in a class $\mathcal{G}$ of $c$-colored graphs of bounded expansion.

\paragraph{Distributed Certification.}

Last but not least, our results extend to distributed certification. Roughly, a \emph{distributed certification scheme}~\cite{GoosS16,KormanKP10} for a graph property~$\varphi$ is a pair prover-verifier. The prover is a centralized, computationally unlimited oracle that assigns \emph{certificates} to the nodes. The verifier is a 1-round distributed algorithm whose role is to check whether the certificates assigned to the nodes form a proof that the graph satisfies the property. A distributed certification scheme is correct if it satisfies: 
\begin{description}
    \item[Completeness:] if $G \models \varphi$, then the prover can assign certificates such that the verifier accepts at all nodes; 
    \item[Soundness:] if $G \not\models \varphi$, then, no matter the certificates assigned to the nodes, the verifier rejects in at least one node. 
\end{description}

The main complexity measure is the \emph{size} of the certificates assigned by the prover on legal instances. In recent years, a vibrant and rapidly evolving line of research has emerged around algorithmic meta-theorems for distributed certification.  
It began with the seminal result by Bousquet, Feuilloley and Pierron \cite{BousquetFP21} showing that any MSO property can be certified with \(O(\log n)\)-bit certificates on graphs of bounded treedepth.  
This was then extended in \cite{FraigniaudMRT24} to bounded treewidth graphs, where all MSO properties admit \(O(\log^2 n)\)-bit certificates.  
Further extensions covered dense graph classes of bounded clique-width for MSO\(_1\) properties \cite{FraigniaudM0RT23}.  
More recently, Baterisna and Chang \cite{baterisna2025optimal} achieved \(O(\log n)\)-bit certificates on graphs of bounded pathwidth for MSO\(_2\) properties, and Cook, Kim and Masařík \cite{cook2025tight} proved that on graphs of bounded treewidth, every MSO\(_2\) property admits \(O(\log n)\)-bit certification.  

Our work contributes to this active line of research by significantly broadening its scope.  
We show that every first-order (FO) property can be certified with \(O(\log n)\)-bit certificates on all graphs of bounded expansion. As in  the previous theorems, the  theorem below extends to labeled graphs (cf. Theorem~\ref{thm:thmFOexpCER}).

\begin{theorem}
\label{thm:thmFOexpCER-informal} 
For every FO formula $\varphi$ on graphs, and for every class  of graphs $\mathcal{G}$ of bounded expansion, there exists a distributed certification scheme that, for every $n$-node network $G \in \mathcal{G}$, certifies $G \models \varphi$ using certificates on $\cO(\log n)$ bits.
\end{theorem}

Note that, with $O(\log n)$-bit certificates, the verifier can be implemented in a single round of communication in the \CONGEST model, as desired. Note also that our certification scheme is actually a \emph{proof-labeling scheme}~\cite{KormanKP10}, which is a more demanding form of distributed certification than, e.g., \emph{locally checkable proofs}~\cite{GoosS16}.
As a direct consequence of  \Cref{thm:thmFOexpCER-informal}, certifying $C_4$-freeness can be done with $\cO(\log n)$-bit certificates in graphs of bounded expansion. In contrast, certifying $C_4$-freeness requires certificates on $\tilde{\Omega}(\sqrt{n})$ bits in general graphs, using the reduction from~\cite{DruckerKO13}. More generally,  
\cite{FeuilloleyBP22} shows that every MSO formula can be \emph{certified} with $O(\log n)$-bit certificates in graphs of bounded treedepth. \Cref{thm:thmFOexpCER-informal} extends this result. Indeed, while  \Cref{thm:thmFOexpCER-informal} is for FO formulas, and the result of \cite{FeuilloleyBP22} is for MSO, both logics have identical expressive power on graphs of bounded treedepth. 
Last but not least, \Cref{thm:thmFOexpCER-informal} is the best that can be achieved. Indeed, we show that it does not extend to MSO, as certifying non-3-colorability requires certificates of $\tilde{\Omega}(\sqrt{n})$ bits in planar graphs (see \Cref{thm:lowerboundMSOinPlanar}), and even certificates of $\tilde{\Theta}(n)$ bits in bounded degree graphs (see \Cref{thm:lowerboundMSOinEXP}).

\paragraph{Remark.} 

The open problem by Ne\v{s}et\v{r}il and Ossona de Mendez was actually stated for the \BCONGEST model, i.e., a weaker version of the \CONGEST\/ model in which, at every round, each node $v$ must send the \emph{same} $\cO(\log n)$-bit message  to all its neighbors (see, e.g., ~\cite{DruckerKO13,korhonen_et_al:LIPIcs:2018:8625}). However, 
(1)~the problem is not necessarily easier to solve in the general version of the \CONGEST model, and 
(2)~we solve it anyway by the affirmative using an algorithm designed for the \BCONGEST model, and the complexity of our algorithm is expressed under the latter. 
For the sake of simplifying the notations, we simply refer to \CONGEST in the statements of our theorems, but all our upper bound results (except \Cref{thm:thmFOexpOPT-informal}) hold under the \BCONGEST model too\footnote{The main reason why \Cref{thm:thmFOexpOPT}  does not hold in \BCONGEST is that our algorithm involves an information dissemination phase going downward a tree, in which different messages are sent to different children. }.

\subsection{Our Techniques}

In this subsection, we provide a general overview of our techniques, postponing a more detailed description to \Cref{sec:high-level-description-proof}.
In a nutshell, our technical contributions cover several aspect of distributed computing and logic:
\begin{tcolorbox}[colback=blue!5!white, colframe=blue!40!black,
                  width=1\textwidth, boxrule=0.6pt,
                  arc=3pt, auto outer arc,
                  left=8pt, right=8pt, top=6pt, bottom=6pt]
\begin{itemize}[leftmargin=1.2em]
    \item Distributed implementation of the quantifier-elimination method.
    \item A new ``locality-preserving technique'' for quantifier-elimination.
    \item A new ``controlled quantifier-elimination technique'' that enables extending quantifier-elimination to both local and global counting, and solving these problems in a distributed manner.
\end{itemize}
\end{tcolorbox}

Let us first clarify why the approach by Ne\v{s}et\v{r}il and Ossona de Mendez~\cite{NesetrilM16} for marking all nodes $v$ satisfying $G\models \varphi(v)$ fails beyond the case of local \emph{and} existential FO formula, and how we address the challenge of going further than this specific type of formulas. Given a local  \emph{existential} FO formula 
\[
\varphi=\exists x_1 \,\dots\,\exists x_p\, \psi(x_1,\dots,x_p),
\]
where $\psi$ is quantifier-free, the fact that the input graph $G$ belongs to a graph class $\mathcal{G}$ of expansion $f:\mathbb{N}\to\mathbb{N}$ can be exploited somehow ``directly'' for checking whether $G\models \varphi$. 
Indeed, this checking is achieved by:
\begin{enumerate}
    \item Computing a coloring $c:V(G)\to \{1,\dots,f(p)\}$ satisfying that, for every set $U\subseteq \{1,\dots,f(p)\}$ of $p$ colors, the subgraph $G[U]$ induced by the nodes of $G$ colored with colors in~$U$ has treedepth~$p$ (such a coloring can be computed in $\cO(\log n)$ rounds in \CONGEST~\cite{NesetrilM16}).

    \item Observing that  $G\models \varphi$ if and only if there exists such a set $U$ such that $G[U]\models \varphi$, it is sufficient to test all sets $U\subseteq \{1,\dots,f(p)\}$ of cardinality~$p$, and, for each of them, to check whether  $G[U]\models \varphi$. The latter can be performed in $\cO(1)$ rounds as $G[U]$ has bounded treedepth~\cite{FominFMRT24}.
\end{enumerate}
This approach works but under two conditions only. First, the formula $\varphi$ must be \emph{local}. Second, it must be \emph{existential}. Instead, in this paper, we address the general case of deciding whether $G\models \varphi$ where $\varphi$ is general, and in particular not necessarily local. Moreover, for local formulas, we address the case of marking all nodes $v$ satisfying $G\models \varphi(v)$ for general local formulas, not necessarily existential. To do so, our first contribution is a distributed algorithm for quantifier-elimination, which is challenging as it must run under the \CONGEST model, in which only adjacent nodes can exchange messages, and where the amount of information two adjacent nodes can exchange is limited.

\paragraph{Distributed quantifier-elimination.}

The goal is to transform $\varphi$ into an equivalent \emph{quantifier-free} formula $\widehat{\varphi}$. That is, $\widehat{\varphi}$ is quantifier-free, and it satisfies $G\models \varphi$ if and only if $\widehat{G}\models \widehat{\varphi}$ where $\widehat{G}$ is the same graph as $G$ but with nodes labeled with labels taken from a finite set. One first contribution of this paper is to prove that quantifier elimination can be performed in $\cO(D)$ rounds in graphs of bounded expansion and diameter~$D$, using the framework developed in~\cite{PilipczukST18} for parameterized circuit complexity. Computing $\widehat{\varphi}$ can be done by every node without communication. The challenge is to assign the new labels to the nodes. We show that this is doable in $\cO(D)$ rounds in \CONGEST. 

In essence, the reassignment of labels is performed by going from graphs of bounded expansion to graphs of bounded treedepth (thanks to the coloring~$c:V(G)\to \{1,\dots,f(p)\}$ where $p$ is the number of variables of~$\varphi$), then from graphs of bounded treedepth to so-called \emph{skeletons} (i.e., collections of forests), and then from skeletons to rooted trees of bounded depth. One quantifier can be eliminated at this level, thanks to assigning new labels to each node. Handling the next quantifier requires to return back to the level of graphs with bounded treedepth. This back-and-forth process must be executed a constant number of times, for eliminating each quantifier one by one. The main difficulty here is to execute this protocol, which proceeds at different levels of abstraction, while insuring that all communications occur between neighbors in the graphs, and that the exchanged messages are of limited size. The exact number of executions depends  on the formula $\varphi$ only, and each execution is performed in $\cO(D)$ rounds in graphs of diameter~$D$. 

Overall, our distributed quantifier-elimination algorithm consumes $\cO(D+\log n)$ rounds in total, where $\cO(\log n)$ rounds are used for coloring the graphs with $f(p)$ colors, and $\cO(D)$ rounds are consumed for eliminating the quantifiers, and eventually evaluating the formula on the resulting labeled graph. 

\paragraph{Local FO formulas.}

For \emph{local} FO formula $\varphi(x)$ with one free variable, marking all nodes $v$ of the (possibly labeled) graph satisfying the formula, i.e., for which  $G \models \varphi(v)$ holds, is also performed via quantifier elimination. However, quantifier elimination faces an obstacle when it comes to eliminating universal quantifiers in subformulas of the form $\forall y\, \psi(x,y,z_1,\dots,z_q)$. A basic example is marking all nodes having a twin. 
Indeed, quantifier elimination proceeds in a standard way by rewriting the formula using only existential quantifiers, and negations. E.g., the formula is rewritten as  $\neg \exists y\, \neg \psi(x,y,z_1,\dots,z_q)$. As a result, locality is lost. For instance, 
\[
\textsc{Twin}(x) = \exists y\, \forall z\, 
\big(\mathsf{adj}(x,z) \Leftrightarrow \mathsf{adj}(y,z)\big)
\]
 can be rewritten as $\textsc{Twin}(x) = \exists y\, \neg \textsc{NoTwins}(x,y)$, with 
\[
\textsc{NoTwins}(x,y) =
\exists z\, \big[(\mathsf{adj}(x,z)\wedge \neg \mathsf{adj}(y,z))
\vee (\neg \mathsf{adj}(x,z)\wedge \mathsf{adj}(y,z))\big].
\]
The  formula $\textsc{NoTwins}(x,y)$ is however not local around $x$, even with $x$ fixed, as there may exist $y$ and $z$ at arbitrarily large distance from $x$ such that $\textsc{NoTwins}(x,y)$ holds. As a consequence, performing quantifier elimination blindly incurs  an overhead of $\Omega(D)$ rounds. This is not an issue for general formulas, for which a round-complexity $\Omega(D)$ cannot be overcame, but this is a severe problem for local formulas. 

The main challenge here is to show how to perform quantifier-elimination in a distributed manner while preserving locality. To prevent dependencies between distant vertices introduced during quantifier elimination, we adopt a redundant \emph{relativization} of quantifiers. Instead of ranging over the entire structure, each quantified variable is evaluated only within a bounded region. This requires to generalize locality with respect to more than one variable, and to revisit the entire quantifier elimination process. We  show that our \emph{local} quantifier-elimination process can be implemented in $\cO(\log n)$ rounds in \CONGEST, which allows us to eventually solve the open problem by Ne\v{s}et\v{r}il and Ossona de Mendez (cf. Question~\ref{prob:NM}).

\paragraph{Controlled quantifier-elimination for counting.}

Quantifier elimination is not sufficient in itself, even in its local form, when it comes to solving counting problems. Let us explain why. In the context of counting, quantifier elimination is applied to a formula $\varphi(\overline{x})$ that involves several free variables, $\overline{x}=(x_1,\dots,x_k)$, $k\geq 1$.
This produces an equivalent quantifier-free formula $\widehat{\varphi}(\overline{x})$ defined over a \emph{skeleton} structure (i.e., a collection of rooted forests) that captures all bounded–distance relations among the vertices of~$\overline{x}$. Indeed, quantifier-elimination actually occurs at the level of the skeleton.
However, counting operates at the level of labeled graphs. Therefore, \(\widehat\varphi\) must be translated back to the world of graphs, which requires to overcome many obstacles.  In particular, formulas in skeletons use predicates such as $\lca_i(x,y)$, expressing ``the common ancestor of $x$ and $y$ is at depth~$i$'' (of some tree in a forest). Reformulating such predicates with adjacency and equality predicates can be done thanks to an existential formula expressing, e.g., ``there exist $w_1, \dots, w_q$ forming a path between $x$ and $y$ such that one of the $w$’s is at depth~$i$''. The issue is that when one performs counting in the resulting formula \(\exists \overline w\, \psi(\overline x, \overline w)\), all the variables are involved, including the newly introduced variable $\overline{w}=(w_1, \dots, w_q)$. If these $w_1, \dots, w_q$ were unique, we would be fine, but this is not necessarily the case. Controlling uniqueness is one of the  challenges that we had to address. 

For limiting the effect of new variables introduced in the process of quantifier-elimination, we design a new mechanism referred to as \emph{controlled} quantifier-elimination. 
The main idea is to use the uniqueness of a node's parent in each tree of the skeleton to enforce the uniqueness of each \(w_i\) (given fixed \(x_1,\dots,x_k\)). To that end, we encode the structure of the skeleton into additional labels for each node of the graph. Then, we rewrite our existential formula \(\exists \overline w\, \psi(\overline x, \overline w)\) to make sure that each \(w_i\) is the parent of some \(w_j\), $j\neq i$, or of one of the \(x_j\)s in some tree of the skeleton. This guarantees that, for each tuple \(\overline x\) satisfying \(\exists \overline w\, \psi(\overline x, \overline w)\), there exists a \emph{unique} tuple \(\overline w\) such that \(\psi(\overline x, \overline w)\) holds. 

Note that, even for quantifier-free formulas, another important issue must be addressed. When testing all subset $U$ of colors, i.e., all subgraphs $G[U]$ of bounded treedepth, the same tuple $\overline v$ may appear in several subsets $U$ as it may not be the case that the vertices $v_1,\dots,v_k$ instantiating the variables $x_1,\dots,x_k$ are colored with distinct colors. This causes over counting. For avoiding over counting, we use the inclusion–exclusion principle over the family of color sets.  Finally, since each subgraph $G[U]$ has bounded treedepth, the existential formula can be evaluated by performing dynamic programming over a rooted forest decomposition of bounded depth, which takes $\cO(1)$ rounds~\cite{FominFMRT24}. 

\paragraph{Optimization and Certification.}

For optimization problems, we follow the same guideline as the one sketched above, where every vertex $v$ has a weight $w(v)$ provided as  input.  
Dynamic programming can  be used to compute the optimal value for every subset $U$ of colors, and the global optimum is the maximum (or minimum) over all sets $U$. For each $U$, the maximum (or minimum) value can be computed by aggregating the values of all connected components (as $G[U]$ may not be connected) in $\cO(D)$ rounds thanks to a breadth-first search (BFS) tree.

For certification we let the prover provide each vertex with the transcript of the execution of quantifier elimination as certificate. This includes, among other parameters, colors, forest membership, labels, etc.  
The verifier checks consistency of the transcript using  $\cO(\log n)$-bit certificates for bounded-treedepth graphs~\cite{BousquetFP21}, and it verifies any residual global aggregation along a certified BFS tree~\cite{KormanKP10}.

\subsection{Related Work}\label{sect:relatedwork}

The results of our paper lie at the intersection of two research directions.  
The first is the study of  distributed computing and distributed certification schemes on sparse networks, such as planar graphs, or graphs excluding a fixed minor, for fundamental graph optimization problems.  
The second is the development of meta-theorems for sparse graph classes in centralized computation.  

\subsubsection{Distributed Computing on Sparse Networks}

The class of graphs with bounded expansion was introduced by Ne\v{s}et\v{r}il and Ossona de Mendez as part of the sparsity theory they developed~\cite{NesetrilOdM12}. 
One of the major research areas in distributed computing concerns the study of classical combinatorial problems on sparse networks, including planar graphs, bounded-degree graphs, graphs of bounded treewidth, and graphs excluding a fixed minor. In this context, a central goal is to design distributed algorithms that exploit structural sparsity to achieve improved efficiency---compared to general networks---particularly in terms of the number of communication rounds. Typical problems studied in this context include coloring, matching, dominating sets, cuts, spanning trees, and network decompositions. Research directions focus on understanding how sparsity affects the complexity of these problems in standard distributed models such as \textsf{LOCAL} and  \CONGEST, and on designing algorithms with optimal or near-optimal round complexity. This line of work often builds on structural properties of sparse networks, revealing deep and fruitful connections between structural graph theory and distributed computation.

Examples of work in this direction include distributed algorithms for maximal matchings, maximal independent sets (MIS) and graph colorings in bounded-degree graphs 
\cite{panconesi2001simple}; depth-first search (DFS) in planar graphs \cite{ghaffari:2017}; minimum spanning tree (MST) and min-cut in planar, bounded genus, and $H$-minor-free graphs \cite{GhaffariH16a,GhaffariH16b,GhaffariH21,haeupler2016low,HaeuplerIZ16,HaeuplerLZ18}; minimum dominating set on planar and  $H$-minor-free graphs \cite{amiri2019distributed,bonamy2025local, czygrinow2008fast,czygrinow2018distributed,lenzen2013distributed},  and MIS for graphs excluding a fixed minor \cite{chang2023efficient}.

The $\tilde{\mathcal{O}}(D)$ round complexity  we achieved in  \Cref{thm:thmFOexpCONG-informal}, is considered a “gold standard” for algorithms on sparse graphs in the \CONGEST model. This complexity has been obtained, for instance, for MST on planar graphs and, more generally, on $H$-minor-free graphs~\cite{GhaffariH16a,GhaffariH16b,GhaffariH21,haeupler2016low,HaeuplerIZ16,HaeuplerLZ18}.  
Algorithms with $\tilde{\mathcal{O}}(D)$ complexity in the \CONGEST model have also been developed for bounded-treewidth graphs for a variety of problems, including matching, directed single-source shortest paths, bipartite unweighted maximum matching, and girth~\cite{izumi2022fully}.

We are not aware of any meta-theorem in the \CONGEST model that combines both logical and combinatorial constraints, with one exception. The paper~\cite{FominFMRT24} shows that every MSO formula can be decided in a constant number of rounds on graphs of bounded treedepth. The results of~\cite{FominFMRT24} are incomparable to \Cref{thm:thmFOlocal-informal,thm:thmFOexpCONG-informal}. Indeed, \Cref{thm:thmFOlocal-informal,thm:thmFOexpCONG-informal} apply to graph classes that are significantly more general than graphs of bounded treedepth, whereas the number of rounds in the algorithm of~\cite{FominFMRT24} is completely independent of the number of vertices. We incorporate some of the methods developed in~\cite{FominFMRT24} into our work.

Regarding distributed algorithms on graphs of bounded expansion, the foundational work in this area is due to Ne\v{s}et\v{r}il and Ossona de Mendez~\cite{NesetrilM16}. In the \CONGEST model, they presented a distributed algorithm that computes a low-treedepth decomposition of graphs of bounded expansion in $O(\log n)$ rounds. We use this construction as a subroutine in our algorithms. In the same paper, Ne\v{s}et\v{r}il and Ossona de Mendez also used this decomposition to design an algorithm for checking whether a graph of bounded expansion contains a given \emph{connected} subgraph $H$ of constant size, also in $O(\log n)$ rounds. \Cref{thm:thmFOlocal-informal} significantly extends this result.

Approximation algorithms for the minimum dominating set problem on graphs of bounded expansion in the \textsf{LOCAL} model were developed in \cite{amiri2019distributed,heydt2025distributed,kublenz2021constant}.
In \cite{SiebertzV19}, Siebertz and Vigny study the complexity of the model-checking problem for first-order logic on graphs of bounded expansion, in the \textsf{CONGESTED CLIQUE} model.
 
Distributed certification is a very active area of research; see \cite{FeuilloleyF16} for a survey. 
Several meta-theorems are known in this setting.
A significant result in the field was achieved in  \cite{FeuilloleyBP22}, which proved that  any MSO property can be locally \emph{certified} on graphs of bounded treedepth using only a logarithmic number of bits per node, a result widely regarded as  a gold standard in certification. This theorem has numerous implications; for more details we refer to \cite{FeuilloleyBP22}. In particular, the FO property of $C_4$-freeness and the MSO property of non-3-colorability, which in general graphs require polynomial-size certificates, can be certified using just $\mathcal{O}(\log n)$ bits per node on bounded treedepth graphs. The result by Bousquet et al. has since been extended to broader 
graph classes, including graphs excluding a small minor \cite{BousquetFP21}, graphs of bounded \emph{treewidth} \cite{baterisna2025optimal,cook2025tight,FraigniaudMRT24}, and graphs of bounded \emph{cliquewidth} \cite{FraigniaudM0RT23}.
 
\subsubsection{Centralized Meta-Theorems}

Meta-theorems, according to Grohe and Kreutzer \cite{GroheK09}, are results 
  stating that certain  model-checking problems for formulas of a given logic are tractable (in some sense) on specific classes of graphs.
 One of the most celebrated examples of a meta-theorem is Courcelle's theorem, which asserts that graph properties definable in MSO are decidable in linear time on graphs of bounded treewidth \cite{Courcelle90}.
There is a long line of meta-theorems for FO logic on sparse graph classes in the centralized model of computation. The starting point is Seese’s theorem  \cite{Seese96}, which establishes the linear-time decidability of FO properties on graphs of bounded degree. Fricke and Grohe \cite{FrickG01} presented linear-time algorithms for planar graphs and all apex-minor-free graph classes, as well as $\cO(n^{1+\varepsilon})$ algorithms for graphs of bounded local treewidth. Flum and Grohe \cite{FlumG01} proved that deciding FO properties is fixed-parameter tractable on graph classes with excluded minors.  The next step, to classes of graphs locally excluding a minor was achieved  by  Dawar, Grohe, and Kreutzer \cite{DawarGK07}.

The sparsity theory of 
\cite{NesetrilOdM12} opened new horizons for model checking on sparse graphs. In \cite{DvorakKT13} it was proved that 
 FO properties can be decided in linear time on graph classes of bounded expansion. 
  Further extensions of these results are known for nowhere-dense graphs ~\cite{GroheKS14} and graphs of bounded twinwidth \cite{Bonnet0TW20twinw}.
 See \cite{Kreutzer11algo}  and \cite{Grohe07logi} for surveys.

From a technical perspective, the most important aspect for our work is the quantifier elimination procedure introduced in \cite{DvorakKT13} and \cite{GroheK09}, as well as in \cite{PilipczukST18}. Similar to these works, we apply the idea of modifying the structure by adding new unary relations and functions without changing the Gaifman graph.

\section{High Level Description of the Proofs}
\label{sec:high-level-description-proof}

Let us first focus on sketching the proof of Theorem~\ref{thm:thmFOlocal-informal}. 

\subsection{Deciding Local FO Formulas Locally}

We first recall structural tools that allow us to work with bounded-expansion graphs in the distributed setting.  
These lead naturally to \emph{low-treedepth colorings}, which will be the foundation for the construction of skeletons and for the distributed quantifier-elimination method developed later. 

\paragraph{Low-treedepth colorings.}

We recall the notion of \emph{low-treedepth colorings} introduced by Nešetřil and Ossona de Mendez. 
Given a graph $G$ and an integer $p \in \mathbb{N}$, a proper coloring $c : V(G) \rightarrow \mathbb{N}$ is a \emph{$p$-treedepth coloring} if, for every set $U \subseteq \mathbb{N}$ with $|U| \leq p$, the subgraph of $G$ induced by the vertices $\{u \in V(G) : c(u) \in U\}$ has tree-depth at most $|U|$. 
The minimum number of colors required for such a coloring is denoted by $\chi_p(G)$. 
A graph class $\mathcal{G}$ has \emph{bounded expansion} if there exists a function $f:\mathbb{N}\to\mathbb{N}$ such that, for every $G \in \mathcal{G}$ and every $p \in \mathbb{N}$, $\chi_p(G) \leq f(p)$, i.e., every $G \in \mathcal{G}$ admits a $p$-treedepth coloring using at most $f(p)$ colors.  

\medskip

For the rest of the section, let us fix a graph class $\mathcal{G}$ of  expansion bounded by~$f:\mathbb{N}\to\mathbb{N}$. We denote by $\Col(p)$ the family of sets $U\subset [1,f(p)]$  with cardinality at most~$p$.

\paragraph{Subgraph detection via treedepth colorings.}

Let $H$ be a connected graph on $p$ vertices, and let us first assume we are given a $p$-treedepth coloring of a graph $G$.  
To detect whether $G$ contains $H$ as a subgraph, we can then check, for every $U \in \Col(p)$, whether $H$ appears in $G[U]$, i.e., in the subgraph of $G$ induced by the nodes with colors in~$U$.  
If $G$ contains a copy of $H$, then there must exist a color set $U$ such that $G[U]$ contains~$H$.  
Moreover, as $H$ is connected, $H$~lies in a single component of $G[U]$.  
Since $G[U]$ has bounded treedepth, this check can be done efficiently in the \CONGEST model~\cite{FominFMRT24}.
Nešetřil and Ossona de Mendez~\cite{NesetrilM16} showed that, for every bounded-expansion class $\mathcal{G}$ with expansion function $f$, there is a \CONGEST algorithm that computes a $p$-treedepth coloring with $f(p)$ colors, in $O(\log n)$ rounds.  
It follows that there exists a \CONGEST algorithm that marks every vertex belonging to a copy of $H$ in~$G$, in $\mathcal{O}(\log n)$ rounds~\cite{NesetrilM16}.  

The reasoning above extends to every property expressible by an \emph{existential} first-order formula ($\exists$FO) that is local, i.e., a formula of the form
\[
\varphi(x) = \exists y_1\, \dots \, \exists y_k \; \psi(x, y_1, \dots, y_k),
\]
where (1)~$\psi$ is quantifier-free with $k{+}1$ free variables, and (2)~the satisfiability of $\varphi(u)$ depends only on the structure of a ball of constant radius around vertex~$u$ in the graph.  
As in subgraph detection, a solution is determined by $k{+}1$ vertices, and thus $k{+}1$ colors.  
Given a $(k{+}1)$-treedepth coloring, it suffices to search within connected subgraphs of bounded treedepth.  
It is known~\cite{FominFMRT24} that any FO property can be decided in $O(1)$ rounds on graphs of bounded treedepth in \CONGEST. Therefore, as for subgraph detection,  combining~\cite{NesetrilM16} and~\cite{FominFMRT24} yields an $O(\log n)$-round \CONGEST algorithm for any local existential FO property.

\paragraph{Beyond existential FO properties.}

For general local FO formulas, with both existential and universal quantifiers, the situation is more intricate.  
Indeed, the solution may no longer be determined by a fixed set of vertices, and thus restricting the search to a bounded-treedepth subgraph may not be sufficient. Let us revisit the aforementioned twin problem, with free variable $x$, asking whether there is a node $y$ with exactly the same neighborhood as~$x$. That is, let us consider the formula 
\[
\textsc{Twin}(x) = \exists y \; \forall z \; \big(\mathsf{adj}(x,z) \Leftrightarrow \mathsf{adj}(y,z)\big), 
\]
where it is assumed that $x,y$, and $z$ refer to different vertices (for the sake of clarifying the presentation, we omit these constraints in the definition of $\textsc{Twin}(x)$). 
Since the neighborhood of a vertex may have $\Omega(n)$ nodes, not all neighbors $z$ of $x$ and $y$ may lie in a region of bounded treedepth resulting from a low-treedepth coloring.  
\textsc{Twin} is however local as it suffices to inspect the ball of radius at most~$3$ around vertex $u$ to decide $\textsc{Twin}(u)$. Yet, the logical structure of \textsc{Twin}, which includes the variable $z$ quantified by a universal quantifier, prevents a direct application of the technique applicable to local existential FO formula.
To proceed further, let us rewrite $\textsc{Twin}$ using De Morgan’s laws:
\begin{align*}
  \textsc{Twin}(u) &= \exists v \; \neg \exists w \; \neg \big(\mathsf{adj}(u,w) \Leftrightarrow \mathsf{adj}(v,w)\big)\\ 
  &= \exists v \; \neg \exists w \;
\big[\big(\mathsf{adj}(u,w) \wedge \neg \mathsf{adj}(v,w)\big) \vee \big(\neg \mathsf{adj}(u,w) \wedge \mathsf{adj}(v,w)\big)\big].
\end{align*}
Note that the negation between the first and second quantifier still prevents us from using the reasoning for local $\exists$FO.   
Let us thus focus on the inner existential part of $\textsc{Twin}(u)$, denoted by
\[
\textsc{NoTwins}(u,v) = \exists w \;
\big[\big(\mathsf{adj}(u,w) \wedge \neg \mathsf{adj}(v,w)\big) \vee \big(\neg \mathsf{adj}(u,w) \wedge \mathsf{adj}(v,w)\big)\big].
\]
The trick is to apply \emph{quantifier elimination}, that is, removing the existential quantifier from $\textsc{NoTwins}$, as explained hereafter. 

\paragraph{Quantifier elimination.}

From this point on, we consider \emph{labeled graphs} $(G,\ell)$, where each vertex $u$ carries a set of labels $\ell(u)$ where the labels are taken from a set $\Lambda$ of constant size. We now consider labeled graphs because quantifier-elimination induces labeling the nodes.  
First-order formulas over $(G,\ell)$ use adjacency, equality, and \emph{unary label predicates} $\lab_\lambda(x)$, which hold whenever $\lambda \in \ell(x)$. 
A first-order formula over unlabeled graphs can thus be regarded as a particular case without unary predicates.

Quantifier elimination removes quantifiers while preserving equivalence within the class of structures under consideration. 
This technique is at the core of recent advances in algorithmic model checking on sparse graphs~\cite{DvorakKT13,GroheK09,PilipczukST18}.  
Let $\varphi$ be an existential FO formula with free variables $x_1, \dots, x_k$, and quantified variables $y_1, \dots, y_t$, i.e.,
\[
\varphi(x_1,\dots,x_k) = \exists y_1\, \dots\, \exists y_t \, \psi(x_1, \dots, x_k, y_1, \dots, y_t),
\]
where $\psi$ is quantifier-free.  
Quantifier elimination maps $\varphi$ into a quantifier-free FO formula $\widehat{\varphi}$ over an \emph{enriched structure}, which, in our setting, will later correspond to the so-called \emph{skeletons}~\cite{PilipczukST18}.

\paragraph{Skeletons.}

Given an integer~$p$, which corresponds to the size of the sets $U\in\Col(p)$, a $p$-skeleton is a collection $\{F_1, \dots, F_s\}$ of labeled rooted forests, each of depth at most~$d$, where both $s$ and $d$ depend only on~$p$.  
All forests share the same vertex set~$V$, corresponding to that of the underlying labeled graph~$(G,\ell)$.  
Each vertex $u\in V$ retains its set of labels $\ell(u)\subseteq\Lambda$, while the structural relations between vertices are captured by the following predicates:
\begin{itemize}[leftmargin=1.2em]
    \item \emph{Label predicates} $\lab_\lambda(x)$, already defined to hold whenever $\lambda\in \ell(x)$.
    \item \emph{Least-common-ancestor (lca) predicates} $\lca^i_j(x,y)$, which hold when $x$ and $y$ belong to the same tree of the $i$-th forest $F_i$, and their least common ancestor $\mathrm{lca}(x,y)$ lies at depth~$j$ in~$F_i$.  
    The predicate $\lca^i_{-1}(x,y)$ holds when $x$ and $y$ lie in distinct trees of~$F_i$.
\end{itemize}

The sets of label predicates, and $\lca$ predicates are both finite, depending only on $p$ and $|\Lambda|$, where $\Lambda$ is the set of original labels.  
An FO formula over $p$-skeletons uses \emph{only} these two kinds of predicates. In particular, it does \emph{not} use adjacency predicates.   
It was proved in~\cite{PilipczukST18} that, for any graph class $\mathcal{G}$ of bounded expansion, any existential formula 
\[
\varphi(x_1,\dots,x_k)= \exists y_1\, \dots\, \exists y_t \, \psi(x_1, \dots, x_k, y_1, \dots, y_t)
\]
and any $G \in \mathcal{G}$, there exist a \emph{quantifier-free} FO formula $\widehat{\varphi}$ on skeletons and a $(t{+}k)$-skeleton $\widehat{S}(G)$ such that, for every $k$-tuple $(v_1,\dots,v_k)$ of vertices of~$G$,
\[
(G,\ell)\models\varphi(v_1, \dots, v_k)
\;\;\iff\;\;
\widehat{S}(G) \models \widehat{\varphi}(v_1, \dots, v_k).
\]
A central result (\Cref{lem:qelim:remove-lca-general}) in our paper further shows that, for local existential FO formulas, the skeleton $\widehat{S}(G)$ can be constructed in a distributed manner in $\mathcal{O}(\log n)$ rounds in the \CONGEST\ model.

Before performing quantifier elimination, we first define a canonical $p$-skeleton $S(G)$ associated with each graph~$G$.  
The purpose of this construction is to provide a representation of~$G$ by a bounded number of forests of bounded depth.
This intermediate structure serves as a bridge between graphs of bounded expansion and trees, allowing us to apply the known quantifier-elimination techniques that operate on tree-like structures.  
Once quantifiers are eliminated, we obtain a new skeleton $\widehat{S}(G)$ that captures, in a quantifier-free way, the properties expressed in the original formula.

\paragraph{Construction of the Skeleton.}

For $G \in \mathcal{G}$ and $p \in \mathbb{N}$, the $p$-skeleton $S(G)$ is built as follows.  
Let us fix a $p$-treedepth coloring $c:V(G)\to [f(p)]=\{1,\dots,f(p)\}$ of $G$.  
For every $U \in \Col(p)$, we compute a spanning forest $F^U$ of $G[U]$ by performing a depth-first search (DFS) traversal in each connected component of $G[U]$, starting from the node of minimum identifier in that component.
Since $G[U]$ has treedepth at most $p$, the trees of $F^U$ have depth at most $d=2^{p}$ (see~\Cref{prop:depth}).  
The forests of $S(G)$ are those in the set $\{F^U: U \in \Col(p)\}$. New labels are assigned as follows.   
 
\begin{itemize}
    \item For every node~$u$, its color label $c(u)$ is added to its new set of labels $\ell(u)$, that is, the color $c(u)$ of node~$u$ is included in its collection of labels~$\ell(u)$.   

    \item For every $U \in \Col(p)$, if $u$ is at depth $i \in [d]$ in $F^U$, then the label $(\depth,i,U)$, denoted by $\depth_i^U$ for short, is added to $\ell(u)$, that is, every vertex colored by $U$ has its depth in $F^U$ included in its collection of labels.

    \item $F^U$ is a collection of DFS trees. Therefore, every edge in $E(G[U]) $ connects a vertex to one of its ancestors in~$F^U$. For every such edge $\{u,v\}$ with $u$ a descendant of~$v$, and $v$ at depth~$i$, the label $(\level,i,U)$, abbreviated into $\level_i^U$, is added to~$\ell(u)$. That is, for each ancestor \(v\) of \(u\) that is adjacent to \(u\) in $G[U]$, \(\ell(u)\) contains \(\level_i^U\), where \(i\) is the depth of \(v\) in $F^U$.
\end{itemize}

\noindent\textit{Remark.} It follows from the last item that every edge of $G$ is represented in $S(G)$ by a label attached to one of its endpoints.  
Indeed, let $e=\{u,v\}$ be an edge of $G$, and let $U\in\Col(p)$ be a color set containing the colors of both endpoints of~$e$.  
Assume that $u$ lies at depth~$j$ and $v$ at depth~$i<j$ in the forest $F^U$.  
Then the edge $e$ is encoded as the label $\level_i^U$ attached to the deeper vertex~$u$.

\medskip

Any FO formula $\varphi$ over labeled graphs can then be translated into an equivalent  FO formula $\widetilde{\varphi}$ over $p$-skeletons by rewriting adjacency and equality predicates as Boolean combinations of $\lca$ and label predicates. For instance,
\begin{align*}\adj(x,y) = \bigvee_{U\in \Col(p)}\bigvee_{0 \leq i < j \leq d} \Big[&\left(\lab_{\depth_i^U}(x) \wedge \lab_{\depth_j^U}(y) \wedge \lca_i(x,y) \wedge \lab_{\level^U_i}(y)  \right) \\
& \vee \left(\lab_{\depth_i^U}(y) \wedge \lab_{\depth_j^U}(x) \wedge \lca_i(x,y) \wedge \lab_{\level^U_i}(x)  \right) \Big].
\end{align*}
The equality predicates can be rewritten in a similar way.  Therefore, for every assignment $\overline{v}$ to the free variables,
\[
(G,\ell) \models \varphi(\overline{v}) \iff S(G) \models \widetilde{\varphi}(\overline{v}),
\]
That is, $\varphi$ and $\widetilde{\varphi}$ are semantically equivalent. One refers to labeled graphs while the other refers to skeletons, i.e.,   
$\widetilde{\varphi}$ is the syntactic rewriting of the  formula $\varphi$ into the language of skeletons. We explain next how to remove quantifiers from $\widetilde{\varphi}$ to get the desired quantifier-free formula~$\widehat{\varphi}$. Quantifier elimination is performed on skeletons.

\paragraph{Quantifier elimination in skeletons.}

Recall that we start from an existential  formula
\[
\varphi(x_1,\dots,x_k) = \exists y_1\, \dots\, \exists y_t \, \psi(x_1, \dots, x_k, y_1, \dots, y_t),
\]
where $\psi$ is quantifier-free. We have that, for every assignment $\overline{v} = v_1,\dots,v_k$ and $\overline{u} = u_1,\dots,u_t$ of the variables such that  $(G,\ell) \models \psi(\overline{v},\overline{u})$ holds, there exists $U \in \Col(p)$ where  
$p = t + k$ such that the colors of all $u_i$ and $v_j$ are in~$U$.  
Hence, it follows from the previous discussion that 
\begin{align*}
    (G,\ell) \models \psi(\overline{x},\overline{y}) 
& \iff \exists U \in \Col(p)\,:\, (G[U], \ell) \models \psi(\overline{x},\overline{y}) \\
& \iff \exists U \in \Col(p) \,:\, S(F^U) \models  \widetilde{\psi}^U(\overline{x},\overline{y}),
\end{align*}
where $\widetilde{\psi}^U$ refers to nodes colored with  colors in $U$ only, and $S(F^U)$ denotes the skeleton $S(G)$ restricted to $F^U$.  
Thus the elimination of existential quantifiers from formulas on skeletons reduces to the elimination of these quantifiers from formulas on forests of bounded depth.

\paragraph{Quantifier elimination in bounded-depth forests.}

We eliminate existential quantifiers one at a time. Concretely, that is, we consider a subformula of the form
\[
\exists y \; \psi(x_1,\dots,x_k,y),
\]
and we remove the quantifier for this single variable $y$ before proceeding to the next one. The quantifier-elimination procedure on bounded-depth forests relies on the notion of \emph{types} (see \Cref{def:types}). For a rough definiton of what we mean by type, recall that, in  forests of bounded depth, the structural relations between vertices can be expressed using the predicates $\lca_i(x,y)$ indicating that the lowest common ancestor of $x$ and $y$ lies at depth~$i$ (in what follows we omit the super-indices on the lca predicates). On a forest of depth~$d$, a first-order formula therefore depends only on (i) the labels that are held by each vertex, and (ii) how the vertices of the tuple $\overline{x}=(x_1,\dots,x_k)$ are positioned with respect to one another in this ancestor hierarchy. These local configurations are captured thanks to different types. For a tuple of variables $\overline{x}$, a type is determined by 
\[
\mbox{a \emph{label map}} \; \gamma:[k]\to 2^\Lambda, 
\;\mbox{and an \emph{ancestry map}}\; 
\delta:[k]\times[k]\to\{-1,0,\dots,d-1\}
\]
that jointly describe which labels each variable $x_i$ carries, and at which depths its ancestors stand. Each consistent pair $(\gamma,\delta)$ defines a conjunction $\zeta(\overline{x},y)$ of atomic predicates fixing both the label and ancestral structure of the tuple. Since there are finitely many such combinations, every quantifier-free formula over lca and label predicates can be rewritten as a finite disjunction of types. This presentation is referred to as the \emph{basic normal form} of the formula (see \Cref{prop:basicnormal}).
To eliminate an existential quantifier $\exists y\,\psi(\overline{x},y)$, we proceed type by type. For a fixed type $\zeta(\overline{x},y)$, we identify a variable $x_s$ closest to $y$, and we distinguish three cases. 

\begin{itemize}
\item If $y$ is an ancestor of $x_s$ at depth $t$, then we mark each vertex with an ancestor at depth $t$ and carrying the appropriate labels by a new unary predicate $\good$. We then replace the formula by the restriction of the type to $\overline{x}$ conjoined with the predicate $\lab_{\good}(x_s)$. 

\item If $y$ lies in the subtree of a child of $x_s$, then we mark all nodes whose subtrees contain a potential \emph{witness}\footnote{In this context, a \emph{witness} for a tuple $\overline{x}$ in a formula $\exists y\,\psi(\overline{x},y)$ is simply a vertex $v$ such that $\psi(\overline{x},v)$ holds. That is, $v$ is a node whose existence makes the existential statement true. During quantifier elimination we never select this vertex explicitly. Instead, we locally encode, through additional labels, whether such a witness exists in the relevant part of the forest.}, and attach a counter with bounded values indicating how many such subtrees exist, reducing the existential to a finite disjunction over child patterns. 

\item If $y$ belongs to a different tree than all $\overline{x}$, then we label as \emph{active} the roots of all trees that contain a suitable witness, and we propagate this information across components by communicating along the edges of~$G$.
\end{itemize}
In all cases, we obtain a quantifier-free formula $\widehat{\psi}(\overline{x})$ with a few additional  labels, unary predicates, and bounded counters. 

\paragraph{Local vs.\ non-local formulas.}
Before going further, let us clarify how the locality of the formula affects the running time of quantifier elimination. Recall that our algorithm runs in $O(\log n)$ rounds for local formulas, and in $O(D+\log n)$ rounds for general ones, where $D$ denotes the diameter of the communication graph. The following intuition explains this difference.

If $\psi$ is local, then all occurrences of its  variables lie within the same tree, at a bounded distance from one another. Hence, whether a node $v$ should receive a new label (as part of eliminating the quantifier in~$\psi$) can be decided by $v$ using only local information from its own tree, whose depth is bounded. In the distributed setting, both the computation of these new labels and the evaluation of the resulting quantifier-free formula take place entirely within each tree, requiring only $\cO(1)$ communication rounds in the \CONGEST\ model.

If $\psi$ is non-local, however, a witness for some tuple $\overline{x}$ may lie in a different tree. 
Then, the label assigned to a node $v$ after eliminating the quantifier may depend on information carried by \emph{active roots} located outside $v$’s component, potentially at distance $\Omega(D)$ in the original graph~$G$. 
Quantifier elimination for non-local formulas may thus require communication across components, along edges of~$G$, at distances up to~$D$, leading to a total complexity of $\cO(D+\log n)$ rounds.

\paragraph{From existential to general formulas.}
So far, we have described how to eliminate   quantifiers
from a formula of the form $\exists y_1 \dots \exists y_k\,\psi(x_1,\dots,x_k,y_1, \dots, y_k)$.
To handle arbitrary FO formulas, which may alternate existential and universal quantifiers, we follow a standard transformation: every universal quantifier is rewritten as the negation of an existential one. 
This allows us to apply the existential elimination procedure iteratively. 
The process can be illustrated with the search for twins. Recall that
\[
\textsc{Twin}(x)=\exists y\,\forall z\,(\mathsf{adj}(x,z)\Leftrightarrow \mathsf{adj}(y,z)).
\]
We first rewrite the universal quantifier as a negated existential:
\[
\textsc{Twin}(x)=\exists y\,\neg\exists z\,
\neg(\mathsf{adj}(x,z)\Leftrightarrow \mathsf{adj}(y,z)).
\]
We then focus on the innermost existential subformula,
\[
\textsc{NoTwins}(x,y)
=\exists z\,\neg(\mathsf{adj}(x,z)\Leftrightarrow \mathsf{adj}(y,z)).
\]
At this point, we eliminate the quantifier $\exists z$ using the procedure described before, producing a quantifier-free formula
\[
\widehat{\textsc{NoTwins}}(x,y),
\]
which, through new labels, encodes whether there exists a vertex $z$ distinguishing $x$ and $y$ by adjacency.  
Our elimination procedure never identifies this witness explicitly. Instead, it records locally, via additional labels, whether such a witness exists within the  neighborhood of the relevant variable.
Substituting this expression into the original one yields the formula 
\[
\exists y\,\neg \widehat{\textsc{NoTwins}}(x,y).
\]
Before eliminating the next quantifier~$\exists y$, the formula $\neg\widehat{\textsc{NoTwins}}(x,y)$, which is expressed over $p$-skeletons, must be translated back into the standard structure of labeled graphs. 
This translation replaces each $\lca$-predicate by an equivalent existential formula using only adjacency, equality and label predicates (cf. \Cref{lem:qelim:remove-lca-general}). 
After this translation, we obtain again an existential formula $\psi'(x,y)$ over labeled graphs. 
Adding the outer quantifier $\exists y$ preserves this existential form, yielding again a formula of the type $\exists y\,\psi'(x,y)$.
We can therefore apply the same quantifier-elimination procedure once more.
The outcome of this second elimination is a quantifier-free formula with a single free variable~$x$, denoted by 
\[
\widehat{\textsc{Twin}}(x).
\]
The resulting formula is expressed as a Boolean combination of unary predicates, each corresponding to a label on~$x$ that records whether a suitable witness $y$ existed at a previous step.

This iterative process, alternating quantifier elimination and negation, extends the existential elimination method to arbitrary FO formulas. 
Each quantifier, regardless of its type,
is ultimately replaced by local labels and Boolean combinations of unary predicates. But, there is a big but: 

\medskip
\centerline{\it this procedure does not necessarily preserve locality!}
\medskip

\noindent Additional care is thus required for ensuring that each intermediate formula remains local. This is our second contribution, described next. 

\subsection{Locality-preserving quantifier elimination}

Standard quantifier elimination at the core of model checking on sparse graphs~\cite{DvorakKT13,GroheK09,PilipczukST18} may \emph{not preserve locality}. 
That is, even if the input formula is local, the resulting quantifier-free formula may refer to nodes that are arbitrarily far apart in the graph. 
Before illustrating this issue, let us first formalize what we mean by locality for formulas with \emph{several} free variables.

\paragraph{Local formulas and distance predicates.}

Let $r\ge 0$ be an integer.  
A formula $\psi(x,\bar y)$, with $\bar y=(y_1,\dots,y_k)$, is said to be \emph{$r$-local around~$x$} if the following two conditions hold:
\begin{enumerate}
\item[(i)] Whenever $(G,\ell)\models \psi(x,\bar y)$, all variables $\bar y$ are interpreted within the $r$-ball of~$x$, that is,
\[
(G,\ell)\models \psi(x,\bar y)\ \Longrightarrow\ \bar y\subseteq B_G(x,r),
\]
where $B_G(x,r)$ denotes the set of vertices at distance at most~$r$ from~$x$ in~$G$.
\item[(ii)] The truth of $\psi$ depends only on the labeled structure inside this ball:  
for any two labeled graphs $(G,\ell)$ and $(G',\ell')$, and vertices $x\in V(G)$, $x'\in V(G')$,  
if there exists a label-preserving isomorphism $h:B_G(x,r)\to B_{G'}(x',r)$ with $h(x)=x'$, then
\[
(G,\ell)\models\psi(x,\bar y)\iff (G',\ell')\models\psi(h(x),h(\bar y)).
\]
\end{enumerate}
In formulas with several free variables, the distinguished variable~$x$ is called the \emph{reference vertex}, and locality is always defined with respect to it.

In practice, enforcing $r$-locality syntactically is difficult.  
To make the dependence explicit, we use \emph{distance predicates} $\ball(y,x,r)$, meaning that $y$ is at distance at most~$r$ from~$x$. By extension, for a tuple of variables \(\overline y\), we write \(\ball(\overline y, x, r)\) for \(\bigwedge_{z\in\overline y} \ball(z,x,r)\).
During quantifier elimination, quantified variables are guarded by appropriate $\ball$ predicates, and these guards are kept redundantly so that locality is preserved even under negation.

\paragraph{Failure of standard elimination.}

The loss of locality can be seen on the \textsc{Twin} property:
\[
\textsc{Twin}(x) = \exists y\, \forall z\, 
\big(\mathsf{adj}(x,z) \Leftrightarrow \mathsf{adj}(y,z)\big),
\]
which can be rewritten as
\[
\textsc{NoTwins}(x,y) =
\exists z\, \big[(\mathsf{adj}(x,z)\wedge \neg \mathsf{adj}(y,z))
\vee (\neg \mathsf{adj}(x,z)\wedge \mathsf{adj}(y,z))\big],
\]
so that 
\[
\textsc{Twin}(x) = \exists y\, \neg \textsc{NoTwins}(x,y).
\]
The formula $\textsc{NoTwins}(x,y)$ is not local \emph{around $x$}: even with $x$ fixed, there may exist $y$ (and the witnessing $z$) at arbitrarily large distance from $x$. 

\paragraph{Relativizing quantifiers.}

To prevent dependencies between distant vertices introduced during quantifier elimination, we adopt a redundant \emph{relativization} of quantifiers.  
Instead of ranging over the entire structure, each quantified variable is evaluated only within the bounded region centered at the reference vertex, which is then translated to a tree of bounded depth.  
The relativization is expressed syntactically through distance predicates.

Formally, we define a \emph{local form} for local formulas by induction on quantifier depth, as follows. A formula \(\varphi(x, \overline y)\) is in \(r\)-local form if one of the following conditions holds: 
\begin{itemize}
    \item \(\varphi(x,\overline y)=\psi(x,\overline y) \wedge \ball(\overline y,x, r)\), where
    \(\psi\) is quantifier-free,
    
    \item \(\varphi(x, \overline y)=\exists z\; \psi(x, \overline y, z)\), where \(\psi\) is in \(r\)-local
    form, or
    
    \item \(\varphi(x,\overline y)=\ball(\overline y,x,r)
    \wedge \neg\exists z\; \ball((\overline y,z),x,r)
     \wedge \neg \psi(x,\overline y, z)\) where
    \(\psi\) is in \(r\)-local form.
\end{itemize}
Note how the universal quantifier requires a particular care, as it is handled as a negated existential quantifier. This formulation introduces negations, which do not preserve locality, as illustrated by the \textsc{Twin} example. This is why we introduce \emph{redundant} \(\ball\) predicates that guarantee syntactically the locality property.

Given a labeled graph $(G,\ell)$, and an FO formula $\varphi$ in local form, our \emph{locality-preserving quantifier elimination procedure} produces an updated labeling $\widehat{\ell}$ and a quantifier-free formula $\widehat{\varphi}$ such that, (1)~for every vertex $v\in V(G)$,
\[
(G,\ell)\models\varphi(v)\iff (G,\widehat{\ell})\models\widehat\varphi(v),
\]
and (2)~every intermediate formula produced by the process remains in local form. That is we maintain explicit \(\ball\) predicates throughout the transformation.
All computations and communications therefore occur within the bounded-depth tree (or $r$-ball) associated with the reference vertex.

A key result in this paper is that any local formula can be rewritten in local form. Let us run through this transformation for the \textsc{Twin} formula.

\paragraph{Example: rewriting \textsc{Twin} in local form.}

Recall that \(\textsc{Twin}(x) = \exists y\, \forall z\, 
\big(\mathsf{adj}(x,z) \Leftrightarrow \mathsf{adj}(y,z)\big)\) is \(3\)-local. We start by \emph{relativizing} each quantifier. This way, the locality of \(y\) and \(z\) with respect to \(x\) is explicit:
\[
\textsc{Twin}^r(x) = \exists y \, \Big( \ball(y,x,3) \wedge \forall z\, \Big[\neg\ball(z,x,3) \vee \big(\mathsf{adj}(x,z) \Leftrightarrow \mathsf{adj}(y,z)\big)\Big]\Big).
\]
Note that the formula above is logically equivalent to \(\textsc{Twin}\), i.e., \(\textsc{Twin}(x)\equiv \textsc{Twin}^r(x)\).
For handling the universal quantifier~$\forall z$, that is, rather,  for handling the negated existential quantifier $\neg\exists z$, we introduce a redundant \(\ball(y,x,3)\) predicate. This yields
\begin{align*}
    \textsc{Twin}(x) & \equiv \exists y\,\Big( \ball(y,x,3) \wedge \neg\exists z\, \Big[\ball(z,x,3) \wedge \neg\big(\mathsf{adj}(x,z) \Leftrightarrow \mathsf{adj}(y,z)\big)\Big]\Big)\\
    & \equiv \exists y\, \Big(\ball(y,x,3) \wedge \neg\exists z\, \Big[\ball(y,x,3)\wedge\ball(z,x,3) \wedge \neg\big(\mathsf{adj}(x,z) \Leftrightarrow \mathsf{adj}(y,z)\big)\Big]\Big) \\
    & \equiv \exists y\, \Big(\ball(y,x,3) \wedge \neg\exists z\, \Big[\ball((y,z),x,3) \wedge \neg\big(\mathsf{adj}(x,z) \Leftrightarrow \mathsf{adj}(y,z)\big)\Big]\Big)
\end{align*}
At this point, one more step is needed for obtaining a formula in local form, as the innermost part of the formula, i.e., \(\mathsf{adj}(x,z) \Leftrightarrow \mathsf{adj}(y,z)\), still requires a \(\ball\) predicate to enforce its locality. This yields
\[
    \textsc{Twin}(x) \equiv \exists y\, \Big(\ball(y,x,3) \wedge \neg \exists z\, \Big[\ball((y,z),x,3) \wedge \neg\big(\ball((y,z),x,3) \wedge \mathsf{adj}(x,z) \Leftrightarrow \mathsf{adj}(y,z)\big)\Big]\Big).
\]
This process successfully eventually produces a formula in \(3\)-local form, as desired. In particular, note that,  wherever there is a negation, there is a corresponding  ball predicate whose role is to enforce locality not only before but also \emph{after} the negation.

\subsection{Model Checking, Counting, Optimization, and Certification}

\paragraph{Model Checking General FO Predicates.} 

For checking whether $G \models \varphi$ in \CONGEST, for a given FO formula $\varphi$, and a given labeled graph $G\in\mathcal{G}$, we use the same techniques as the ones used to mark all nodes satisfying a given local FO formula. That is, the proof of Theorem~\ref{thm:thmFOexpCONG-informal} follows the same guidelines as the ones used to establish Theorem~\ref{thm:thmFOlocal-informal}. On the one hand, the approach is simplified as one does not need to care about relativization. On the other hand, the "non-locality" of the formula requires checking properties that potentially occur at nodes that are far away from each other in the graph (e.g., the presence of at least two nodes of degree~3). In the context of quantifier elimination, this translates into counting the number of nodes in which unary  predicates are satisfied. This counting is merely achieved along the edges of a breadth-first search (BFS) tree, in $\cO(D)$ rounds. We show that all other operations involved in quantifier elimination can be performed in $\cO(\log n)$ rounds, for a total of $\cO(D+\log n)$ rounds. 

\paragraph{Counting and Optimization.} 

Establishing Theorem~\ref{thm:thmFOexpCNT-informal}, i.e., showing how to efficiently  count the number of $k$-tuples of vertices $(v_1,\dots,v_k)$ such that $(G,\ell)\models \varphi(v_1,\dots,v_k)$, requires more work. For the global counting problem, we proceed in two stages. 

\medskip

\noindent\emph{Stage 1: Quantifier-free formulas.} 
As for model checking, we use the fact that there exists an integer function $g$ depending only on the class $\mathcal{G}$ such that, for any graph $G \in \mathcal{G}$, and any constant~$p$, $G$ admits a coloring with $g(p)$ colors such that any subset of at most $p$ colors induces a subgraph of treedepth at most~$p$. For $p=k$, each $k$-uple of vertices $(v_1,\dots,v_k)$ of the labeled graph in $(G,\ell)$ satisfying the formula $\varphi(x_1, \dots, x_k)$ covers some set $U$ of colors, with $|U|\leq p$. Assuming given an $\cO(D + \log n)$-round algorithm counting the solutions of $\varphi(x_1,\dots,x_k)$ on labeled subgraphs $(G[U],\ell)$ of treedepth at most~$k$, we show how to count the solutions on the whole labeled graph $(G,\ell)$. For any two sets of colors $U_1, U_2$, the set of solutions of $\varphi(x_1,\dots,x_k)$ on $(G[U_1 \cap U_2],\ell)$ is exactly the intersection of the set of solutions for $(G[U_1],\ell)$ and for $(G[U_2],L)$. This holds thanks to the fact that $\varphi(x_1,\dots,x_k)$ has no  variables other than $(x_1,\dots,x_k)$. As a consequence, by counting the solutions on $(G[U],\ell)$ for every set $U$ of at most $k$ colors, and then using the inclusion-exclusion principle, we are able to count solutions of $\varphi$ on $(G,\ell)$. 

\medskip

\noindent\emph{Stage 2: General formulas.}
    In order to use the previous stage on quantifier free formulas, we first proceed with
    quantifier elimination. This introduces label and $\lca$
    predicates, which prevents us from directly applying the previous stage. To overcome this
    difficulty, we eliminate the $\lca$ predicates by introducing new variables, which are
    existentially quantified. To avoid over-counting, we enforce unicity, i.e., at most one instance of each newly introduced variable may satisfy the formula. Informally, it is as
    if every existential quantifier ``there exists \(y\) such that'' was replaced by
    ``there exists \emph{a unique} \(y\) such that''.

To establish \Cref{thm:thmFOexpOPT-informal}, we first utilize the entire quantifier elimination machinery developed for distributed model-checking in \CONGEST. We emphasize that, during the process, the variables of the quantifier-free formula $\widehat{\varphi}$ are precisely the free variables of $\varphi(\bar{x})$, but $\widehat{\varphi}$ may have a new set $\overline{y}$ of existentially-quantified variables. Let $p = |\overline{x}|+|\overline{y}|$. We construct now  a coloring of $G$ using $g(p)$ colors in $\cO(\log n)$ rounds, and we observe that, for any $k$-uple of vertices $(v_1,\dots,v_k)$ of $G$, we have $(G,\ell)\models \varphi(v_1,\dots,v_k)$ if and only if there exists a set $U$ of at most $p$ colors such that $(G[U],\widehat{\ell})\models \widehat{\varphi}(v_1,\dots,v_k)$. Indeed, $U$ is obtained by guessing not only the colors of $v_1,\dots,v_k$, but also the colors of the tuple of vertices $u_1,\dots,u_{p-k}$ assigned to~$\overline{y}$. Optimizing the solution of $\varphi$ on $G$ boils down to optimizing over all optimal solutions of $\widehat{\varphi}$ on graphs $(G[U],\widehat{\ell})$ of treedepth at most~$p$.

Based on the above, it remains to deal with optimization and global counting for FO formulas on labeled subgraphs $G[U]$ of $G$, each of which has a constant treedepth. For this purpose, we extend the results from~\cite{FominFMRT24} stating that counting and optimization of FO formulae can be achieved on \emph{connected} labeled graphs of bounded treedepth, in a constant number of rounds. Indeed, we show that one can perform optimization and counting for FO formulae over an induced  (not necessarily connected) subgraph of constant treedepth of the labeled graph~$G$, in $\cO(D)$ rounds. In~\cite{FominFMRT24}, the algorithm starts by constructing a tree decomposition of small width of the graph, and then performs dynamic programming over this tree decomposition. If the input graph of small treedepth is connected, this directly enables the construction of a rooted tree decomposition of constant depth and constant width. Moreover, the decomposition tree is a spanning tree of the graph, which allows communications to proceed with dynamic programming bottom-up the tree in a constant number of rounds. The root of the tree can then compute the output by simulating the centralized algorithm for optimization and counting on graphs of bounded treewidth of~\cite{BoPaTo92}. However, in our case, the graphs $G[U]$ may not be connected, and thus we must deal with a forest decomposition~$ F$ instead of a mere tree decomposition. A BFS tree rooted at an arbitrary vertex can span all the roots of the trees in the forest~$F$, resulting in a rooted tree of depth $D + \cO(1)$. This tree can serve as a base for constructing a tree-decomposition of $G[U]$. Moreover, since the nodes of the tree-decomposition correspond to vertices of~$G$, the centralized dynamic programming algorithm of~\cite{BoPaTo92} can be simulated for counting and optimization on $G[U]$, in $D + \cO(1)$ rounds. The information stored in the dynamic programming tables can indeed be computed in a sequence of rounds, starting from the leaves, upward to the root of the tree decomposition by communicating along the edges of~$G$. Altogether, the whole process for counting and optimization for formula $\varphi$ on the pair $(G,\ell)$ costs $\cO(D + \log n)$ rounds in \CONGEST.

Our method for local counting follows the same lines as that for global counting. That is, we first reduce the study to quantifier-free formulas, using the (locality preserving) quantifier
elimination method. Similarly to the case of global counting, this reduction introduces \lca-predicates,
which need to be eliminated in a locality-preserving fashion. It turns out that the method
applied to the case of global counting happens to preserve locality, although the proof is
somewhat technical.
For quantifier-free formulas, we essentially apply the same idea as for global counting, based on the
inclusion-exclusion principle. However, we also use the fact that the formula is
local, by considering some set of colors~\(U\), and a tuple of vertices \(\overline v\) such that
\((G[U],\ell)\) satisfies \(\varphi(\overline v)\). Then, since \(\varphi\) is local, all vertices
in \(\overline v\) belong to the same connected component of~\(G[U]\), which implies that we only
need to count the solutions to \(\varphi\) in the connected components of~\(G[U]\). The latter are
connected graphs of bounded treedepth, and thus counting is achieved  in \(\cO(1)\) rounds using the algorithm in~\cite{FominFMRT24}.

\paragraph{Distributed Certification.}


At a high level, our certification scheme for an FO formula $\varphi$ and a class of graphs $\mathcal{G}$ with bounded expansion essentially consists of the prover providing each node with the transcript of the execution of our distributed model-checking algorithm at this node. The role of the verifier is then to check (in a distributed manner) whether this transcript encodes a correct execution of our algorithm. Note that the certification scheme does not have to certify the bounded expansion property of $\mathcal{G}$, which is a promise.  However, given a coloring of the input graph $G=(V,E)\in \mathcal{G}$, the certification scheme must certify that any set $U$ of $k=|\varphi|$ colors does induce a subgraph with treedepth~$k$. This is done in a specific manner, using an algorithm by Bousquet et al.~\cite{BousquetFP21} for certifying MSO formulas on graphs of bounded treedepth with certificates of size $\mathcal {O} (\log n)$. The main difficulty is to encode the transcript of the quantifier elimination procedure, i.e., the transformation from $(\varphi,G)$ to $(\widehat{\varphi},\widehat{\ell})$ with $\widehat{\varphi}$ quantifier-free. Each $\lca$ predicate refers to a local condition, and the value of such a predicate can be certified locally at each node, in forests of bounded depth. On the other hand, the unary predicates may refer to non-local properties. Nevertheless, we demonstrate that verifying the correctness of the unary predicates can be accomplished using a spanning tree whose root is capable of making decisions. (Certify a spanning tree is standard using $\cO(\log n)$-bit certificates~\cite{KormanKP10}.) 

\section{Model and Definitions}
\label{sec:model-and-definitions}

In this section, we recall basic concepts of distributed computing, logic, and graphs. Readers familiar with these frameworks may skip the first two sub-sections, but Section~\ref{ss:DefLocal}, which introduces and extends the notion of \emph{local} formulas, and Section~\ref{ss:DefBE}, which recalls properties of graphs with \emph{bounded expansion}, may require special attention. 

\subsection{The CONGEST Model}

The \CONGEST model is a standard model of distributed computation in networks~\cite{Peleg2000}. The network is modeled as a connected simple (i.e., no self-loops nor multiple edges) graph $G=(V,E)$. The $n\geq 1$ vertices of $G$ model the processing nodes in the networks, and its edges models the communication links. In the following, we may indistinguishably use ``node'' of ``vertex'' for referring to the processing elements in the network. Each node is assigned an identifier, which is assumed to be unique in the network. The identifiers may not be between 1 and~$n$, but the \CONGEST model assumes that each identifier can be stored on $\cO(\log n)$ bits in $n$-node networks. In fact, it is usually assumed that all the nodes of a graph $G$ are given a polynomial upper bound $N=\mbox{poly}(n)$ on the number $n$ of vertices in~$G$, with all identifiers in $\{1,\dots,N\}$.  

Computation proceeds in lockstep, as a sequence of synchronous rounds. All nodes start simultaneously, at the same round. At each round, every node~$v$ sends an $\cO(\log n)$-bit message to each of its neighbors, receives the messages sent by all its neighbors, and performs some individual computation. The round-complexity of an algorithm $A$ is the function $R_A(n)=\max_{|V(G)|\leq n}R_A(G)$ where $R_A(G)$ denotes the number of rounds performed by the nodes executing $A$ in $G$ until all of them terminate. 

As said before, it is worth mentioning the \textsf{BROADCAST CONGEST} model, that is, the weaker variant of the \CONGEST model in which, at every round, each node must send the \emph{same} $\cO(\log n)$-bit message to all its neighbors. For the sake of simplifying the notations, we solely refer to the \CONGEST model in the statements of our results. Yet, all our upper bound results also holds in the weaker \textsf{BROADCAST CONGEST} model, apart from  \Cref{thm:thmFOexpOPT}.

\paragraph{Remark.}

Throughout the paper, we will use the following convention regarding the inputs and outputs of our algorithms. When an algorithm receives a ``global value'' as input, whether it be an integer~$p$, a logical formula~$\varphi$, or a function (e.g., a bound $f:\mathbb{N}\to\mathbb{N}$ on the expansion), we assume that all nodes are given this global value as input. Instead, when an algorithm receives as input a node labeling function~$\ell$, we assume that each node $v$ is only given its label $\ell(v)$. Similarly, when an algorithm returns a node-labeling function $\widehat\ell$, we mean that each node $v$ returns its label $\widehat\ell(v)$.

\subsection{First-Order Formulas on Labeled Graphs}
\label{sec:FOlabel}

\subsubsection{Definition} 

Recall that a \emph{first-order} (FO) formula over graphs is a combination (conjunction, disjunction, and negations) of Boolean predicates over variables, which may be quantified or not.  Every variable represents a vertex of a graph, and there are two binary predicates: equality $x=y$, and adjacency $\adj(x,y)$, where the latter states whether or not there is an edge between the two vertices $x$ and~$y$. The variables of an FO formula can be quantified with existential or universal quantifiers. Non-quantified variables are called \emph{free variables}. Given an FO formula $\varphi$ with free variables $x_1, \dots, x_k$, we make explicit that the latter are free by denoting $\varphi=\varphi(x_1, \dots, x_k)$. In the following, $\FO$ denotes the set of all first order formulas over undirected graphs. 
 
Throughout the paper, a tuple  $(x_1, \dots, x_k)$ is abbreviated in $\overline{x}$. The length $k$ of a tuple $(x_1, \dots, x_k)$ is then denoted by~$|\overline{x}|$. Expressions such as $\exists \overline{x}$ or $\forall \overline{x}$ are shorthand for $\exists x_1 \dots \exists x_k$ and $\forall x_1 \dots \forall x_k$, respectively. We may also abuse notation in context where no confusions may occur by letting $\overline{x}$ denote the \emph{set} $\{x_1, \dots, x_k\}$. 

Given an FO formula $\varphi$ and a graph $G$, $G\models \varphi$ denotes that $G$ satisfies $\varphi$ (i.e., $G$ is a \emph{model} for~$\varphi$). The negation of $G\models \varphi$ is denoted by $G\not\models \varphi$. Given an FO formula $\varphi(\overline{x})$ with free variables, a graph $G = (V,E)$, and a set of vertices $\overline{v} \in V^k$ where $k=|\overline{x}|$, we say that $G$ satisfies $\varphi(\overline{x})$ for vertices $\overline{v}$, denoted by $G\models \varphi(\overline{v})$, if the formula is satisfied by~$G$ whenever assigning $\overline{v}$ to the free variables of $\varphi(\overline{x})$.  

A graph property is called a \emph{first-order} (FO) property if it can be expressed by a first-order formula with no free variables. For example, the property ``the graph contains a triangle'' is clearly an FO property: it holds whenever there is a set of three pairwise adjacent vertices in the graph. Indeed, it can be written as the formula
\[
\varphi \;=\; \exists x_1 \,\exists x_2 \,\exists x_3\,
\bigl(
\mathrm{adj}(x_1,x_2) \;\wedge\;\mathrm{adj}(x_2,x_3) \;\wedge\;\mathrm{adj}(x_1,x_3)
\bigr).
\]
More generally, the property of containing a fixed graph~\(H\) as a subgraph (sometimes called ``subgraph isomorphism'') is an FO property. Its negation---namely, the property that a graph does not contain~\(H\) as a subgraph (often called ``\(H\)\!-freeness'')---is also an FO property, since first-order logic is closed under negation.

A formula with free variables expresses properties over graphs and tuples of vertices. As an example, the following formula $\psi(x)$ holds whenever $x$ is a node in a triangle, formulated as
$$
\psi(x) =  \exists y_1 \exists y_2, \adj(x,y_1) \wedge \adj(x,y_2) \wedge \adj(y_1,y_2).
$$  
The graph $G$ with vertex set~$V=\{a,b,c,d\}$, and edge set $E=\{\{a,b\}, \{a,c\}, \{b,c\}, \{c,d\}\}$ satisfies $G\models \varphi$ as the vertices $\{a,b,c\}$ form a triangle. Similarly, $G\models \psi(a)$, but $G  \not\models \psi(d)$ as $d$ is not in any triangle. 

\subsubsection{FO Formulas for Labeled Graphs}
\label{sss:labeled-fo-def}

The formalism of FO properties can be extended to labeled graphs. Let $\Lambda$ be a finite set, whose elements are called \emph{labels}. A $\Lambda$-labeled graph is a pair $(G,\ell)$ where $G = (V,E)$ is a graph, and $\ell: V \rightarrow 2^{\Lambda}$ is an assignment of a set of labels to each vertex of~$G$.  An FO formula over $\Lambda$-labeled graphs  is a first order formula $\varphi$ over graphs, extended with predicates of the form $\lab_\lambda(x)$ for every $\lambda \in  \Lambda$. The predicate $\lab_\lambda(x)$ is true if vertex $x$ contains $\lambda$ as one of its labels.  Such predicates are called \emph{label predicates}. Whenever $\lab_\lambda(x)$ holds, we say that $x$ is labeled with~$\lambda$. Note that since $\ell$ assigns a set of labels to each vertex, it may be the case that $\lab_\lambda(x)\land\lab_{\lambda'}(x)$ holds for two different labels $\lambda$ and $\lambda'$. It may thus be convenient to  define the predicates $\lab_S$ for sets $S\subseteq \Lambda$, as follows. Given $S\in\Lambda$, $\lab_S(x)$ holds if and only if $\ell(x)=S$. In particular, $\lab_S(x)\land \lab_{S'}(x)$ is necessarily false whenever $S\neq S'$.  The set of all first-order logic formulas over $\Lambda$-labeled graphs is denoted $\FO[\Lambda]$. 

Given a $\Lambda$-labeled graph $(G,\ell)$, and a subgraph $G'$ of $G$, we abuse notation and denote $(G',\ell)$ the $\Lambda$-labeled graph defined by the restriction of $\ell$ to the nodes in~$G'$.
For instance, let $\Lambda = \{0,1\}$, and let us consider the following graph property on labeled graphs: ``existence of a triangle with a node labeled~$1$''. This property can be expressed by the formula:
$$
\varphi = \exists x_1 \exists x_2 \exists x_3, \adj(x_1,x_2) \wedge \adj(x_2,x_3) \wedge \adj(x_1,x_3) \wedge  \lab_1(x_1).
$$ 
We may also consider formula with free variables, such as 
$$
\psi(x) =  \exists y_1 \exists y_2, \adj(x,y_1) \wedge \adj(x,y_2) \wedge \adj(y_1,y_2) \wedge  \lab_1(x).
$$  
The graph $G$ with vertex set $V=\{a,b,c,d\}$ and edge set $E=\{\{a,b\}, \{a,c\}, \{b,c\}, \{c,d\}\}$ with labels $\ell(a) = \{0\}, \ell(b) = \{0\}, \ell(c) = \{1\}$, and $\ell(d) = \{1\}$ satisfies $(G,\ell) \models \varphi$ as the nodes $\{a,b,c\}$ form a triangle where $c$ is labeled~$1$. For the same reason, we have that $(G,\ell)\models \psi(c)$. However, $(G,\ell)\not\models \psi(a)$ as $a$ is not labeled~$1$ (even if it is in a triangle), and $(G,\ell) \not\models \psi(d)$ as node $d$ is not in a triangle (even if it is labeled $1$). 

For $\varphi(x_1, \dots, x_k) \in \FO[\Lambda]$, and a $\Lambda$-labeled graph $(G,\ell)$, we denote by $\true(\varphi,G,\ell)$ the set of all tuples of vertices $(v_1,\dots, v_k)$ satisfying that $(G,\ell)\models \varphi(v_1, \dots, v_k)$. If $\varphi$ has no free variables then $\true(\varphi, G, \ell)$ is equal to the truth value of  $(G,\ell)\models \varphi$.
Two $\FO[\Lambda]$ formulas $\varphi_1$ and $\varphi_2$ are \emph{equivalent} if $\true(\varphi_1,G,\ell) = \true(\varphi_2,G,\ell)$  for every $\Lambda$-labeled graph $(G,\ell)$. The following basic definitions play a central role in this paper. 
\begin{itemize}
    \item A formula $\varphi(\overline{x})\in \FO[\Lambda]$  is called  \emph{quantifier-free} if it has no quantifiers, i.e., it can be expressed as a boolean combination of predicates involving only the free variables in~$\overline{x}$. 
    \item A formula $\varphi(\overline{x}) \in \FO[\Lambda]$ is \emph{existential} if it can be written as $\varphi(\overline{x}) = \exists \overline{y}\, \psi(\overline{x}, \overline{y})$, where $\psi(\overline{x}, \overline{y})$ is quantifier-free formula in $\FO[\Lambda]$. 
\end{itemize}
Furthermore, a formula $\varphi(\overline{x})$ is in \emph{prenex normal form} if it can be written in the form
$$
\varphi(\overline{x})=Q_1 y_1 \, \dots \, Q_t y_t \; \psi(\overline{x}, \overline{y}),
$$ 
where $\psi(\overline{x}, \overline{y}) \in \FO[\Lambda]$, and $Q_1, \dots, Q_t$ are existential or universal quantifiers. The value of $t$ is the \emph{quantifier depth} of $\varphi$. Every $\FO$ formula  admits an equivalent formula in prenex normal form (see, e.g., \cite{epstein2011classical}). This fact can be easily established by induction on the length of the formula, and it can be  extended to $\FO[\Lambda]$, i.e., first order formulas for labeled graphs, in a straightforward manner. 

\paragraph{Notation.} For every family $\mathcal{G}$ of graphs, and every finite set $\Lambda$, we denote by $\mathcal{G}[\Lambda]$ the set of all $\Lambda$-labeled graphs $(G,\ell)$ with $G=(V,E)\in\mathcal{G}$, and $\ell:V\to 2^\Lambda$.

\subsection{Local Formulas}
\label{ss:DefLocal}

\emph{Local formulas} can be defined as follows. For a graph
\(G = (V,E)\) and a vertex \(v\in V\), let  \(B_G(v, r)\) denote the \(r\)-neighborhood of \(v\)
in \(G\), that is, the subgraph of \(G\) induced by vertices at distance at most \(r\) from
\(v\). By extension, when considering a labeled graph \((G,\ell)\), \(B_G(v,r)\) is also a
labeled graph: each node \(u \in B_G(v,r)\) is labeled by the same set of labels~\(\ell(u)\).

\begin{definition}\label{def:r-local}
    Let \(\Lambda\) be a finite set of labels, and let $r\geq 0$ be an integer. A formula \(\varphi(x) \in FO[\Lambda]\) with one free variable is \emph{\(r\)-local} if, for any two \(\Lambda\)-labeled graphs
    \((G,\ell)\) and
    \((G',\ell')\), and for any two vertices $v \in V(G)$ and $v'\in V(G')$, we have
    \[
    B_G(v,r) = B_{G'}(v',r) \implies
    \Bigl((G,\ell) \models \varphi(v) \iff (G',\ell') \models \varphi(v')\Bigr).
    \]
    where \(B_G(v,r) = B_{G'}(v',r)\) means that there exists a label-preserving  isomorphism
    \(h: B_G(v,r) \to B_{G'}(v',r)\) such that \(h(v) = v'\).
    
    We say \(\varphi(x)\) is \emph{local} if there exists $r\geq 0$ such that $\varphi(x)$ is $r$-local.
\end{definition}

For technical reasons that will appear clear later, we need to extend the notion of local formulas
to formulas \(\varphi(x, \overline y)\) with more than a single free variable.
Our notion of locality is defined with respect to the first free variable \(x\) of the formula:
intuitively, we require that each variable is in a ball centered on \(x\).
This is different from the notion of Gaifman-locality \cite{gaifman1982local},
in which \emph{quantified}
variables are required to be in a ball centered on \emph{any} of the free variables.

\begin{definition}\label{def:r-local-general}
    Let \(\Lambda\) be a finite set of labels, and let \(r\geq 0\) be an integer.
    A formula \(\varphi(x, \overline y)\in FO[\Lambda]\) is \emph{\(r\)-local} if,
    for any \(\Lambda\)-labeled graph \((G,\ell)\) with $G=(V,E)$, the following two properties are satisfied:
    \begin{itemize}
        \item for every pair \((v,\overline u) \in \true(\varphi, G, \ell)\), it holds that \(\overline u\subseteq B_G(v,r)\);
        
        \item for every \(\Lambda\)-labeled graph \((G', \ell')\) with $G'=(V',E')$, for every vertices
        \(v \in V, v' \in V'\), $\overline u \subseteq B_G(v,r)$ and $\overline{u'}\subseteq B_{G'}
        (v',r)$, if there exists a label-preserving  isomorphism \(h: B_G(v,r) \to B_{G'}(v',r)\) such that
        \(h(v) = v'\) and \(h(\overline u)\) = \(\overline{u'}\), then:
        \[
            (G,\ell) \models \varphi(v,\overline u) \iff
            (G', \ell') \models \varphi(v', \overline{u'}
)        \]
    \end{itemize}
\end{definition}

We shall often write \(\varphi(\overline x)\) instead of \(\varphi(x,\overline y)\),
in which case it is  \(x_1\) that takes the particular role of \(x\) in the definition above.
 
The above definition is, unfortunately, very difficult in practice, because it provides little
information on the form of the considered formula. We therefore  define the
notion of \emph{local form} in \Cref{def:r-local-form}.
To that end, let us introduce the new predicate $\ball(y,x,r)$ defined as 
\[
G \models \ball(u,v,r) \iff u \in B_G(v,r). 
\]
More generally, for $\overline{y}=(y_1,\dots,y_k)$, we define 
\[
\ball(\overline{y},x,r)=\bigwedge_{i=1}^{|\overline y|} \ball(y_i,x,r).
\]
These predicates will be referred to as \emph{distance predicates}.
Note that they can easily be written as existential FO formulas, see \Cref{prop:beta}.
We use distance predicates to \emph{relativize} quantifiers with respect to the \(r\)-neighborhood
of a given vertex (relativization of quantifiers is standard technique, see for example \cite{hodges1993model,Ebbinghaus1995}):
let us then consider an existential \(r\)-local formula \(\varphi(x) = \exists \overline y\, \psi(x, \overline y)\). Since $\varphi(x)$ is $r$-local, it can be rewritten as
\[\varphi(x) = \exists \overline y\; \big(\psi(x, \overline y) \wedge \ball(\overline y,x,r)\big).\]
Such a formula, where the locality is enforced explicitly by a distance predicate, is referred to as a formula in
\emph{\(r\)-local form}. The following definition formalizes this notion.

\begin{definition}\label{def:r-local-form}
Let $r\geq 0$ be an integer. A formula \(\varphi(x, \overline y) \in FO[\Lambda]\) with at
least one free variable is in \emph{\(r\)-local form} if 
\begin{itemize}
    \item \(\varphi(x,\overline y)=\psi(x,\overline y) \wedge \ball(\overline y,x, r)\), where
    \(\psi\in FO[\Lambda]\) is quantifier-free, or
    
    \item \(\varphi(x, \overline y)=\exists z\; \psi(x, \overline y, z)\), where \(\psi\in FO[\Lambda]\) is in \(r\)-local
    form, or
    
    \item \(\varphi(x,\overline y)=\ball(\overline y,x,r)
    \wedge \forall z\; \big( \neg \ball((\overline y,z),x,r)
     \vee \psi(x,\overline y, z)\big)\) where
    \(\psi\in FO[\Lambda]\) is in \(r\)-local form.
    \end{itemize}
    We say \(\varphi(x, \overline y)\) is in \emph{local form} if there exists $r\geq 0$ such that $\varphi(x, \overline y)$ is in \(r\)-local form.
\end{definition}

When considering a formula with several free variables, we say it is (\(r\)-)local if it can be
rewritten in (\(r\)-)local form.

\paragraph{Remark.}

The last case of Definition~\ref{def:r-local-form} introduces a lot of redundancy, with two distance predicates, in addition to requiring that \(\psi\) must be in \(r\)-local form. The main motivation for this redundancy is to guarantee that 
every ``subformula'' of \(\varphi\) remains local. Thanks to this redundancy, all the
formulas below are local, despite the fact that negation does \emph{not} preserve locality
for formulas with several free variables.
In particular, the following two formulas are equivalent, where the universal quantifier is merely replaced by a negated existential quantifier: 
\begin{align*}
\ball(\overline y,x,r) &
    \wedge \forall z\; \big( \neg \ball((\overline y,z),x,r)
     \vee \psi(x, \overline y, z)\big) \\
& \equiv 
\ball(\overline y, x, r)
    \wedge \neg\exists z\; \big( \ball((\overline y,z), x, r)
    \wedge \neg \psi(x,\overline y,z)\big).
\end{align*}
Similarly,  the formulas 
\[
\exists z\; \big( \ball((\overline y,z), x, r)
    \wedge \neg \psi(x,\overline y,z)\big), \;\text{and}\;
\ball((\overline y,z), x, r) 
    \wedge \neg \psi(x, \overline y, z)
\]
are both local. In other words, the implicit or explicit presence of negations forces us to introduce redundancies in the 
distance predicates for ``balancing''  the potential loss of the locality property.

\medskip

The following result establishes the consistency of Definitions~\ref{def:r-local-general} and~\ref{def:r-local-form}.  

\begin{lemma}\label{prop:local_form}
For every integer $r\geq 0$, every \(r\)-local formula can be rewritten in \(r\)-local
    form.
\end{lemma}

\begin{proof}
    Let \(\varphi(x, \overline y)\) be \(r\)-local formula, and assume, w.l.o.g, that it is written in prenex normal form
    $\varphi(x, \overline y)=Q_1 z_1 \, \dots \, Q_t z_t \; \psi(\overline x, \overline{y}, \overline z)$.
    First, we rewrite \(\varphi\) by \emph{relativizing} the quantifier with respect to the
    \(r\)-neighborhood of~\(x\). Intuitively, we replace existential quantification  \(\exists z\) by
    \(\exists z \in B_G(x, r)\), i.e., $z$ must be in the ball of radius~$r$ around~$x$ in~$G$.  Similarly universal quantification \(\forall z\) is replaced by \(\forall z \in B_G(x,r)\), i.e., the formula is checked only for nodes $z$ close to~$x$.
    Formally, for the \(r\)-local formula \(\varphi(x, \overline y)\),
    we define \(\varphi^r(x,\overline y)\) by induction on quantifier depth.
    If \(\varphi\) is quantifier free, then \(\varphi^r = \varphi\).
    If \(\varphi(x, \overline y) = \exists z, \zeta(x, \overline y, z)\), then
    \(\varphi^r(x, \overline y) = \exists z, \ball(z, x, r)\wedge \zeta^r(x,\overline y,z)\).
    Finally, if \(\varphi(x,\overline y) = \forall z, \zeta(x,\overline y,z)\), then \(\varphi^r(x,\overline y) = \neg \exists z, \ball(z, x, r) \wedge \neg \zeta^r(x, \overline y, z)\).
    
    By definition, for \(v \in V\) and \(\overline u \subseteq B_G(v,r)\),
    \[
        G\models \varphi^r(v,\overline u)
        \iff B_G(v,r) \models \varphi^r(v, \overline u)
        \iff B_G(v,r) \models \varphi(v, \overline u).
    \]
    By locality of \(\varphi\), we deduce that
    \[
        G \models \ball(\overline u, v, r) \wedge \varphi^r(v, \overline u) \iff
        B_G(v,r) \models \varphi(v, \overline u) \iff
        G \models \varphi(v, \overline u).
    \]
    It only remains to show that the formula
    \(\ball(\overline y, x, r) \wedge \varphi^r(x,\overline y)\) can be rewritten in \(r\)-local
    form. Once again, we proceed by induction on quantifier depth. By definition of \(\varphi^r\),
    there are only three cases to consider.
    \begin{description}
        \item[Case 1:] \(\varphi\) is quantifier free.  Then \(\varphi^r = \varphi\)
    and \(\ball(\overline y, x, r) \wedge \varphi(x, \overline y)\) is already in local form.

    \item[Case 2:] \(\varphi(x, \overline y) = \exists z \, \psi(x,\overline y,z)\). Then,
    \begin{align*}
        \ball(\overline y, x, r) \wedge \varphi^r(x,\overline y)
        & = \ball(\overline y, x, r) \wedge
        \exists z \, \big( \ball(z, x,r) \wedge \psi^r(x, \overline y, z)\big) \\
        & = \exists z \, \big(\ball((\overline y,z), x,r)\wedge \psi^r(x, \overline y,z)\big)
    \end{align*}
    By induction, we can rewrite \(\ball((\overline y,z), x,r)\wedge \psi^r(x, \overline y,z)\)
    in \(r\)-local form.

    \item[Case 3:] \(\varphi(x,\overline y) = \forall z, \psi(x,\overline y,z)\). Then, 
    \begin{align*}
        \ball(\overline y, x, r) \wedge \varphi^r(x, \overline y) 
        & = \ball(\overline y, x, r) \wedge
        \neg\exists z \,\big( \ball(z, x, r) \wedge \neg\psi^r(x, \overline y, z)\big )
        \\
        &= \ball(\overline y, x, r) \wedge \neg\Big(\ball(\overline y,x,r) \wedge
        \exists z \, \big(\ball(z,x,r) \wedge \neg\psi^r(x,\overline y,z)\big)
        \Big)
        \\
        &= \ball(\overline y, x, r) \wedge \neg
        \exists z \, \big(\ball((\overline y,z),x,r) \wedge \neg\psi^r(x,\overline y,z)
        \big )
        \\
        &= \ball(\overline y, x, r) \wedge \forall z \, \big( \neg\ball((\overline y, z), x,r)
        \vee \psi^r(x,\overline y,z)\big)
        \\
        &= \ball(\overline y, x, r) \wedge \forall z \, \Big( \neg\ball((\overline y, z), x,r)
        \vee \big( \ball((\overline y,z), x,r) \wedge \psi^r(x, \overline y,z)\big)\Big)
    \end{align*}
    By induction, \(\ball((\overline y,z), x,r) \wedge \psi^r(x, \overline y,z)\) can be rewritten
    in \(r\)-local form.
        \end{description}
The study of these three cases completes the proof.
\end{proof}

We conclude this section by some simple observation that will be useful at several places in the paper.

\begin{lemma}
\label{prop:beta}
Let $r \in \N$, and let $\overline{x}=(x_1,\dots,x_k)$. There exists an existential FO formula 
\[
\beta_r(\overline x) = \exists \overline y\,
    \psi(\overline x, \overline y),
\] where \(\psi\) is quantifier free such that, for every graph \(G\), and every $k$-tuple $\bar v$ of vertices of~$G$, 
\[
G \models \beta_r(\overline v) \iff G \models \ball(\overline v, v_1, r).
\]
Moreover, for every \(r' \geq r\), every graph \(G\), and every $k$-tuple $\bar v$ of vertices of~$G$,
\[
G \models \beta_r(\overline v) \iff G \models \beta_{r,r'}(\overline v),
\]
where \(\beta_{r,r'}(\overline x) = \exists \overline y \,\big( \psi(\overline x, \overline y)
    \wedge \ball((\overline x, \overline y), x_1, r')\big)\).
\end{lemma}

\begin{proof}
    It is sufficient to describe \(\beta_r(x,y)\) with  two variables only, as $\beta_r(x_1, \dots, x_k)$ can be written as 
    \(\bigwedge_{j=2}^k \beta_r(x_1, x_j)\).
    We define
    \[
    \beta_r(x, y) = \exists z_0 \,\dots\, \exists z_r \, \left[ (x = z_0) \wedge
    \bigvee_{j = 0}^r  \Big((y = z_j) \wedge \bigwedge_{i=0}^{j-1} \adj(z_i, z_{i+1})
    \Big)\right].
    \]
    In this formula, \(j\) is the distance between \(x\) and \(y\), and \(z_0,\dots, z_j\)
    represents the path from \(x\) to \(y\). This formula satisfies the first item of
    \Cref{prop:beta}. Let us now prove that it satisfies the second item.
    Let \(r' \geq r\). By definition, if 
    \(G\models \beta_{r,r'}(u,v)\) then \(G\models \beta_{r}(u,v)\).
    Conversely, let us assume that \(G \models \beta_r(u,v)\). Then, there exists some
    \(\overline w=w_0, \dots, w_r\)
    such that \(G\models \psi(u,v, \overline w)\). Moreover, there exists \(j \in [0, r]\) such
    that 
    \[
    G \models (u = w_0) \wedge (v = w_j) \wedge \bigwedge_{i=0}^{j-1} \adj(w_i, w_{i+1}).
    \]
    For every \(i \in [0,j]\), let \(w'_i = w_i\), and, for every \(i \in [j+1, r]\), let \(w'_i = u\). By definition, every vertex 
    \(w'_i\) satisfies \(\ball(w'_i, u, r)\), and thus \(\ball(w'_i, u,r')\) since \(r'\geq r\). This implies that
    \[
    G \models (u = w'_0) \wedge (v = w'_j) \wedge \bigwedge_{i=0}^{j-1} \adj(w'_i, w'_{i+1})
    \wedge \ball((u,v,\overline {w'}), u, r').
    \]
    As a consequence,  \(G \models \psi(u, v, \overline {w'})
    \wedge \ball((u,v,\overline {w'}), u, r')\), and thus \(G \models \beta_{r,r'}(u,v)\).
\end{proof}

\begin{lemma}\label{lem:local-disjunction}
    Let \(\Lambda\) be a set of labels, and let \(\varphi(\overline x)\in\FO[\Lambda]\) be a
    \(r\)-local, quantifier-free formula. If
    \(\varphi(\overline x) = \varphi_1(\overline x) \vee \varphi_2(\overline x)\) for some formulas
    \(\varphi_1\) and \(\varphi_2\), then both \(\varphi_1\) and \(\varphi_2\) must be \(r\)-local.
\end{lemma}

\begin{proof}
    Since \(\varphi_1\) and \(\varphi_2\) play symmetric roles, we  only prove that
    \(\varphi_1\) is local. Since $\varphi$ is quantifier free, and therefore $\varphi_1$ as well, checking the $r$-locality of $\varphi_1$ is equivalent to checking
    that, for all \(\Lambda\)-labeled graphs \((G,\ell)\), 
    \[
        (G,\ell) \models \varphi_1(\overline v) \implies \overline v \subseteq B_G(v_1, r).
    \]
    This holds trivially as, by locality of \(\varphi\),
    \[
        (G,\ell) \models \varphi_1(\overline v) \implies
        (G,\ell) \models \varphi(\overline v) \implies
        \overline v \subseteq B_G(v_1, r), 
    \]
    which completes the proof. 
\end{proof}

\subsection{Graphs of Bounded Expansion }\label{ss:DefBE}

Recall that a graph $G=(V,E)$ has \emph{treedepth} at most~$d$ if there exists a forest~$F$ of rooted trees with depth at most~$d$ satisfying that there exists a one-to-one mapping $f:V\to V(F)$ such that, for every edge $\{u,v\}\in E$, $f(u)$ is an ancestor of $f(v)$ in a tree of~$F$, or $f(v)$ is an ancestor of $f(u)$ in a tree of~$F$. The treedepth of a graph is the smallest~$d$ such that such a mapping $f$ exists. 

It has been recently shown in~\cite{FominFMRT24} that decision, counting and optimization of FO formulas can be performed in $\cO(1)$ rounds in \CONGEST, for \emph{connected} labeled graphs of bounded treedepth.

\begin{proposition}[\cite{FominFMRT24}]\label{thm:FFMRT24}
    Let $\mathcal{G}[\Lambda]$
    be a class of \emph{connected} labeled graphs of bounded treedepth, and let $\varphi(x_1,\dots,x_k) \in FO[\Lambda]$. Then, for any 
    and any $(G,\ell) \in \mathcal{G}[\Lambda]$, the following problems can be solved by a distributed algorithm running in  $\cO(1)$
    rounds in \CONGEST: 
    \begin{description}
        \item[- Deciding $\varphi$:]  deciding whether it exists a $k$-uple of vertices $(v_1,\dots,v_k)$ such that\linebreak $(G,\ell)\models \varphi(v_1,\dots,v_k)$.
        \item[- Counting the solutions of $\varphi$:] computing the number of $k$-uples $(v_1,\dots,v_k)$ such that\linebreak $(G,\ell)\models \varphi(v_1,\dots,v_k)$.
        \item[- Optimizing $\varphi$:] assuming $n$-node graphs $G=(V,E)$ are provided with a weight function\linebreak ${\omega:V\to \mathbb{N}}$ satisfying $\omega(v)=O(n^c)$ for some $c\geq 0$, computing a $k$-tuple of vertices $(v_1,\dots,v_k)$ such that $G\models\varphi(v_1,\dots,v_k)$, and $\sum_{i=1}^k\omega(v_i)$ is maximum (or minimum).
    \end{description}
\end{proposition}

Graphs of bounded expansion~\cite{NesetrilOdM12} can be defined in term of local minors, and
degeneracy. It is however more convenient for us to use the
following characterization. A class $\mathcal{G}$ of graphs has \emph{bounded expansion} if there
exists a function $f:\mathbb{N}\to\mathbb{N}$ such that, for every $G\in\mathcal{G}$, and for every
$k\in \mathbb{N}$, there exists a coloring of the vertices of $G$ with colors in
$[f(k)]=\{1,\dots,f(k)\}$ such that, for every set $U\subseteq [f(k)]$
of cardinality at most~$k$
the subgraph of~$G$ induced by the nodes colored by a color in~$U$ has
treedepth at most~$|U|$. In particular, each subset of at most $k$ colors induces a subgraph of
treedepth at most $k$. 
For every positive integers $p$ and $k$, a \emph{$(p,k)$-treedepth coloring} of a graph $G$ is a $k$-coloring of $G$ satisfying that each set $U\subseteq[k]$ of at most $p$ colors induces a graph of treedepth at most $|U|$. 

\begin{definition}
    Given a function $f:\mathbb{N}\rightarrow \mathbb{N}$, a class of graphs $\mathcal{G}$ has \emph{expansion} $f$ if, for every $p\in \mathbb{N}_{>0}$, every $G\in \mathcal{G}$ has a $(p, f(p))$-treedepth coloring. A class of graphs has \emph{bounded} expansion if it has expansion $f$ for some function~$f:\mathbb{N}\rightarrow \mathbb{N}$.  
\end{definition}

A class of graphs has \emph{effective} bounded expansion if it has expansion $f$ for some Turing-computable function~$f$. 
In the following, we always assume that the considered functions $f$ are Turing-computable. Interestingly, treedepth colorings can be efficiently computed in the \CONGEST model. 

\begin{proposition}[\cite{NesetrilM16}]\label{th:congestltd}
   Let $\mathcal{G}$ be a class of graphs of bounded expansion~$f$. There exists a function $g: \mathbb{N} \to \mathbb{N}$ such that, for any graph $G \in \mathcal{G}$, and any constant positive integer~$p$, a $(p,g(p))$-treedepth coloring of $G$ can be computed in $\cO(\log n)$ rounds in the \CONGEST\ model.
\end{proposition}

The algorithm of \cite{NesetrilM16} takes as input the function~$f$, and an integer~$p$, and outputs a $(p,g(p))$-treedepth coloring of the actual graph~$G$ in a distributed manner. The function $g$ may however be different from the given bound~$f$ on the expansion of the class~\(\mathcal G\), which may itself be different from the ``best'' expansion function for the class~\(\mathcal G\).

\section{Distributed Model Checking of Local Formulas}
\label{sec:local-model-checking}

We have now all the ingredients to solve the open problem in~\cite{NesetrilM16}. 

\begin{theorem}
\label{thm:thmFOlocal}
Let $\Lambda$ be a finite set. For every local FO formula $\varphi(x)$ on $\Lambda$-labeled graphs, and for every class of graphs $\mathcal{G}$ of bounded expansion, there exists a deterministic algorithm that, for every $n$-node network $G=(V,E)\in \mathcal{G}$ and every labeling $\ell:V\to 2^\Lambda$, marks all vertices $v\in V$ such that $(G,\ell)\models \varphi(v)$, in $\cO(\log n)$ rounds in the \CONGEST\ model.
\end{theorem}

Our algorithm runs on labeled graphs \((G,\ell) \in \mathcal G[\Lambda]\). Initially, every node $v\in V$ of an $n$-node labeled graph $(G,\ell)\in \mathcal{G}[\Lambda]$ is only given its identifier $\id(v)$ on $\mathcal{O}(\log n)$ bits, which is unique in the graph, and its labels \(\ell(u)\in 2^\Lambda\) (in addition to a polynomial upper bound on the number $n$ of vertices, as specified in the \CONGEST\/ model).
At the end of the execution of our algorithm, every node $v$ outputs \(\top\) if \((G,\ell) \models \varphi(v)\), and \(\bot\) otherwise.

We prove Theorem~\ref{thm:thmFOlocal}  by adapting the (centralized) quantifier elimination technique of~\cite{PilipczukST18} (see also \cite{DvorakKT13,GroheK09}) to the distributed setting. Let $(G,\ell)$ be a labeled graph with $G \in \mathcal{G}$ and $\ell:V\to 2^\Lambda$, on which we aim at checking a local formula $\varphi(x)$ at every node. Quantifier elimination proceeds by induction on the depth of the formula. The formula considered at any stage of the induction may thus have several free variables~$\overline{x}$ with $x_1=x$. We therefore describe a distributed algorithm that, in $\cO(\log n)$ rounds, transforms both the formula $\varphi(\overline{x})$ and the node-labeling  $\ell$ into a new formula $\widehat{\varphi}(\overline{x})$ and a new node labeling $\widehat{\ell}$ such that:
\begin{itemize}
    \item $\widehat{\varphi}(\overline{x})$ is quantifier-free, but uses new unary and binary predicates, 
    
    \item $\widehat{\ell}$ uses a new set of labels $\widehat{\Lambda}\supseteq \Lambda$ for the nodes of~$G$, using additional labels involved in the new unary and binary predicates, and
    
    \item the transformation preserves locality, and, for any tuple $\overline{v}$ of $|\overline{x}|$ vertices of~$G$, we have: $$(G,\ell) \models \varphi(\overline{v}) \iff (G,\widehat{\ell}) \models  \widehat{\varphi}(\overline{v}).$$  
\end{itemize}
Note that, while the transformation preserves locality, it might be the case that \(\varphi\) is \(r\)-local for some~$r$, where as \(\widehat\varphi\) is \(r'\)-local for some \(r'\geq r\). Quantifiers are eliminated one by one, and the main difficulty consists in the elimination of the existential quantifiers. As in~\cite{PilipczukST18}, the proof works in three steps. 

\begin{itemize}
    \item First, the transformation from $(\varphi, \ell)$ to $(\widehat{\varphi},\widehat{\ell})$ is performed on rooted forests of constant depth (cf. Subsection~\ref{ss:local_forests}). This is the most technical part of the quantifier elimination process. 

    \item Second, the quantifier elimination is performed on graphs of bounded treedepth (cf. Subsection~\ref{ss:local_treedepth}), by reducing the analysis of graphs of bounded treedepth to the analysis of rooted forests. This is where the set of labels is enriched. 

    \item Third, the analysis of graphs of bounded expansion is reduced to the analysis of graphs of bounded treedepth thanks to a $(p,f(p))$-treedepth coloring (cf. Subsection~\ref{ss:local_bexp}), for an appropriate~$p$. This is the steps for which the preservation of the locality property needs particular care. 
\end{itemize}

We emphasis that, at each step of the transformation of the original formula~$\varphi$, and the original labeled graph $(G,\ell)$, the computations must be efficiently performed in the \CONGEST\ model. Typically, one needs to make sure that adding new labels to the graph can be done efficiently in \CONGEST. The rooted forests involved in Subsection~\ref{ss:local_forests}, as well as the graphs of bounded treedepth involved in Subsection~\ref{ss:local_treedepth}, are subgraphs of the communication graph~$G$, and thus all communications passing through the edges of the former can be implemented along the edges of the latter.

\subsection{Rooted Forests of Bounded Depth}
\label{ss:local_forests}

In this section we adapt the quantifier elimination procedure from \cite{PilipczukST18} for rooted forests of constant depth (i.e., every tree in the forest is rooted and has constant depth) to the case of local formulas, and show how to implement this procedure in \CONGEST.
The basic idea for eliminating the quantifiers is to encode information
about the satisfiability of the formula in node-labels. This approach does increase the size of the
label set. Nevertheless, by taking advantage of the simple structure of FO formulas on forests of constant
depth, it is possible to bound the number of labels by a constant. This constant depends solely on the size of the formula, and on the depth of the trees.

As the locality property focuses on distances between nodes, the diameter of each tree plays
an important role. Therefore, instead of considering forests with depth at most \(d\) for some~$d$,
we consider forests in which each tree has diameter at most \(d-1\) for some~$d$.
For every \(d \in \N\), let us denote by $\mathcal{F}_d$ the set of forests in which each tree has diameter at most~\(d-1\).
Note that, in particular, all forests of \(\mathcal F_d\) have depth at most \(d\).
For a set of labels $\Lambda$, let us denote by $\mathcal{F}_d[\Lambda]$ the set of all $\Lambda$-labeled
graphs $(G,\ell)$ such that $G\in \mathcal{F}_d$.

\subsubsection{First-Order Formulas on Forests of Bounded Depth}\label{sec:forests}

Let $\Lambda$ be a (finite) set of labels, and let $d$ be a positive integer. It is convenient to consider a specific way of encoding formulas, suitable for $\mathcal{F}_d[\Lambda]$. For each $i\in \{-1,0,1,\dots, d-1\}$, we define the predicate $\lca_i(x,y)$ (for least common ancestor) as the predicate with free variables $x$ and $y$ that holds if the path $P_x$ from $x$ to the root of the tree containing $x$, and the path $P_y$ from $y$ to the root of the tree containing $y$ share exactly $i+1$ nodes. Moreover, we set $\lca_{-1}(x,y)$ to be true whenever $x$ and $y$ are in different trees. Note that $\lca_i(x,x)$ holds if and only if $x$ is a node at depth~$i$ in its tree. In particular, $\lca_0(x,x)$ holds if and only if $x$ is a root. The predicates $\lca_i$ are called \emph{$\lca$-predicates}. As opposed to the label predicates, for which it may be the case that $\lab_\lambda(x)$ and $\lab_{\lambda'}(x)$ both hold for two different labels $\lambda\neq\lambda'$, we have that, for every $x$ and~$y$, $\lca_i(x,y)$ holds for exactly one index $i\in \{-1,0,1,\dots, d-1\}$. 

A formula $\varphi\in FO[\Lambda]$ is said to be $\lca$-expressed if it contains solely predicates of the form $\lca_i(x,y)$ ($\lca$ predicates), and $\lab_\lambda(x)$ (label predicates). In other words, in a $\lca$-expressed formula, there are no predicates of type $\adj(x,y)$ or $x=y$. The following is a direct consequence of the definitions. 

\begin{lemma}\label{rem:fromFOtolca}
Every formula $\varphi\in \FO[\Lambda]$ to be evaluated in $\mathcal{F}_d[\Lambda]$ can be $\lca$-expressed by replacing every equality predicate $x=y$ by 
$$
\bigvee_{i\in[0,d-1]} \Big (\lca_i(x,x) \land \lca_i(x,y) \land \lca_i(y,y)\Big),
$$
and every adjacency predicate $\adj(x,y)$ by 
$$
\bigvee_{i\in [0,d-2]}
\Big(
\lca_i(x,y) 
\land 
\big( 
(\lca_i(x,x) \land \lca_{i+1}(y,y)) 
    \lor 
(\lca_{i+1}(x,x) \land \lca_i(y,y)) 
\big) 
\Big). 
$$ 
\end{lemma}

Given a fixed set $\Lambda$ of labels, a formula $\varphi \in \FO[\Lambda]$ is called \emph{$\lca$-reduced} if (1)~it is $\lca$-expressed, and (2)~it is quantifier-free. In other words, an $\lca$-reduced formula is a Boolean combination of label and $\lca$ predicates on free variables. 

\begin{definition}\label{def:types}
  Let $\overline{x} = (x_1, \dots, x_k)$ be a $k$-tuple of variables. For any two functions $\gamma: [k]\rightarrow 2^\Lambda$, and $\delta:[k]\times[k] \rightarrow [-1,d-1]$, let us define the  formula 
$$
\type_{\gamma, \delta}(\overline{x}) \in \FO[\Lambda]
$$
as the $\lca$-reduced formula defined by the conjunction of the following two predicates:
$$
\mbox{\rm Predicate 1: }
\forall i \in \{1,\dots, k\}, \; \bigwedge_{\lambda \in \gamma(i)} \lab_\lambda(x_i) \wedge \bigwedge_{\lambda \in \Lambda \smallsetminus \gamma(i)} \neg  \lab_\lambda(x_i),
$$
and
$$
\mbox{\rm Predicate 2: }
\forall (i,j) \in \{1,\dots, k\}^2, \; \lca_{\delta(i,j)}(x_i,x_j).
$$
We call such predicate $\lca$-type, and we denote by $\Type(k,d,\Lambda)$ the set of all $\lca$-type predicates with $k$ variables. 
  \end{definition}

Observe that the size of $\Type(k,d,\Lambda)$ is at most $2^{k|\Lambda|} (d+1)^{k^2}$, and that it can be computed from $k,d$ and~$\Lambda$. Note also that, given $\gamma$ and $\delta$, it may be the case that  the formula $\type_{\gamma, \delta}(\overline{x})$ is always false. This is typically the case when the function $\delta$ yields a conjunction of incompatible $\lca$ predicates (e.g., $\delta(i,j) = 5$, $\delta(i,i)=2$, and $\delta(j,j) = 1$, or merely any function $\delta$ such that $\delta(i,j) \neq \delta(j,i)$). In such cases, i.e., when $\type_{\gamma, \delta}(\overline{x})$ is always false, we say that the formula is \emph{trivial}. 

For any non-trivial $\lca$-type formulas $\varphi(\overline{x}) = \type_{\gamma,\delta}(\overline{x})$, we have that $\delta$ and $\gamma$ induce a labeled forest $(F_\varphi, \ell_\varphi) \in \mathcal{F}_d[\Lambda]$, as follows. Each $i\in [k]$ is identified with a node $u_i$ in $F_\varphi$, and each leaf of $F_\varphi$ belongs to the $k$-tuple $\overline{u}_\varphi = (u_1, \dots, u_k)$. Furthermore, for every $(i,j)\in [k]\times [k]$, $\lca_{\delta(i,j)}(u_i, u_j)$ holds in $F_\varphi$. Finally, for every $i\in [k]$, $\ell_\varphi(u_i) = \gamma(i)$, and $\ell_\varphi(v) = \emptyset$ whenever $v \notin \overline{u}_\varphi$.

For any labeled forest $(F,\ell)$ with $F=(V,E)$, any integer~$k$, and any set of vertices  $\overline{v}\in V^k$, let us denote by $F[\overline{v}]$ the subgraph of $F$ induced by all the nodes in $\overline{v}$, and their ancestors.

\begin{lemma}\label{prop:isomorphism}
 Let $(F,\ell) \in \mathcal{F}_d[\Lambda]$ with $F=(V,E)$, and let $\varphi(\overline{x})=\type_{\gamma,\delta}(\overline{x})$  be a non-trivial $\lca$-type formula with $k$ free variables. For every $\overline{v}=(v_1,\dots,v_k) \in V^k$, $(F,\ell) \models \varphi(\overline{v})$ if and only if 
 there exists a graph isomorphism $f$ between $F[\overline{v}]$ and $F_\varphi$ such that, for every $i\in[k]$,  $\ell_\varphi(f(v_i)) = \ell(v_i)$.   
\end{lemma}

\begin{proof}
Recall that $\overline{u}_\varphi = (u_1,\dots, u_k)$ includes the leaves of $F_\varphi$. Let  $\overline{v} = (v_1, \dots, v_k)$ be a set of $k$ vertices of $F$. Suppose first that $(F,\ell) \models \varphi(\overline{v})$. Then, for every $(i,j)\in [k]\times [k]$, 
\[
\gamma(i) = \ell(v_i), \text{ and } \lca(v_i, v_j) = \delta(i,j).
\]
Let $f$ be the function that, for every $i\in[k]$, $f(v_i)=u_i$, and, for every $r\in \{1,\dots, \delta(i,i)-1\}$, $f$ maps the ancestor of $v_i$ in $F$ at distance~$r$ from $v_i$ to the corresponding ancestor of $u_i$ in~$F_\varphi$. By construction $f$ is an isomorphism between $F[\overline{v}]$ and $F_\varphi$. Moreover, for every $i\in[k]$, we do have $\ell_\varphi(f(v_i)) = \ell_\varphi(u_i) = \gamma(i) = \ell(v_i)$, as claimed. 

Conversely, let us assume that there exists a graph isomorphism $f$ between $F[\overline{v}]$ and $F_\varphi$ such that, for every $i\in[k]$,  $\ell_\varphi(f(v_i)) = \ell(v_i)$.
Then we directly infer that, for every  $(i,j)\in[k]\times [k]$, $\gamma(i)= \ell(v_i)$ and $\lca(v_i,v_j)=\delta(i,j)$. This implies that $(F,\ell)\models \varphi(\overline{v})$, as claimed.
\end{proof}

The following statement establishes that every $\lca$-reduced formula can be expressed as a disjunction of $\lca$-types.

\begin{lemma}\label{prop:basicnormal}
For every $\lca$-reduced formula  $\varphi(\overline{x})\in \FO[\Lambda]$, there exists a set $I\subseteq \Type(|\overline{x}|,d, \Lambda)$, and a formula 
$$\varphi'(\overline{x}) = \bigvee_{\psi \in I} \psi(\overline{x})$$ such that, for every   $(F,\ell)\in \mathcal{F}_d[\Lambda]$, $\true(\varphi,F,\ell) = \true(\varphi',F,\ell)$.
Moreover, $\varphi'(\overline{x})$ can be computed given only $d$, $\Lambda$, and~$\varphi$. The formula $\varphi'(\overline{x})$ is called the \emph{basic normal form} of $\varphi$.
\end{lemma}

\begin{proof}
Let  $\varphi(\overline{x})$ be an $\lca$-reduced formula with $k$ free variables, and let $(F,\ell) \in \mathcal{F}_d[\Lambda]$. Let  $\overline{v} = (v_1, \dots, v_k)$ be a tuple of nodes such that $(F,\ell) \models \varphi(\overline{v})$. Let us define the function 
$$\delta_{\overline{v}}:[k]\times[k]\rightarrow [-1,d-1]$$ 
by $\delta_{\overline{v}}(i,j) = h$ if the path from $v_i$ to the root of the tree  of $F[\overline{v}]$ containing $v_i$, and the path from $v_j$ to the root of the tree of $F[\overline{v}]$ containing $v_j$ share exactly $h+1$ nodes. We also define the function $\gamma_{\overline{v}}$ by $$\gamma_{\overline{v}}(i)=\ell(v_i)$$ for every $i\in[k]$. We obtain that 
$$(F,\ell) \models \type_{\gamma_{\overline{v}},\delta_{\overline{v}}}(\overline{v}).$$
Now let $I\subseteq \Type(k,d,\Lambda)$ be the set of all possible $\lca$-types $\type_{\gamma_{\overline{v}},\delta_{\overline{v}}}$ that can be defined from a labeled forest $(F,\ell)\in \mathcal{F}_d$, and a tuple of vertives $\overline{v}\in \true(\varphi,F,\ell)$.
By construction, the function $$\varphi'(\overline{x}) = \bigvee_{\psi \in I} \psi(\overline{x})$$ satisfies $\true(\varphi,F,\ell) = \true(\varphi',F,\ell)$ for every $(F,\ell)\in \mathcal{F}_d[\Lambda]$.
\end{proof}

\paragraph{Example.} Let us consider the following formula using labels in $\{0,1\}$, hence the two predicates $\lab_0$ and~$\lab_1$: 
$$
\varphi(x_1, x_2,x_3) = \adj(x_1,x_2) \wedge \adj(x_1,x_3) \wedge \lab_1(x_1) \wedge \lab_0(x_2))\wedge \lab_0(x_3),$$ which holds when $x_1$ is adjacent to~$x_2$, $x_1$ is labeled~$1$, and $x_2$ and $x_3$ are both labeled~$0$.
This formula can be rewritten as the $\lca$-reduced formula
$$
\tilde{\varphi}(x_1, x_2, x_3) = \left( \bigvee_{i\in [0,d-2]}  \xi^1_i(x_1,x_2,x_3)\right) \vee \left(\bigvee_{i\in [0,d-3]} \xi_i^2(x_1,x_2,x_3)\right) \vee \left(\bigvee_{i\in [0,d-3]} \xi_i^3(x_1,x_2,x_3) \right)
$$
where
\begin{align*}
     \xi^1_i(x_1,x_2,x_3) = \bigg(& \lca_i(x_1,x_1) \wedge \lca_{i+1}(x_2,x_2) \wedge \lca_{i+1}(x_3,x_3) \wedge\\& \lca_i(x_1, x_2)  \wedge \lca_i(x_1,x_3) \wedge \lca_i(x_2,x_3) \wedge\\&
     \lab_1(x_1) \wedge \lab_0(x_2) \wedge  \lab_0(x_3)
     \bigg),
\end{align*}
\begin{align*}
     \xi^2_i(x_1,x_2,x_3) = \bigg(& \lca_{i+1}(x_1,x_1) \wedge \lca_i(x_2,x_2) \wedge \lca_{i+2}(x_3,x_3) \wedge\\& \lca_i(x_1, x_2) \wedge \lca_{i+1}(x_1,x_3)  \wedge \lca_i(x_2,x_3) \wedge\\&
     \lab_1(x_1) \wedge \lab_0(x_2) \wedge  \lab_0(x_3)
     \bigg),
\end{align*}
and
\begin{align*}
     \xi^3_i(x_1,x_2,x_3) = \bigg(& \lca_{i+1}(x_1,x_1)\wedge \lca_{i+2}(x_2,x_2) \wedge \lca_i(x_3,x_3) \wedge\\& \lca_{i+1}(x_1, x_2) \wedge \lca_i(x_1,x_3) \wedge \lca_i(x_2,x_3)\wedge\\&
     \lab_1(x_1) \wedge \lab_0(x_2) \wedge  \lab_0(x_3)
     \bigg).
\end{align*}
Observe that $\tilde{\varphi}$ is the basic normal form of $\varphi$. Indeed, for every $i\in [0,d-2]$, $\xi_i^1$ is an $\lca$-type formula, and, for every $i\in [0,d-3]$, $\xi_i^2$ and $\xi_i^3$ are $\lca$-type formulas. For instance, $\xi_i^1 = \type_{\gamma, \delta}$ with 
$$
\gamma(1) = \{1\},~ \gamma(2) = \{0\},~ \gamma(3) =\{0\},
$$ 
and with 
$$
\delta(1,1) = i,~ \delta(2,2) =i+1,~ \delta(3,3) = i+1,~ \delta(1,2) = i,~ \delta(1,3) = i \text{ and } \delta(2,3)=i.
$$ 
\Cref{fig:example3} depicts $F_{\xi^1_i}$,  $F_{\xi^2_i}$, and $F_{\xi^3_i}$. If a labeled forest $(F,\ell)$ with label set $\{0,1\}$ satisfies $\varphi$, then $F$ contains a forest isomorphic to $F_{\xi^1_i}$,  $F_{\xi^2_i}$ or $F_{\xi^3_i}$, for some~$i\in [0,d-2]$.

\begin{figure}[ht]
    \centering
    \scalebox{0.75}{

\tikzset{every picture/.style={line width=0.75pt}} 

\begin{tikzpicture}[x=0.75pt,y=0.75pt,yscale=-1,xscale=1]

\draw [line width=1.5]    (110,150) -- (110,230) ;
\draw [shift={(110,230)}, rotate = 270] [color={rgb, 255:red, 0; green, 0; blue, 0 }  ][line width=1.5]    (0,6.71) -- (0,-6.71)   ;
\draw [shift={(110,150)}, rotate = 270] [color={rgb, 255:red, 0; green, 0; blue, 0 }  ][line width=1.5]    (0,6.71) -- (0,-6.71)   ;
\draw [line width=1.5]    (300,150) -- (300,230) ;
\draw [shift={(300,230)}, rotate = 270] [color={rgb, 255:red, 0; green, 0; blue, 0 }  ][line width=1.5]    (0,6.71) -- (0,-6.71)   ;
\draw [shift={(300,150)}, rotate = 270] [color={rgb, 255:red, 0; green, 0; blue, 0 }  ][line width=1.5]    (0,6.71) -- (0,-6.71)   ;
\draw [line width=1.5]    (480,150) -- (480,230) ;
\draw [shift={(480,230)}, rotate = 270] [color={rgb, 255:red, 0; green, 0; blue, 0 }  ][line width=1.5]    (0,6.71) -- (0,-6.71)   ;
\draw [shift={(480,150)}, rotate = 270] [color={rgb, 255:red, 0; green, 0; blue, 0 }  ][line width=1.5]    (0,6.71) -- (0,-6.71)   ;

\draw  [line width=1.5]   (140, 230.5) circle [x radius= 15.91, y radius= 15.91]   ;
\draw (131,222.4) node [anchor=north west][inner sep=0.75pt]    {$x_{1}$};
\draw  [line width=1.5]   (108, 287.5) circle [x radius= 15.91, y radius= 15.91]   ;
\draw (99,279.4) node [anchor=north west][inner sep=0.75pt]    {$x_{2}$};
\draw  [line width=1.5]   (180, 289.5) circle [x radius= 15.91, y radius= 15.91]   ;
\draw (171,281.4) node [anchor=north west][inner sep=0.75pt]    {$x_{3}$};
\draw  [line width=1.5]   (142.5, 150) circle [x radius= 13.9, y radius= 13.9]   ;
\draw (136,142.4) node [anchor=north west][inner sep=0.75pt]    {$ $};
\draw (91,182.4) node [anchor=north west][inner sep=0.75pt]    {$i$};
\draw  [line width=1.5]   (330, 231.5) circle [x radius= 15.91, y radius= 15.91]   ;
\draw (321,223.4) node [anchor=north west][inner sep=0.75pt]    {$x_{2}$};
\draw  [line width=1.5]   (330, 281.5) circle [x radius= 15.91, y radius= 15.91]   ;
\draw (321,273.4) node [anchor=north west][inner sep=0.75pt]    {$x_{1}$};
\draw  [line width=1.5]   (330, 332.5) circle [x radius= 15.91, y radius= 15.91]   ;
\draw (321,324.4) node [anchor=north west][inner sep=0.75pt]    {$x_{3}$};
\draw  [line width=1.5]   (329.5, 150) circle [x radius= 13.9, y radius= 13.9]   ;
\draw (323,142.4) node [anchor=north west][inner sep=0.75pt]    {$ $};
\draw (281,182.4) node [anchor=north west][inner sep=0.75pt]    {$i$};
\draw  [line width=1.5]   (510, 228.5) circle [x radius= 15.91, y radius= 15.91]   ;
\draw (501,220.4) node [anchor=north west][inner sep=0.75pt]    {$x_{3}$};
\draw  [line width=1.5]   (510, 278.5) circle [x radius= 15.91, y radius= 15.91]   ;
\draw (501,270.4) node [anchor=north west][inner sep=0.75pt]    {$x_{1}$};
\draw  [line width=1.5]   (510, 329.5) circle [x radius= 15.91, y radius= 15.91]   ;
\draw (501,321.4) node [anchor=north west][inner sep=0.75pt]    {$x_{2}$};
\draw  [line width=1.5]   (509.5, 149) circle [x radius= 13.9, y radius= 13.9]   ;
\draw (503,141.4) node [anchor=north west][inner sep=0.75pt]    {$ $};
\draw (461,182.4) node [anchor=north west][inner sep=0.75pt]    {$i$};
\draw [line width=1.5]    (115.79,273.62) -- (132.21,244.38) ;
\draw [line width=1.5]    (171.07,276.33) -- (148.93,243.67) ;
\draw [line width=1.5]    (330,265.59) -- (330,247.41) ;
\draw [line width=1.5]    (330,316.59) -- (330,297.41) ;
\draw [line width=1.5]  [dash pattern={on 1.69pt off 2.76pt}]  (329.59,163.9) -- (329.9,215.59) ;
\draw [line width=1.5]  [dash pattern={on 1.69pt off 2.76pt}]  (142.07,163.89) -- (140.49,214.59) ;
\draw [line width=1.5]    (510,262.59) -- (510,244.41) ;
\draw [line width=1.5]    (510,313.59) -- (510,294.41) ;
\draw [line width=1.5]  [dash pattern={on 1.69pt off 2.76pt}]  (509.59,162.9) -- (509.9,212.59) ;

\end{tikzpicture}}
    
    \caption{$F_{\xi^1_i}$ (left), $F_{\xi_i^2}$ (middle), and  $F_{\xi^3_i}$ (right).}
    \label{fig:example3}
\end{figure}
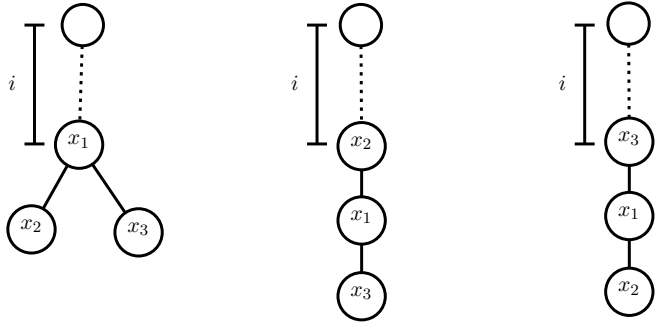

\subsubsection{Quantifier Elimination of $\lca$-Expressed Formulas}

In this section, we show that it is possible to eliminate existential quantifiers on $\lca$-expressed formulas,  by augmenting the set of labels.

\begin{definition}\label{def:reducibility}
    Let $\Lambda$ be a finite set of labels, and let $d$ be a positive integer.  A formula $\varphi(\overline{x})\in \FO[\Lambda]$ is said \emph{reducible on $\mathcal{F}_d[\Lambda]$} if there exists a set $\widehat{\Lambda}$, and an $\lca$-reduced formula $\widehat{\varphi}(\overline{x}) \in \FO[\widehat{\Lambda}]$ such that (1)~$\widehat{\Lambda}$ and $\widehat{\varphi}$ depend only on $\Lambda$, $\varphi$, and~$d$, and (2)~for every $(F,\ell)\in \mathcal{F}_d[\Lambda]$, there exists a $\widehat{\Lambda}$-labeling $\widehat{\ell}$ of $F$ such that 
    $$
    \true(\varphi,F,\ell) = \true(\widehat{\varphi},F,\widehat{\ell}).
    $$
\end{definition}

Each element in the triplet $(\widehat{\Lambda},\widehat{\varphi},\widehat{\ell})$ of \Cref{def:reducibility} is called the \emph{reduction on $\mathcal{F}_d[\Lambda]$} of $\Lambda$, $\varphi$, and $\ell$, respectively. We define a similar notion for local formulas.

\begin{definition}\label{def:local-reducibility}
    Let \(\Lambda\) be a finite set of labels, and let \(d\) be a positive integer. A formula
    \(\varphi(\overline x) \in FO[\Lambda]\) is said \emph{\(r\)-local reducible on
    \(\mathcal F_d[\Lambda]\)} if there exists a set $\widehat{\Lambda}$, and an $\lca$-reduced formula $\widehat{\varphi}(\overline{x}) \in \FO[\widehat{\Lambda}]$ such that (1)~$\widehat{\Lambda}$ and $\widehat{\varphi}$ depend only on $\Lambda$, $\varphi$, and~$d$, and (2)~for every $(F,\ell)\in \mathcal{F}_d[\Lambda]$, there exists a $\widehat{\Lambda}$-labeling $\widehat{\ell}$ of $F$ such that 
    $$
    \true(\varphi,F,\ell) = \true(\widehat{\varphi},F,\widehat{\ell})
    = \true(\widehat\varphi\wedge\ball(\overline x,x_1,r), F, \widehat \ell).
    $$
\end{definition}

In Definition~\ref{def:local-reducibility}, one insists on the fact that the reduction $\widehat{\varphi}$ of \(\varphi\) must be \(r\)-local. This is called
a \emph{\(r\)-local reduction}.

\paragraph{Remark.} The fact that both Definitions~\ref{def:reducibility} and~\ref{def:local-reducibility} require $\widehat{\Lambda}$ and $\widehat{\varphi}$ to depend only on $\Lambda$, $\varphi$ and~$d$ is not crucial for deciding (local) formulas, but will be used for the design of our certification scheme later in the paper (the prover gives $\widehat{\Lambda}$ and $\widehat{\varphi}$ to the nodes as constant-size certificates).
\medskip

The following lemma is at the core of quantifier elimination.

\begin{lemma}\label{lem:elimination_forest}
Every formula $\varphi(\overline{x})= \exists \overline{y} \; \zeta(\overline{x},\overline{y})$ in $\FO[\Lambda]$, where $\zeta$ is an $\lca$-reduced formula, is reducible on $\mathcal{F}_d[\Lambda]$.
\end{lemma}

\begin{proof}
Without loss of generality, let us assume that $|x|=k\geq 1$. (In the case where $\varphi$ has no free-variables, we can create a formula $\varphi(z) = \varphi$ for a dummy variable $z$.) We only deal with the case where  $|y|=1$, i.e., with the elimination of a single quantified variable. Indeed, for the case where $|y|>1$ we can iteratively eliminate the quantifiers one by one, until we obtain a quantifier-free (actually an $\lca$-reduced) formula. 
Using \Cref{rem:fromFOtolca} we can assume without loss of generality that  $\zeta$ is $\lca$-reduced. Moreover,  thanks to  \Cref{prop:basicnormal}, we can also assume that $\zeta(\overline{x}, y)$ is expressed in the basic normal form. That is, there exist $I\subseteq \Type(k+1,d,\Lambda)$ such that $\zeta(\overline{x},y) = \bigvee_{\psi \in I} \psi(\overline{x},y)$. Since the existential quantifier commutes with the disjunction, we get that 
$$
\varphi(\overline{x}) = \bigvee_{\psi \in I} \exists y \; \psi(\overline{x}, y).
$$ 
Let us fix an arbitrary labeled graph $(F,\ell) \in \mathcal{F}_d[\Lambda]$. Let us first assume that each formula $\psi \in I$ is reducible on $\mathcal{F}_d[\Lambda]$ (we show how to do this reduction later in the proof). Under this assumption, let us denote by $\widehat{\Lambda}^\psi$, $\widehat{\psi}$, and $\widehat{\ell}^\psi$ the resulting reductions. We can relabel every set $\Lambda^\psi$, for $\psi \in I$, so that the sets $\widehat{\Lambda}^{\psi} \smallsetminus \Lambda$ are pairwise disjoint. Let us then define:
$$
\widehat{\Lambda} = \bigcup_{\psi \in I} \widehat{\Lambda}^\psi, \;
\widehat{\varphi} = \bigvee_{\psi \in I} \widehat{\psi},
\text{ and, for every node $u\in F$, } 
\; \widehat{\ell}(u) = \bigcup_{\psi \in I} \widehat{\ell}^\psi(u). 
$$
We get that $\widehat{\Lambda}$, $\widehat{\varphi}$, and $\widehat{\ell}$ are reductions on $\mathcal{F}_d[\Lambda]$ of $\Lambda$, $\varphi$, and $\ell$, respectively. Therefore, it just remains to prove that every $\psi \in \Type(k+1,d,\Lambda)$ is reducible. 

Let $\gamma$ and $\delta$ be such that  $\psi(\overline{x},y) = \type_{\gamma, \delta}(\overline{x},y) \in \Type(k+1, d, \Lambda)$, where $y$ is identified with variable numbered $k+1$ in the definitions of $\gamma$ and $\delta$. Let $s\in [k]$ such that $\delta(s,k+1)$ is maximum (this value exists because $\varphi$ has at least one free variable). We define 
$$
h = \delta(s,k+1), \;  h_{s} = \delta(s,s), \text{ and } h_y = \delta(k+1,k+1).
$$
In other words, $x_s$ is at depth $h_s$, $y$ is of depth $h_y$, and the least common ancestor of $x_s$ and $y$ (if exists) is at depth $h$ in $F_\psi$. We consider three cases, depending on the value of $\delta(s,k+1)$, and on the relative values of $h_y$ and $h$. 

\paragraph{Case 1: $\delta(s,k+1)\geq 0$ and $h_y = h$.}  

In this case $y$ is an ancestor of $x_s$ in $F_\psi$. Let $(F,\ell)\in \mathcal{F}_d[\Lambda]$. A node $w$ at depth $h$ is called a \emph{candidate} if $\ell(w)= \gamma(k+1)$. A node at depth $h_s$ having an ancestor that is a candidate is called a \emph{good node}.  Let $\widehat{\Lambda}^\psi = \Lambda \cup \{ \good\}$. We define a $\widehat{\Lambda}^\psi$-labeling $\widehat{\ell}^\psi$ of $F$ as follows: 
\[
\widehat{\ell}^\psi(u) = \left\{
\begin{array}{ll}
\ell(u) \cup \{\good\} & \mbox{if $u$ is a good node, }\\
\ell(u) & \mbox{otherwise.}
\end{array}
\right.
\]
 
Let us now construct an $\lca$-reduced formula $\widehat{\psi}$ that is equivalent to $\psi$. Let us denote by $\widehat{\gamma}: [k]\rightarrow 2^\Lambda$ and  $\widehat{\delta}: [k]\times[k] \rightarrow [-1,d-1]$ the restrictions of $\gamma$ and $\delta$ to $[k]$ (i.e., for each $i,j\in [k]$, $\widehat{\gamma}(i) = \gamma(i)$ and $\widehat{\delta}(i,j)= \delta(i,j)$). Let $\xi =  \type_{\widehat{\gamma},\widehat{\delta}}$. Since $h_y= h$, we have that $F_\xi=F_\psi$, and $\ell_\xi(\overline{u}_\xi) = \ell_\psi(\overline{u}_\psi)$. 
Now, given a tuple of vertices $\overline{v}= (v_1,\dots, v_k)$, we have that $(F,\ell)\models \exists y, \psi(\overline{v},y)$ if and only if the following conditions hold:
 \begin{itemize}
     \item $F[\overline{v}]$ is isomorphic to $F_\psi$ (which equals $F_\xi$), 
     \item for every $i\in[k]$, $\ell(v_i) = \gamma(i) = \widehat{\gamma}(i)$, and
     \item the ancestor of $v_s$ at level $h$ in $F[\overline{v}]$ is labeled $\gamma(k+1)$, i.e., $v_s$ is a good node. 
 \end{itemize}
From \Cref{prop:isomorphism} we obtain that the first two conditions imply that $(F,\ell) \models \psi(\overline{v})$.  The third condition can be expressed as $\good \in \ell(v_1)$. Therefore, $\widehat{\psi}$ is defined as the conjunction of the two previous expressions, that is,  $$\widehat{\psi}(\overline{x}) = \type_{\widehat{\gamma},\widehat{\delta}}(\overline{x}) \land \lab_{\good}(x_s). $$
Therefore, $\true(\psi,F,\ell) = \true(\widehat{\psi},F, \widehat{\ell}^\psi)$  for every $\Lambda$-labeled forest $(F,\ell)$ of depth at most $d$.

\paragraph{Case 2: $\delta(s,k+1)\geq 0$ and $h_y > h$.} 

Let $(F,\ell)\in \mathcal{F}_d[\Lambda]$. In this case, a node $v$ at depth $h_y$ is called a \emph{candidate} if $\ell(v) = \gamma(k+1)$. A node $v$ at depth greater than $h$ is called a $\emph{good node}$ if at least one of its descendants is a candidate node. Finally, a node $v$ at depth $h$ is called a \emph{pivot}, and we denote by $\kappa(v)$ the number of children of $v$ that are good nodes. 

We define $\widehat{\gamma}$,  $\widehat{\delta}$, and $\xi = \type_{\widehat{\gamma}, \widehat{\delta}}$ the same way than in Case~1. Under the assumptions of Case~2, $F_\xi$ is however not necessarily equal to $F_\psi$. In fact $F_\xi$ is a sub-forest of $F_\psi$. Indeed, $F_\psi$ can be obtained from $F_\xi$ by adding a path that connects a pivot node in $F_\psi$ to a candidate node not in~$F_\psi$. 
 
Given a tuple $\overline{v}$ of $k$ vertices, checking whether $F[\overline{v}]$ is isomorphic to $F_\xi$  can be achieved by merely using formula~$\xi$. Let $z$ be the pivot  in $F[\overline{v}]$, that is, $z$ is the ancestor of $v_s$ at level $h$. Let $\kappa(z,\overline{v})$ be the set of good children of $z$ contained in $F[\overline{v}]$.  To eliminate the existential quantifier in $\varphi(\overline{x})$,  it is sufficient to check whether one can extend $F[\overline{v}]$ by plugging a path from $z$ to a candidate node not contained in $F[\overline{v}]$. This condition holds if $\kappa(z,\overline{v})<\kappa(z)$, and it is exactly what is checked when considering the formula~$\widehat{\psi}$. 

Let us define $\widehat{\Lambda}^\psi = \Lambda \cup [0,k+1]\cup \{\good\}$, and let us consider the $\widehat{\Lambda}^\psi$-labeling $\widehat{\ell}^\psi$ of $F$ defined as follows. Initially, for every node $u$ of $F$, $\widehat{\ell}^\psi(u)=\ell(u)$. We then update the label $\widehat{\ell}^\psi(u)$ of every node $u$ as follows. 
\begin{itemize}
    \item If $u$ is a node at depth $h_s$, then let us denote $v$ the ancestor of $u$ at depth $h$. The label $\lambda$ is added to $\widehat{\ell}^\psi(u)$ where $$\lambda = \begin{cases}
        \kappa(v) & \text{ if } \kappa(v)\leq k\\
        k+1 & \text{otherwise}.
    \end{cases}$$
    \item If $u$ is a good node, the label $\good$ is added to $\widehat{\ell}^\psi(u)$. 
\end{itemize}
Observe that the two conditions above are not necessarily exclusive, and it may be the case that both labels $\lambda$ and $\good$ are added to the set of labels of a the same node~$u$. 

Let us suppose that there exists an $\lca$-reduced $\widehat{\Lambda}^\psi$-formula $\alpha(\overline{x})$ that checks whether $\kappa(z)>\kappa(z,\overline{x})$ where $z$ is the ancestor at level $h$ of $x_s$. We can then define  $\widehat{\psi}$ as the $\lca$-reduced $\widehat{\Lambda}^\psi$-formula 
$$
\widehat{\psi}(\overline{x}) = \type_{\widehat{\gamma},\widehat{\delta}}(\overline{x}) ~\wedge ~ \alpha(\overline{x}), 
$$ 
where

$$ 
\alpha(\overline{x}) = \bigvee_{q\in [p]} \bigg[\bigg(\bigvee_{q\leq m\leq k+1}\lab_{m}(x_s)\bigg)\wedge  \bigvee_{S \in K(p,q-1)}  \bigg(\bigwedge_{i \in S} \bigvee_{j\in X_i} \lab_{\good}(x_j)  \wedge\bigwedge_{i\notin S} \bigwedge_{j\in X_i}\neg \lab_{\good}(x_j) \bigg)\bigg].
$$

We claim that $\true(\psi,F,\ell) = \true(\widehat{\psi},F, \widehat{\ell}^\psi)$ for every $\Lambda$-labeled forest $(F,\ell) \in \mathcal F_d$. Indeed, observe first that $\kappa(z, \overline{x})\leq k$ as the forest $F_\xi$ has degree at most $k$ (because it has at most $k$ leaves). Therefore, $\alpha(\overline{x})$ holds whenever $\kappa(z)>k$, which is equivalent to the predicate $\lab_{k+1}(x_s)$. If $\kappa(z)\leq k$ then the value of $\kappa(z)$ is stored in the label of $x_s$. Let us denote by $z_1, \dots, z_p$ the $p\leq k$ children of $z$ in $F_\psi$. For every $i\in [p]$, let us denote by $X_i$ the set of variables in $\overline{x}$ that are descendants of~$z_i$. For every $q\in [p]$, we also denote by $K(p,q)$ the set of all subsets of $[p]$ of at most $q$ elements. If $q=\kappa(z)$ and $t=\kappa(z,\overline{x})$, then checking $q>t$ can be done by checking the existence of set $S \in K(p,q-1)$ satisfying that, for every $i\in S$, $X_i$ contains a good node, and, for every $i\notin S$, $X_i$ does not contain a good node. 

\paragraph{Case 3: $\delta(s, k+1)=-1$.}  

In this case, the node $y$ in  $F_\psi$ has no common ancestors with any of the nodes $x_1,\dots, x_k$. In other words, $y$ is in a tree of $F_\psi$ that is different from any of the trees containing $x_1, \dots, x_k$. Let $(F, \ell)\in \mathcal{F}_d[\Lambda]$. Like in previous cases, a node $u$ at depth $h_y$ is called a \emph{candidate} if $\ell(u)=\gamma(k+1)$. A root $r$ of $F$ is said to be \emph{active} if some descendant of $r$ is a candidate node. A node is called \emph{good} if it is a descendant of an active root. 

Let $\widehat{\gamma}$, $\widehat{\delta}$ and $\xi = \type_{\widehat{\gamma},\widehat{\delta}}$ be defined as in Case~1.  In Case~3,  $F_\xi$ is a sub-forest of $F_\psi$, where $F_\psi$ can be obtained from $F_\xi$ by adding a path from an active root $r$ not contained in $F_\xi$ to a candidate node.

Let $\overline{v} = (v_1, \dots, v_k)$ be a set of vertices of $F$, and let $\rho(F)$ be the number of active roots in $F$. Similarly, let $\rho(F,\overline{v})$ be the number of active roots in $F[\overline{v}]$. It follows that  $(F,\ell)\models \exists y, \psi(\overline{v},y)$ if and only if $F[\overline{v}]$ is isomorphic to $F_\xi$,  $\ell(v_i) = \gamma(i)$  for each $i\in[k]$, and $F$ contains an active root not contained in $F[\overline{v}]$. The first two conditions can be checked using $\xi$, and the latter is equivalent to $\rho(F) > \rho(F,\overline{v})$.

Let $\widehat{\Lambda} = \Lambda \cup [0,k+1]\cup \{\good\}$ be a set of labels, and let us consider the $\widehat{\Lambda}^\psi$-labeling $\widehat{\ell}^\psi$ of $F$ defined as follows. Initially, for every node $u$ in $F$, let us set $\widehat{\ell}^\psi(u)=\ell(u)$. We then update the label $\widehat{\ell}^\psi(u)$ of every node $u$ as follows. 
\begin{itemize}
    \item For every node $u$,  the label $\lambda$ is added to $\widehat{\ell}^\psi(u)$, where $$\lambda = \begin{cases}
        \rho(F) & \text{ if } \rho(F)\leq k\\
        k+1 & \text{ if } \rho(F) > k.
    \end{cases}$$
    \item If $u$ is a good node, then the label $\good$ is added $\widehat{\ell}^\psi(u)$. 
\end{itemize}

Let us suppose that there exists an $\lca$-reduced $\widehat{\Lambda}^\psi$-formula $\alpha(\overline{x})$ that checks whether $\rho(F)>\rho(F,\overline{x})$ where $z$ is the ancestor at level $h$ of $x_1$. We can  then define $\widehat{\psi}$ as the $\lca$-reduced $\widehat{\Lambda}^\psi$-formula 
$$
\widehat{\psi}(\overline{x}) = \type_{\widehat{\gamma},\widehat{\delta}}(\overline{x}) ~\wedge ~ \xi(\overline{x}), 
$$ 
where

\begin{equation}\label{eq:xi}
\xi(\overline{x}) =  \bigvee_{q\in [p]} \bigg[\bigg(\bigvee_{q\leq m\leq k+1}\lab_{m}(x_1)\bigg)\wedge  \bigvee_{S \in K(p,q-1)}  \bigg(\bigwedge_{i \in S} \bigvee_{j\in X_i} \lab_{\good}(x_j)  \wedge\bigwedge_{i\notin S} \bigwedge_{j\in X_i}\neg \lab_{\good}(x_j) \bigg)\bigg]
\end{equation}

We claim that $\true(\psi,F,\ell) = \true(\widehat{\psi},F, \widehat{\ell}^\psi)$ for every $\Lambda$-labeled forest $(F,\ell) \in \mathcal F_d$.  Indeed, observe that $F_\xi$ can contain at most $k$ roots. Therefore $\alpha$ is true if $\rho(F)>k$, which is equivalent the predicate $\lab_{k+1}(x_1)$. If $\rho(F)\leq k$ then the value of $\rho(F)$ is stored in the label of $x_1$. Let $r_1, \dots, r_p$ denote the  $p\leq k$ roots of $F_\xi$. For every $i\in [p]$, let $X_i$ be the set of variables in $\overline{x}$ that are descendants of~$r_i$. We obtain that checking whether $q=\rho(F) > \rho(F,\overline{x})$ is equivalent to check the existence of a set $S\in K(p,q-1)$ such that, for every $i\in[p]$,  $X_i$ contains a good node if and only if $i\in S$, from which the expression of $\xi$ follows.
\end{proof}

Finally, the following lemma is a refined version of \Cref{lem:elimination_forest} enforcing locality.

\begin{lemma}
\label{lem:local_elimination_forest}
Let \(d \in \N\), \(r\geq d\), and \(\varphi(\overline x) = \exists \overline y \, \big(
\zeta(\overline x, \overline y) \wedge \ball(\{\overline x,\overline y\},x_1,r)\big)\)
be an \(r\)-local \(\lca\)-expressed existential formula in \(\FO[\Lambda]\), i.e., 
\(\zeta\) is a \(\lca\)-reduced formula in \(\FO[\Lambda]\). It holds that
\(\varphi\) is \(d\)-local reducible on \(\mathcal F_d[\Lambda]\).
\end{lemma}

\begin{proof}
Since all trees in the forests \(F \in \mathcal F_d\) have diameter at most \(d\), and since \(r\geq d\), we get that 
\(\ball(z,x_1,r)\) holds if and only if \(x_1\) and \(z\) are in the same tree,
which is equivalent to \(\neg \lca_{-1}(x_1,z)\).
We define:
\[
\varphi_1(\overline x) = \exists \overline y, \zeta(\overline x, \overline y) \wedge
\bigwedge_{z \in \overline x, \overline y} \neg \lca_{-1}(x_1, z)
\]
We have \(\true(\varphi, F, \ell) = \true(\varphi_1, F, \ell)\).
We can now apply \Cref{lem:elimination_forest} on \(\varphi_1\).
(Note that there were three cases the proof of \Cref{lem:elimination_forest}, but the third case cannot happen here as this case corresponds to the case where two variables are not in the same connected components
of \(F\), which is forbidden by the terms of the form \(\neg \lca_{-1}(x_1,z)\).) By \Cref{lem:elimination_forest}, we obtain an \(\lca\)-reduced formula \(\widehat\varphi\)
and a labeling \(\widehat \ell\) such that
\(\true(\varphi_1,F,\ell) =
\true(\widehat\varphi, F, \widehat \ell)\).
Now, as
mentioned above, each \(x_i\) is in the same component as \(x_1\). Since each
tree has diameter at most \(d\), this mere fact implies that \(\ball(x_i,x_1,d)\) holds for every~$i$. As a consequence, 
\(\true(\widehat\varphi, F, \widehat \ell)=\true(\widehat\varphi \wedge \ball(\overline x,x_1, d), F, \widehat \ell)\). 
\end{proof}

\subsubsection{Quantifier Elimination of $\lca$-Reduced Formulas in CONGEST}

We now show that the quantifier elimination on local formulas described in the proof of
\Cref{lem:local_elimination_forest} can be computed in $\cO(1)$ rounds in the \CONGEST model. 

\begin{lemma}\label{lem:local_forest_computing}
Let \(\Lambda\) be a finite set of labels, \(d\in \N\) and \(r\ge d\). For every \(r\)-local \(\lca\)-expressed existential formula 
$\varphi(\overline{x})= \exists \overline{y}\; \big(\zeta(\overline{x},\overline{y} \wedge
\ball(\{\overline x, \overline y\},x_1,r))\big)$ in $\FO[\Lambda]$, there exists
a set of labels \(\widehat\Lambda\), a \(\lca\)-reduced formula \(\widehat \varphi\) and a
distributed algorithm performing in $\cO(1)$ round in the \CONGEST model satisfying the following.
For every $\Lambda$-labeled graph $(G,\ell)$, every subgraph $F\in \mathcal{F}_d$ of $G$:
\begin{itemize}
\item every node $v\in V(G)$ is given as input: a set of labels $\ell(v)\subseteq \Lambda$, whether it belongs to~$F$, and, if applicable, its depth and parent in~\(F\);
\item every node $v\in V(G)$ outputs a set of labels  $\widehat{\ell}(v)\subseteq \widehat\Lambda$ such that
\(\widehat \Lambda, \widehat \varphi, \widehat \ell\) form a \(d\)-local reduction of \(\Lambda, \varphi, \ell\).
\end{itemize}
\end{lemma}

\begin{proof}
Let $k = |\overline{x}|$. As for the proof of \Cref{lem:elimination_forest}, we assume w.l.o.g. that $k\geq 1$, and we treat only the case where $|\overline{y}|=1$ (as the general case $|\overline{y}|\geq 1$ can be obtained by iteratively eliminating the $|\overline{y}|$ quantifiers one by one). In the same way they are defined in \Cref{lem:elimination_forest}, let us define the set $I \subseteq~\Type(k+1, D, \Lambda)$, and, for each $\psi \in I$, the label sets $\widehat{\Lambda}^\psi$, the $\lca$-reduced formula $\widehat{\psi}$, and the $\hat{\Lambda}^\psi$-labeling $\widehat{\ell}^\psi$ of~$F$. Furthermore, we define 
$$
\widehat{\Lambda} = \bigcup_{\psi \in I} \widehat{\Lambda}^\psi, \;  \widehat{\varphi} = \bigvee_{\psi \in I} \widehat{\psi},  
\text{ and, for every node $u\in F$, } \widehat{\ell}(u) = \bigcup_{\psi \in I} \widehat{\ell}^\psi(u). 
$$
We then extend $\widehat{\ell}$  to the nodes of  $G\smallsetminus F$ by merely defining $\widehat{\ell}(v) = \ell(v)$ for all $v\notin V(F)$. 
Since $\Type(k+1, d, \Lambda)$ is a finite set, it is enough to show that, for every $\psi\in I$, the label set $\widehat{\Lambda}^\psi$, and the formulas $\widehat{\psi}$  and $\widehat{\ell}^\psi$ can be computed in $\cO(1)$ rounds. In fact, the definitions of the set $I$,  $\widehat{\Lambda}$, and $\widehat{\varphi}$ do not depend in the sub-forest~$F$ of~$G$, so $I$,  $\widehat{\Lambda}$, and $\widehat{\varphi}$
can be computed at every node without any communication. 
It remains to show how to compute, for any given $\psi\in I$, the value of $\widehat{\ell}^\psi(v)$ at all nodes~$v$ in a constant number of rounds. 

Initially the algorithm sets $\widehat{\ell}^\psi(v) = \ell(v)$ at every node.  Let $\gamma$ and $\delta$ be the functions such that $\psi = \type_{\gamma,\delta}$. Let $s$, $h_s$, $h_y$, and $h$ be as defined in the proof of \Cref{lem:elimination_forest}. Observe that every node knows $\varphi$, and thus every node can compute $\psi$, $h_s$, $h_y$, and $h$ without any communication. Furthermore, every node knows
whether its depth is equal to $h$, $h_s$, or $h_y$, or none of these three values.

We now revisit the three cases treated in the proof of \Cref{lem:elimination_forest}.

\paragraph{Case 1: $\delta(s,k+1)\geq0$ and $h_y = h$.} 

In this case, the algorithm just needs to identify the good nodes in $F$. To do so, for every $i\in [p]$,  every node $v$ at depth $h_y$ in $T_i$ checks whether it is a candidate node. Every such nodes does not need any communication for this, as it knows its depth in $T_i$, as well as $\ell(v)$ and $\gamma(v)$. Then, the algorithm broadcasts ping messages (on a single bit) downward every tree, from every candidate node. These ping messages are forwarded by their descendants downward every $T_i$, for $h_s-h$ rounds. More precisely, each candidate node sends a ping as a single bit~1 to all its children, and, for $h_s-h-1$ rounds, each node that receives that ping from its parent forwards it to all its children.  Finally, every node $u$ at depth $h_s$ that receives the ping know that their ancestor at level $h_y$ are candidates, and thus they mark themselves as good nodes. Moreover, every such node $u$ sets $\widehat{\ell}^\psi(u) = \ell(u) \cup \{\good\}$.  Every other node $u$ in~$F$ sets $\widehat{\ell}^\psi(u) = \ell(u)$. All these operations consume $h_s - h =\cO(d)$ rounds. 

\paragraph{Case 2: $\delta(1,k+1)\geq0$ and $h_y > h$.} 

In this case the algorithm is divided in two phases. Similar to previous case, the candidate nodes can identify themselves without any communication. In the first phase, ping messages are broadcast upward in each tree~$T_i$, starting from the candidate nodes, during $h-h_y$ rounds. In other words, the candidate nodes send a single bit to their parents, and, for $h-h_y-1$ more rounds, each node that receives a ping in the previous round forwards it to its parent. During this phase, every node $u$ at depth $>h$ that receives a ping message is marked as a good node, and it adds $\good$ to $\widehat{\ell}^\psi(u)$. This phase takes $h-h_y  = \cO(d)$ rounds. During the second phase, each pivot node $z$  (i.e., every node at depth $h$) computes $\kappa(z)$ as the number of children that sent a ping message during the last round of the first phase. Then, each pivot node $z$ broadcasts the value $\ell(\kappa(z))$ downward the trees, for $h-h_s$ rounds. After $h-h_s$ rounds, every node $u$ at depth $h_s$ has received the message $\ell(\kappa(z))$ from its pivot ancestor $z$, and it adds this label to $\widehat{\ell}^\psi(u)$. This phase takes $h-h_s  = \cO(d)$ rounds.

\paragraph{Case 3: $\delta(1,k+1)=-1$.} 

As mentioned in the proof of \Cref{lem:local_elimination_forest}, this case cannot happen for
local formulas.

\medskip

We conclude that for every $\psi \in I$, every node $u$ can compute $\ell^\psi(u)$, in $\cO(d)$ rounds. Therefore,  the total number of rounds sufficient for computing $\widehat{\ell}$ is $\cO(d \cdot |I|) = \cO(d \cdot 2^{k|\Lambda|}(d+1)^{k^2})$.
\end{proof}

\subsection{Quantifier Elimination on Graphs of Bounded Treedepth}
\label{ss:local_treedepth}

Given a graph $G =(V,E)$, a \emph{decomposition forest} $F$ of $G$ is a rooted spanning forest of $G$ satisfying that for every edge $e\in E$, one of its endpoints is an ancestor of the other in $F$. Observe that if $G$ is connected then any  decomposition forest  of $G$ is actually a tree, and every DFS tree of $G$ is a decomposition tree of $G$. For every $t>0$, any induced subgraph $H$ of $G$ with treedepth at most $t$ is called \emph{a treedepth-$t$ induced subgraph}.

We stress the fact that a decomposition forest $F$ of $G$ must be a subgraph of $G$, which is very convenient for distributed algorithms because this enables communication along the edges of~$F$.  Note that there exist graphs of treedepth~$t$ in which any decomposition forest is of diameter $\Omega(2^t)$. This is for example the case of paths on $2^t-1$ vertices. Nevertheless, graphs of treedepth at most $t$ have paths of at most $2^t-1$ vertices. In particular, for any class of graphs of bounded treedepth, all possible decomposition forests also have trees of constant (although exponentially larger) diameter.

\begin{proposition}[\cite{NesetrilOdM12}]\label{prop:depth}
    Any decomposition forest of a graph of treedepth at most $t$ belongs to~$\mathcal{F}_{2^t}$.
\end{proposition} 

In the following lemma, we show that any distributed algorithm aiming at checking a quantifier-free formula in a treedepth-$t$ induced subgraph $H$ of a labeled graph~$G$,  may rather check an equivalent formula on a decomposition forest of~$H$. 

\begin{lemma}\label{lem:treedepthtoforest}
Let $\Lambda$ be a finite set of labels, and let $t>0$. For every quantifier free formula $\varphi(\overline{x})\in \FO[\Lambda]$, there exists a finite set $\widehat{\Lambda}$ depending only on $t$ and $\Lambda$, and an $\lca$-reduced formula $\widehat{\varphi}(\overline{x}) \in \FO[\widehat{\Lambda}]$ depending on $\varphi$, $\Lambda$, and $t$, such that, for every $\Lambda$-labeled treedepth-$t$ induced subgraph $(H,\ell)$ of a $\Lambda$-labeled graph $(G,\ell)$, and for every decomposition forest $F$ of~$H$, there exists a  $\widehat{\Lambda}$-labeling $\widehat{\ell}$ of $G$ such that
$$
\true(\varphi,H,\ell) = \true(\widehat{\varphi},F,\widehat{\ell}).
$$ 
Moreover, there is a 1-round \CONGEST\ algorithm that computes $\widehat{\ell}$, assuming each
node \(v\) is given as input \(\ell(v)\), whether it belongs to \(F\) and, if applicable, its
depth and parent in \(F\).
\end{lemma}

\begin{proof}
Each node $u$ initially knows $\Lambda$, $\varphi$, $\ell(u)$, $t$, whether it belongs to $H$, and, if applicable, its depth and parent in $F$. Node $u$ returns $\widehat{\Lambda}$, $\widehat{\varphi}$, and $\widehat{\ell}(u)$.
Let $d=2^t$, and let us consider a $\Lambda$-labeled graph $(G,\ell)$, a treedepth-$t$ induced subgraph $H$ of $G$, and a decomposition forest $F$ of $H$. The idea of the construction is the following. Each node $u$ of $F$ can be labeled by a subset of $[0,d-1]$ indicating the set of levels of $F$ containing a neighbor of $u$ in~$H$. Then, the  $\adj(x,y)$ predicate can be transformed into label and $\lca$ predicates on $x$ and~$y$. More precisely, for every node $v\in V(H)$, let $\levels(v) \subseteq [d]$ be the set of levels of $F$ containing ancestors $u$ of $v$ such that $\{u,v\}\in E(H)$. Then, let $\widehat{\Lambda} = \Lambda \cup [0,d-1]$ (we assume w.l.o.g. that $\Lambda \cap [0,d-1] = \emptyset$), and let us consider the $\widehat{\Lambda}$-labeling $\widehat{\ell}$ of $G$ defined by 
$$
\widehat{\ell}(v) = \begin{cases}
    \ell(v) \cup \levels(v) & \text{if } v \in V(F),\\
    \ell(v)  & \text{otherwise.}
\end{cases}
$$ 
For each $i\in [0,d-1]$, we denote by $\level_i(x)$ the predicate $\lab_i(x)$, which is true when $x$ has a neighbor at depth $i$. Finally, we define $\widehat{\varphi}$ by replacing each $\adj(x,y)$ predicates by

\begin{align*} 
\bigvee_{1\leq i<j\leq d-1}\Bigg[\bigg(\lca_j(x,x)  \wedge \lca_i(y,y) \wedge \level_i(x) \bigg)  \vee \bigg(\lca_i(x,x)  \wedge \lca_j(y,y) \wedge \level_i(y) \bigg)\Bigg],
\end{align*}
and every equality predicate by $\lca$ predicates as we explained in \Cref{rem:fromFOtolca}. In this way, we obtain a $\lca$-reduced formula $\widehat{\varphi}$ such that $\true(\varphi,H,\ell) = \true(\widehat{\varphi},F,\ell)$. 

Finally, observe that every node $u$ can compute $\widehat{\Lambda}, \widehat{\varphi}$, and $\widehat{\ell}(u)$ in a single communication round. In fact $\widehat{\Lambda}, \widehat{\varphi}$, and $\widehat{\ell}(u)$ for every $u\in V(G)\smallsetminus V(H)$ are computed without any communication, as these values do not depend on the instance. To compute $\widehat{\ell}(u)$ for $u\in V(F)$, each node broadcasts its depth in $F$ to all its neighbors. After this communication round, $\levels(u)$ can be set as the set of neighbors with a smaller depth in $F$, and the computation of $\widehat{\ell}(u) = \ell(u) \cup \levels(u)$ follows.
\end{proof}

We now prove that existential quantifier in formulas on graphs of bounded
treedepth can be eliminated efficiently. The reduction results in a \(\lca\)-reduced formula to be checked on a decomposition
forest of the considered graph.

\begin{lemma}
\label{lem:qelim:td}
Let \(\Lambda\) be a set of labels,
\(\varphi(\overline x) = \exists \overline y, \zeta(\overline x, \overline y)
\wedge \ball((\overline x, \overline y), x_1, r)\) be a \(r\)-local existential
formula in \(\FO[\Lambda]\), where \(r\in \N\) and \(\zeta\) is a quantifier free formula.
There exists a set of labels~\(\widehat \Lambda\), and an \(\lca\)-reduced formula
\(\widehat \varphi(\overline x)\) such that, for every treedepth-\(t\) induced subgraph \((H,\ell)\)
of a \(\Lambda\)-labeled graph \((G,\ell)\), and for every decomposition forest \(F\) of \(H\),
there exists a \(\widehat \Lambda\)-labeling \(\widehat \ell\) such that:
\[
\true(\varphi, H, \ell)
= \true(\widehat\varphi, F, \widehat \ell)
= \true(\widehat\varphi \wedge\ball(\overline x, x_1,2^t), F, \widehat \ell).
\]
Moreover, there exists an algorithm computing \(\widehat \ell\) in \(\mathcal O(1)\) \CONGEST
rounds, assuming each node \(v\) is given as input \(\ell(v)\), whether it belongs to \(F\) and,
if applicable, its depth and parent in \(F\).
\end{lemma}

The lemma above considers both a graph \(H\) of bounded treedepth, and a
decomposition forest~\(F\).  The distances in both structures are not the same, and 
two vertices may be further apart in \(F\) than they are in~\(H\). However, the
connected components remain the same in both graphs.

\begin{proof}
We are assuming
that each node \(u\) initially knows \(\Lambda, \varphi, \ell(u)\), whether it belongs to~\(H\),
and, if applicable, its depth and parent in \(F\).
We consider graphs of treedepth at most \(t\), thus their connected components have
diameter at most~\(2^t\).
If \(r < 2^t\), then, by \Cref{prop:beta} we can rewrite
\(\ball(z, x_1, r)\) as an existential formula of
the form \(\exists \overline w, \psi(x_1, z, \overline w) \wedge \ball((z,\overline w),
x_1, 2^t)\).
Thus, without loss of generality, we consider the case \(r \geq 2^t\). Therefore, the balls of
radius \(r\) are then exactly the connected components of \(H\).

By \Cref{lem:treedepthtoforest}, there exists an \(\lca\)-reduced formula
\(\zeta_1(\overline x, \overline y)\) such
that, for every induced labeled subgraph \( (H,\ell)\) of treedepth at most~\(t\), for every
decomposition forest \(F\) of~\(H\), there exists a labeling \(\ell_1\) such that
\(\true(\zeta, H, \ell) = \true(\zeta_1, F, \ell_1)\), and \(\ell_1\) can be computed in
\(\mathcal O(1)\) rounds. Moreover, since connected
components are the same in \(H\) and \(F\), we also have 
\[
\true(\zeta \wedge
\ball((\overline x, \overline y), x_1, r), H, \ell) = \true(\zeta_1 \wedge
\ball((\overline x, \overline y), x_1, r), F, \ell_1).
\]
The proof is completed by applying Lemma \ref{lem:local_forest_computing} to the formula \(\exists
\overline y, \zeta_1(\overline x, \overline y) \wedge \ball((\overline x,\overline y),
x_1,r)\).
\end{proof}

\subsection{Graphs of Bounded Expansion}
\label{ss:local_bexp}

Let us fix a graph class $\mathcal{G}$  of expansion~$f$. Given a set of labels~$\Lambda$,
let $\mathcal{G}[\Lambda]$ be the set of $\Lambda$-labeled graphs $(G,\ell)$ such that
$G\in\mathcal{G}$. Let  $(G,\ell) \in \mathcal{G}[\Lambda]$ be the input instance. For every
integer~$p$, let $\Col(p)$ be the collection of sets $U\subseteq[f(p)]$ with cardinality
at most $p$.

\subsubsection{Skeletons}

We introduce a central notion for our proof of \Cref{thm:thmFOlocal}. 

\begin{definition}
Let \(p \in \N\), and let $G$ be a graph of expansion~$f$.  A \emph{\(p\)-skeleton} of \(G\) is a triple $S = (G, c, \delta)$  where \(c\) is \((p,f(p))\)-treedepth coloring of~\(G\), and $\delta$ is a function mapping every $U\in \Col(p)$ to a decomposition forest $F^U$ of the subgraph \(G[U]\) induced by nodes with colors in~\(U\). Labeled \(p\)-skeletons are pairs \((S,\ell)\), where $S= (G, c, \delta)$ is a \(p\)-skeleton, and \(\ell\) is a labeling of \(G\).
\end{definition}

Note that, by definition of a \((p,f(p))\)-treedepth coloring~$c$, the graph $G[U]$ is of treedepth at most~$p$ for every $U\in \Col(p)$.
Hence, by \Cref{prop:depth}, we have $F^U\in \mathcal{F}_{2^p}$. The following lemma states that a $p$-skeleton can be computed in $\cO(\log n)$ rounds in the \CONGEST\ model, where the big-$\cO$ notation hides constants depending only of~$p$,  and of the expansion of the class.

\begin{lemma}
    \label{prop:skeleton-construction}
    For every \(p \in \N\), there exists an algorithm that constructs a \(p\)-skeleton
    of \(G\) in \(\mathcal O(\log n)\) rounds in  \CONGEST. Specifically, at the end of the algorithm, each node
    knows its own color, and, for every set \(U\in\Col(p)\) containing its color,
    its parent and depth in the forest~\(F^U\) corresponding to the skeleton.
\end{lemma}

\paragraph{Remark.} In the remaining of the paper, when an algorithm takes a skeleton as input,
or outputs a skeleton, it is assumed
that the skeleton is represented in a distributed manner, in the way the algorithm in \Cref{prop:skeleton-construction} computes it.

\begin{proof}
The algorithm of \Cref{th:congestltd} enables to compute a $(p,f(p))$-treedepth coloring $c$ of $G$, in $\cO(\log n)$ rounds. After running this algorithm, each node $u$ knows $c(u)$. Let us order the sets in $\Col(p)$ lexicographically. The rest of the algorithm runs in $N = |\Col(p)| \leq 2^{f(p)}$ phases, where each phase takes at most $2^{2p+1}$ rounds. The $i$-th phase  consists of dealing with the $i$th set $U\in \Col(p)$. Let us fix one such set~$U$. For every $u\in V(G)$, if $c(u)\notin U$, then $u$ does not participate in that phase (it remains idle, and, in particular, it does not forward any messages). Only the nodes $u$ such that $c(u)\in U$ exchange messages during phase~$U$. The phase is divided in two steps. 

\begin{itemize}
    \item In the first step, each node $u\in G[U]$ performs  $2^p$ rounds during which it broadcasts the value of a variable $\textsf{min-id}(u)$.  Initially, $\textsf{min-id}(u)=\id(u)$. In the subsequent rounds, $\textsf{min-id}(u)$ is the minimum identifier of a node in the component of $G[U]$ containing~$u$ (i.e., the minimum identifier received by $u$ from a neighbor in~$U$). Observe that $G[U]$ is a graph of diameter at most~$2^p$, and thus $\textsf{min-id}(u)$ is eventually the smallest identifier in the whole component of $G[U]$ containing~$u$. After the first $2^p$ rounds of the phase, the node with identifier $\textsf{min-id}(u)$ is the root of the tree containing~$u$. Let $V_1, \dots, V_h$ be the components of $G[U]$, and, for every $i\in [h]$, let $r_i$ be the node of minimum identifier $\id(r_i)$ among all nodes in $V_i$. Note that, after step~1 has completed, every node $u\in V_i$ has $\textsf{min-id}(u) = \id(r_i)$. 

    \item In the second step, the algorithm runs in parallel in all components~$V_i$, $i\in\{1,\dots,h\}$ for computing a decomposition tree $T_i$ of $G[U]$ rooted at~$r_i$. For this purpose, we use the algorithm from \cite{FominFMRT24} (cf. Lemma~11 therein). This algorithm performs in $2^{2p}$ rounds. 
\end{itemize}

At the end of these two steps, every node $u$ knows its depth in its tree~$T_i$, as well as the identifier of its parent in~$T_i$. Observe that the collection of trees  $\{T_i\mid i\in [h]\}$ is a decomposition forest $F^U$ of $G[U]$. This completes the description of the phase for set~$U$. The whole algorithm runs in $\cO(\log n) + |\Col(p)|\cdot 2^{2p+1} = \cO(\log n)$ rounds.
\end{proof}

\subsubsection{Reducibility on Graphs of Bounded Expansion}
\label{sss:reducibility-bexp}

We now define a specific notion of quantifier elimination, denoted \emph{reducibility},
for graphs of bounded expansion. For this purpose, similarly to how   FO formulas on
treedepth-\(t\) graphs are transformed into FO formulas on their decomposition forests, we will transform
FO formulas on graphs of \(\mathcal G\) into FO formulas on their \(p\)-skeletons. For this purpose, we define  additional predicates for describing the structure of \(p\)-skeletons.
\begin{itemize}
    \item For each \(k\in[f(p)]\), the predicate \(\col_k(x)\) checks whether the considered vertex~$x$ has color \(k\) in the skeleton.

    \item For every $U\in \Col(p)$, and every $j\in [-1,d-1]$, the predicate $\lca^U_j(x,y)$ is equivalent to the $\lca_j$ predicates, but on $F^U$. That is,  \(\lca^U_j(x,y)\) is true whenever $x$ and $y$ are colored by colors in \(U\), and \(\lca_j(x,y)\) is true on~\(F^U\). Such predicates are called $\lca^U$ predicates.
\end{itemize}
A formula $\varphi(\overline x) \in \FO[\Lambda]$ is called \emph{$p$-$\lca$-reduced} if it is a quantifier-free, and it contains only label, \(\col\) and $\lca^U$ predicates, for some collection of sets $U\in \Col(p)$. For such a formula, one can ask whether a labeled $p$-skeleton $(S,\ell)$ satisfies~$\varphi$ on some vertices \(\overline v \in S\), i.e., one can ask whether 
\[
(S,\ell)\models \varphi(\overline v), 
\]
where $S=(G,s,\delta)$  is a labeled $p$-skeleton. The set $\true(\varphi,S,\ell)$ is then defined accordingly.
Note that we abusedly write \(\overline v\in S\) instead of \(\overline v \in G\),
which itself is an abuse of notation for \(\overline v \in V\), where \(V\) is the set of vertices
of \(G\).

\begin{definition}\label{def:reducibleonBE}
 A formula $\varphi \in \FO[\Lambda]$ is called \emph{reducible on $\mathcal{G}$} if there exists a positive integer~$p$, a set of labels~$\widehat{\Lambda}$, and a $p$-$\lca$-reduced formula $\widehat{\varphi}$ such that, for every labeled graph $(G,\ell) \in \mathcal{G}[\Lambda]$, and every $p$-skeleton $S$ of $G$, there exists a $\widehat{\Lambda}$-labeling  $\widehat{\ell}$ of \(S\) such that $$\true(\varphi,G,\ell) = \true(\widehat{\varphi},S,\widehat{\ell}).$$
Given a graph $G$ and a \(p\)-skeleton \(S\), the tuple $(p, \widehat{\Lambda}, \widehat{\varphi}, \widehat{\ell})$ is called the \emph{reduction} of $(\Lambda, \varphi, \ell)$ on~$S$.
\end{definition}

Importantly,  when considering local formulas, it is also required that the reduction \(\widehat\varphi\) be 
a local formula too.

\begin{definition}\label{def:local_reducible_BE}
An \(r\)-local formula \(\varphi(\bar x) \in \FO[\Lambda]\) is said \emph{locally reducible on
\(\mathcal G\)} if there exists \(r' \in \N\), such that, for every \((G,\ell) \in \mathcal
G[\Lambda]\), and every \(p\)-skeleton \(S\) of \(G\), there exists a reduction
\((p, \widehat\Lambda,\widehat\varphi,\widehat \ell)\) of \((\Lambda, \varphi, \ell)\) on \(S\)
such that:
\[
\true(\varphi, G, \ell) = \true(\widehat\varphi, S, \widehat \ell) =
\true(\widehat\varphi \wedge \ball(\overline x, x_1, r'), S, \widehat \ell).
\]
\end{definition}

\noindent Note that, when considering distance predicates on a skeleton $S=(G,c,\delta)$, we use the distances
on the underlying graph~$G$.
Our goal is now to show that every (local) formula can be efficiently reduced  in $\mathcal{G}$
by a distributed algorithm. We first establish this result for existential formulas. 

\begin{lemma}
\label{lem:qelim:bexp}
Let \(\Lambda\) be a set of labels, and let \(\varphi(\overline x) = \exists \overline y,
\psi(\overline x, \overline y) \wedge \ball((\overline x, \overline y),
x_1, r)\) be an \(r\)-local existential formula in \(\FO[\Lambda]\), where
\(\psi\) is a quantifier free formula.
For every integer \(p\geq r\cdot (|\overline x|+|\overline y|)\), there exists a set of labels
\(\widehat \Lambda\), and a
\(p\)-\(\lca\)-reduced formula
\(\widehat \varphi(\overline x)\) satisfying that, for every \((G,\ell) \in \mathcal
G[\Lambda]\), and for every \(p\)-skeleton \(S\) of~\(G\),
there exists a \(\widehat \Lambda\)-labeling of \(\widehat \ell\) such that
\[
\true(\varphi, G, \ell)
= \true(\widehat\varphi, S, \widehat \ell)
= \true(\widehat\varphi \wedge\ball(\overline x, x_1,2^p), S, \widehat \ell).
\]
Moreover, there exists an algorithm that computes \(\widehat \ell\) in \(\mathcal O(1)\) rounds
in \CONGEST, assuming \(S\) and \(\ell\) are given as input.
\end{lemma}

\begin{proof}
Let \(p \geq r(|\overline x|+|\overline y|)\). We apply Lemma \ref{lem:qelim:td} to
\(\varphi\), which says that 
there exists a set of labels~\(\Lambda_1\), and a \(lca\)-reduced formula
\(\varphi_1(\overline x)\) such that, for every labeled
graph \( (H, \ell)\) of treedepth at most \(p\), and for every decomposition forest \(F\)
of \(H\), there is a \(\Lambda_1\) labeling \(\ell_1\) of \(F\) such that
\[
\true(\varphi, H, \ell) = \true(\varphi_1, F, \ell_1) =
\true(\varphi_1\wedge\ball(\overline x, x_1, 2^p), F, \ell_1).
\]
Let us now consider now a \(p\)-skeleton~\(S=(G,s,\delta)\).
For each \(U\in \Col(p)\), the induced subgraph \(G[U]\) has treedepth at most
\(p\), and \(F^U=\delta(U)\) is a decomposition forest of \(G[U]\). Thus, we can
compute the corresponding labeling
\(\ell_1^U\) in \(\mathcal O(1)\) rounds. Let us denote by \(\widehat\varphi^U(\overline x)\) the \(p\)-lca-reduced
formula \(\varphi_1^U(\overline x)\wedge \col_U(\overline x)\). That is, \(\varphi_1^U\) is the formula \(\varphi_1\) in which all
\(\lca\) predicates of were replaced by \(\lca^U\) predicates.
Finally, let us define
\[
\widehat\varphi(\overline x) = \bigvee_{U \in \Col(p)}
\widehat\varphi^U(\overline x),
\]
and, for every node~\(u\), 
\[
\widehat \ell(u) = \bigcup_{U \in \Col(p)} \ell_1^U(u),
\]
where, w.l.o.g., we assume that the formulas 
\(\varphi^U\) use disjoint  sets of labels.
Let us prove that this formula has the desired property. Let \( (G,\ell)\in
\mathcal G[\Lambda]\),  and let \(S\) be a \(p\)-skeleton of \(G\). We prove the
equivalence of the following three statements:
\begin{enumerate}[label=(\roman*)]
\item \( (G,\ell)\models \varphi(\overline{v})\),
\item \( (S,\widehat \ell) \models \widehat\varphi(\overline v)\), and 
\item \( (S, \widehat \ell) \models \widehat\varphi(\overline v) \wedge \ball(\overline v, 
v_1, 2^p)\).
\end{enumerate}
The implication \((iii)\Rightarrow (ii)\) is obvious. Let us prove \( (i) \Rightarrow  (iii)\).
For this purpose, let us assume that 
\[
(G,\ell)\models \psi(\overline v, \overline w) \wedge \ball((\overline v, \overline w),
v_1,r)
\]
for some \(\overline v,
\overline w \subseteq V(G)\). Since \(p \geq r\cdot (|v|+|w|)\), there exists  \(U
\in \Col(p)\) such that \(G[U]\) contains both \(\overline v\) and \(\overline w\), and
also, for every \(u \in
\overline v, \overline w\), a path of length at most \(r\) from \(v_1\) to \(u\). We get that
\[
(G[U],\ell) \models \psi(\overline v, \overline
w) \wedge \ball((\overline v, \overline w), v_1, r).
\]
The first of the following series of implications is a direct consequence of the latter, and the second implication follows from Lemma~\ref{lem:qelim:td}:
\begin{align*}
 (G,\ell)\models \varphi(\overline v) & 
\implies \exists U\in\Col(p), \overline v \subseteq G[U] \wedge
	(G[U], \ell) \models \varphi(\overline v) \\
& \implies \exists U\in\Col(p), \overline v \subseteq G[U] \wedge
	(F^U, \ell_1^U) \models \varphi_1(\overline v) \wedge
	\ball(\overline v, v_1, 2^p) \\
& \implies \left[\exists U\in \Col(p),
	(S, \ell_1^U) \models \widehat\varphi^U(\overline v)\right] \wedge \left[
	\exists U \in \Col(p),
	F^U \models \ball(\overline v, v_1, 2^p)\right] \\
& \implies \left[(S, \widehat \ell) \models \widehat\varphi(\overline v) \right]
	\wedge \left[S \models \ball(\overline v, v_1, 2^p)\right] \\
& \implies (S, \widehat \ell) \models \widehat\varphi(\overline v) \wedge
	\ball(\overline v, v_1, 2^p)
\end{align*}
Finally we prove \( (ii) \Rightarrow  (i)\). We have
\begin{align*}
(S,\widehat \ell) \models \widehat\varphi(\overline v) 
& \implies \exists U \in \Col(p),
	(S, \ell_1^U) \models \widehat\varphi^U(\overline v) \\
&\implies \exists U \in \Col(p), \overline v \subseteq G[U] \wedge
	(F^U, \ell_1^U) \models \varphi_1(\overline v) \\
& \implies \exists U \in \Col(p), \overline v \subseteq G[U] \wedge
	(G[U], \ell) \models \varphi(\overline v) \\
& \implies (G,\ell) \models \varphi(\overline v)
\end{align*}
where the penultimate implication follows from Lemma~\ref{lem:qelim:td}, and the last from the fact that $\varphi$ is existential (i.e., if the quantified variables exist in a subgraph, then they exist in~$G$).  
This completes the proof. 
\end{proof}

\begin{lemma}
\label{lem:qelim:remove-lca}
Let \(p\in \N\), let \(\Lambda\) be a set of labels, and let \(\varphi(\overline x)\) be a
\(p\)-\(lca\)-reduced formula. There exists a set of labels \(\widehat
\Lambda\), and there exists  a quantifier free formula \(\psi(\overline x, \overline y)\) such
that, for every \(G\in \mathcal G\), for every 
\(\Lambda\)-labeling \(\ell\) of \(G\), and for every  \(p\)-skeleton \(S\) of \(G\), there exists a labeling \(\widehat \ell\) of \(G\) for which 
\[
\true\big((\varphi \wedge \ball(\overline x, x_1, r)), S, \ell\big)
= \true(\widehat\varphi, G, \widehat \ell)
\]
where
\[
\widehat\varphi(\overline x) = \exists \overline y \, \big(\psi(\overline x,
\overline y) \wedge \ball((\overline x, \overline y), x_1, r+2^p)\big).
\]
Moreover, \(\widehat \ell\) can be computed without any communication,
assuming \(S\) and \(\ell\) are given as inputs.
\end{lemma}

The main idea of the proof is to replace each \(\lca^U\)-predicate of \(\varphi\) by an existential
FO formula on the graph \(G\) underlying the (labeled) \(p\)-skeleton \((S,\ell)\). To that end, we
augment the labeling \(\ell\) with additional labels, which encode the structure of \(S\).
Finally, we add distance predicates to enforce the locality of the resulting formula.

\begin{proof}
Let us assume that  \(\varphi\) is written in conjunctive normal form (CNF).
Let \(d=2^p\). We have that 
\[
\neg\lca_i(x,y) \equiv \neg\col_U(x) \vee \neg\col_U(y) \vee \bigvee_{j\in [-1,d] \setminus \{i\}}
\lca_j(x,y).
\]
Therefore, we can assume that \(\varphi\) does not have any 
predicate of the form \(\neg\lca_i(x,y)\), as it suffices to replace each occurrence
of such predicates by the corresponding disjunction of predicates. Note that the formula
we obtain is still in conjunctive normal form.
It only remains to replace each predicate 
\(\lca_i^U(x,y)\) by an equivalent existential formula: since \(\varphi\) is in conjunctive
normal form and no \(\lca^U\) predicate is negated, the formula we obtain is existential.
To that end, we add to \(\Lambda\) new labels representing the structure of the \(p\)-skeleton:
let
\[
\widehat \Lambda = \Lambda
\cup \{\col_k, k \in [f(p)]\}
\cup \{\depth_k^U, U\in \Col(p), k \in [2^p]\}.
\]
We assume w.l.o.g. that \(\Lambda\) does not already contain these labels.
We construct the labeling \(\widehat \ell\) as follows, without any communication. Each node
\(v\) adds its own color to its labeling, and, for each
\(U \in \Col(p)\) that contains \(v\)'s color, node $v$ adds its depth in \(F^U\) to its label.

Consider now a predicate \(\lca_i^U(y,z)\). We define a formula
\(\zeta_i^U \in \FO[\widehat \Lambda]\) such that
\((S,\ell)\models \lca_i^U(u,v) \iff (G,\widehat\ell) \models \zeta_i^U(u,v)\). 
\begin{align*}
\zeta_i^U(y,z) =\ &  \exists (y_0, \dots, y_d)\;  \exists (z_0, \dots, z_d)\;
\bigvee_{(a,b) \in [\max\{0,i\}, d-1]^2} \Biggl( (y_a = y) \wedge (z_b = z)
\\
 & \wedge \Big(\bigwedge_{s\in [0,i]} y_s = z_s\Big)
\wedge \Big(\bigwedge_{s\in [i+1, \min(a,b)]} y_s \neq z_s\Big)
\wedge \depth_0^U(y_0) \wedge \depth_0^U(z_0)
\\
& \wedge \Big(\bigwedge_{s\in [1,a]} \depth_s^U(y_s) \wedge \adj(y_{s-1}, y_s)\Big)
\wedge \Big(\bigwedge_{s\in [1,b]} \depth_s^U(z_s) \wedge \adj(z_{s-1}, z_s)\Big)\Biggr).
\end{align*}
Indeed, \(a\) (resp. \(b\)) is the depth of \(y\) (resp. \(z\)) in \(F^U\),
and \(y_0, \dots, y_a\) (resp. \(z_0, \dots, z_b\)) correspond to the path from
the root of \(y\)'s tree to \(y\) (resp. from the root of \(z\)'s tree to \(z\)).
The equality predicates on the first line, and the depth predicates on the second line
guarantee that the endpoints of the path are as expected.
The equality predicates on the second line guarantee that both paths share exactly
\(i+1\) vertices. The depth and adjacency predicates on the third line
guarantee that, for every \(s\in[1,a]\) (resp., every \(s\in[1,b]\)), \(y_{s-1}\) (resp., \(z_{s-1}\))
is the parent of \(y_{s}\) (resp., \(z_{s}\)). By construction, we can assume that each of the \(y_i\) (resp., \(z_i\)) we
introduced are at distance at most \(d\) from \(y\) (resp., \(z\)).

We would like to replace each occurrence of \(\lca_i^U\) in \(\varphi\) by \(\zeta_i^U\), while
preserving locality. To that end, we add a distance predicate to the formula \(\zeta_i^U\),
which enforces locality.
Recall that we are interested by the formula \(\varphi \wedge \ball(\overline x, x_1, r)\).
Whenever the formula is satisfied, every \(x_i\) is at distance
at most \(r\) from \(x_1\). This implies that every  $y_s$ and $z_s$ introduced in
\(\zeta_i^U(y,z)\) is at distance at most \(r+d\) from \(x_1\). We thus replace
each \(\lca_i^U(y,z)\) predicate in \(\varphi\) by the following formula:
\begin{align*}
\exists (y_0, \dots, y_d)\; &  \exists (z_0, \dots, z_d) \; \Bigg(
\ball((\overline y, \overline z), x_1, r+d) \wedge \bigvee_{(a,b) \in [\max(0,i), 2^p]^2} \Biggl[ (y_a = y) \wedge (z_b = z)
\\
\wedge& \Big(\bigwedge_{s\in [0,i]} y_s = z_s\Big)
\wedge \Big(\bigwedge_{s\in [i+1, \min\{a,b\}]} y_s \neq z_s\Big)
\wedge \depth_0^U(y_0) \wedge \depth_0^U(z_0)
\\
\wedge& \Big(\bigwedge_{s\in [1,a]} \big(\depth_s^U(y_s) \wedge \adj(y_{s-1}, y_s)\big)\Big)
\wedge \Big(\bigwedge_{s\in [1,b]} \big(\depth_s^U(z_s) \wedge \adj(z_{s-1}, z_s)\big)\Big)\Biggr]\Bigg)
\end{align*}
Let us denote by \(\tilde\varphi(\overline x)\) this formula.
By the definition of \(\beta_{r,r'}\) in Proposition \ref{prop:beta}, we can then define
\[
\widehat\varphi(\overline x) = \tilde\varphi(\overline x) \wedge
\beta_{r,r+d}(\overline x) \wedge
\ball(\overline x, x_1, r+d).
\]
By construction, \(\widehat\varphi\) is existential, and thus
\[
\true(\widehat\varphi, G, \widehat \ell) =
\true(\widehat\varphi, G, \widehat \ell) \cap \true(\ball(\overline x, x_1,r), G, \widehat \ell)
= \true(\varphi \wedge \ball(\overline x, x_1, r), S, \ell),
\]
which completes the proof. 
\end{proof}

We now tackle the general case.

\begin{lemma}
\label{lem:qelim:induction}
Let \(\Lambda\) be a set of labels, and let \(\varphi(\overline x)\) be an FO formula in
\(r\)-local form. There exists \(p\in \N\), a set of labels \(\widehat\Lambda\),
and a \(p\)-\(lca\)-reduced formula \(\widehat\varphi(\overline x)\) such that, for
every labeled graph \( (G,\ell)\in \mathcal G[\Lambda]\), and for every  \(p\)-skeleton \(S\) of \(G\), there
exists a \(\widehat \Lambda\)-labeling \(\widehat \ell\) for which
\[
\true(\varphi, G, \ell) = \true(\widehat\varphi, S, \widehat \ell) =
\true(\widehat\varphi \wedge \ball(\overline x, x_1, 2^p),
S, \widehat \ell)
\]
Moreover \(\widehat \ell\) can be computed in \(\mathcal O(\log n)\)  rounds in the \CONGEST model, assuming \(S\) and \(\ell\) are given as inputs.
\end{lemma}

\begin{proof}
The proof proceeds by induction on the structure of the formulas written in local form as defined (in an inductive way) in \Cref{def:r-local-form}. We consider the three cases occurring in \Cref{def:r-local-form}.

\paragraph{Case 1 (base case).} \(\varphi(\overline x) = \psi(\overline x) \wedge \ball(\overline
x, x_1,r)\), where \(\psi\) is quantifier free. We can
directly apply \Cref{lem:qelim:bexp} with \(|\overline y|=0\), to compute \(\widehat\ell\)
from the skeleton \(S\) and labeling \(\ell\) given as inputs.

\paragraph{Case 2.} \(\varphi(\overline x) = \exists y\; \psi(\overline x, y)\),
where \(\psi\) is in local form. By induction, there exists $p'$, $\Lambda_1$, and  $\psi_1(\overline x, y)$ such that, for every \( (G,\ell) \in \mathcal G[\Lambda]\), 
and every  \(p'\)-skeleton \(S'\) of \(G\), there exists a \(\Lambda_1\)-labeling \(\ell_1\) for which 
\[
\true(\psi, G, \ell) = \true(\psi_1, S', \ell_1) = \true(\psi_1 \wedge \ball((\overline x, y),
x_1, 2^{p'}), S', \ell_1)
\]
By  \Cref{prop:skeleton-construction}, a
\(p'\)-skeleton \(S'\) of $G$ can be constructed in \(\mathcal O(\log n)\) rounds in \CONGEST.
By induction, the corresponding labeling \(\ell_1\) can be computed in \(\mathcal O(\log n)\)
rounds.
For eventually applying \Cref{lem:qelim:bexp} on \(\exists y \, \psi_1(\overline x, y)\),  we first apply Lemma~\ref{lem:qelim:remove-lca} on \(\psi_1\), with \(r=2^{p'}\). Thanks to this lemma, 
there exists a set of labels $\Lambda_2$, a quantifier-free formula $\rho_2(\overline x, y, \overline z)$, and a labeling \(\ell_2\) that can be constructed from \(\ell_1\) such that 
\[
\true(\psi_1 \wedge \ball((\overline x,y), x_1, 2^{p'}), S', \ell_1) =
\true(\psi_2(\overline x, y), G, \ell_2),
\]
where \(\psi_2(\overline x, y) = \exists
\overline z \, \rho_2(\overline x, y, \overline z) \wedge \ball((\overline x, y, \overline
z), x_1, 2^{p'+1})\).
Let us then define \(\varphi_2(\overline x) = \exists y \, \psi_2(\overline x, y)\). This is formula is
existential since \(\psi_2\) is existential. We can thus apply Lemma~\ref{lem:qelim:bexp} to \(\varphi_2\), from which the induction step follows. In particular, we can compute \(\widehat\ell\)
from the skeleton \(S\) given as input.

\paragraph{Case 3.} \(\varphi(\overline x) = \ball(\overline x, x_1, r)
\wedge \forall y \, \Big( \neg \ball(\overline x, y), x_1,r)
\vee \psi(\overline x, y)\Big)\) where \(\psi\) is in \(r\)-local form. We can
rewrite \(\varphi\) as
\[
\varphi(\overline x) = \ball(\overline x, x_1,r) \wedge
\neg\exists y, \ball((\overline x,y), x_1, r) \wedge
\neg \psi(\overline x, y).
\]
By induction on \(\psi\), there exist \(p'\in \mathbb{N}\), a set of labels \(\Lambda_1\), and a \(p'\)-\(\lca\)-reduced
formula \(\psi_1(\overline x, y)\) such that, for every labeled graph  \( (G,\ell)\in \mathcal
G[\Lambda]\), and for every \(p'\)-skeleton $S'$ of $G$, there exists a \(\Lambda_1\)-labeling \(\ell_1\) satisfying  
\[
\true(\psi, G, \ell) = \true(\psi_1, S', \ell_1).
\]
By \Cref{prop:skeleton-construction},
a \(p'\)-skeleton \(S'\) of $G$ can be constructed in \(\mathcal O(\log n)\)
rounds in the \CONGEST model. By induction, the labeling \(\ell_1\) can thus be computed in \(\mathcal O(\log n)\) rounds. Let us then take the negation of \(\psi_1\). Note that the resulting formula $\neg\psi_1$ is still a \(p'\)-\(\lca\)-reduced
formula. We then apply Lemma \ref{lem:qelim:remove-lca} to \(\neg\psi_1\), which results in a set of labels \(\Lambda_2\),  an existential formula \(\rho_2\) in \((r+2^p)\)-local form, and a
\(\Lambda_2\)-labeling \(\ell_2\) of \(G\) such that 
\[
\true(\neg\psi_1 \wedge
\ball(\overline x, x_1, r), S', \ell_1) = \true(\rho_2, G,
\ell_2).
\]
As a consequence, we get 
\begin{align*}
(G,\ell) \models \ball((\overline v,u), v_1,r) \wedge \neg \psi(\overline
v, u) 
& \iff \Bigl(G \models \ball((\overline v,u), v_1,r)\Bigr)
\wedge \Bigl((S',\ell_1) \models \neg\psi_1(\overline v, u)\Bigr)\\
& \iff (S',\ell_1) \models \ball((\overline v,u), v_1,r) \wedge
\neg\psi_1(\overline v, u) \\
& \iff (G,\ell_2) \models \rho_2(\overline v, u)
\end{align*}
where the first equivalence is by induction, and the last by application of Lemma \ref{lem:qelim:remove-lca}. 

Let us now focus on the formula
\[
\eta_2(\overline x) = \exists y \; \rho_2(\overline x,y).
\]
Since \(\rho_2\) is
existential,  and since it is in \((r+2^{p'})\)-local form, the formula  \(\eta_2\) is also existential and in \((r+2^{p'})\)-local form.
By Lemma~\ref{lem:qelim:bexp}, there exists an integer \(p \geq r\cdot |\overline x|\), a set of labels~\(\Lambda_3\), and
a \(p\)-\(lca\)-reduced-formula \(\eta_3(\overline x)\) such that, for every 
\(p\)-skeleton \(S\) of \(G\) (and in particular the skeleton \(S\) given as input),
there exists a labeling \(\ell_3\) for which
\[
\true(\eta_2, G, \ell_2) = \true(\eta_3, S, \ell_3).
\]
Once again, the negation of \(\eta_3\) is still \(p\)-\(lca\)-reduced.

We also need to transform the first half of \(\varphi\), namely \(\ball(\overline x, x_1, r)\),  into a \(p\)-\(lca\)-reduced
formula. The formula \(\ball(\overline x, x_1, r)\) is in fact an
existential formula in \(r\)-local form (with zero existential quantifiers). Therefore, we can apply Lemma \ref{lem:qelim:bexp}. We can even apply this lemma with the same \(p\) since
\(p\)  satisfies the required condition \(p \geq r|\overline x|\). Thus, there
exists a set of label \(\Lambda_\gamma\), and a \(p\)-lca-reduced formula
\(\gamma(\overline x)\) such that, for every
\(G \in \mathcal G\), and every \(p\)-skeleton \(S\) of~\(G\), there exists a
\(\Lambda_\gamma\)-labeling \(\ell_\gamma\) for which  
\[
\true(\ball(\overline x,
x_1,r)), G) = \true(\gamma, S, \ell_\gamma) = \true(\gamma \wedge \ball(\overline x,
x_1, 2^{p}), S, \ell_\gamma).
\]
Finally, let \(\widehat \Lambda =
\Lambda_\gamma \cup \Lambda_3\) where, w.l.o.g., we assume that \(\Lambda_\gamma\) and \(\Lambda_3\) are disjoint, let 
\[
\widehat
\varphi(\overline x) = \gamma(\overline x) \wedge \neg \eta_3(\overline x),
\]
and for every node \(v\) let \(\widehat \ell(v) = \ell_\gamma(v) \cup \ell_3(v)\). We get
\begin{align*}
\true(\varphi, G, \ell) 
&= \true\big(\ball(\overline x, x_1,r), G, \ell\big) \cap \true\Big(\neg\exists y\,\big(
\ball((\overline x,y),x_1,r) \wedge \neg\psi(\overline x,y)\big), G, \ell\Big) \\
&= \true\big(\ball(\overline x, x_1,r), G, \ell_2\big) \cap \true(\neg\eta_2, G, \ell_2) \\
&= \true\big(\gamma(\overline x), S, \ell_\gamma\big) \cap \true\big(\neg\eta_3(\overline x), S, \ell_3\big)\\
&= \true\big(\widehat \varphi, S, \widehat \ell\big).
\end{align*}
Moreover, since \(\true(\gamma, S, \ell_\gamma) = \true(\gamma \wedge \ball(\overline x,
x_1, 2^{p}), S, \ell)\), we also have
\[
\true(\widehat \varphi, S,
\widehat \ell) = \true\big(\widehat \varphi \wedge \ball(\overline x, x_1,
2^{p}), S, \widehat\ell\big).
\]

We complete the proof by studying the round-complexity of computing the labeling $\widehat\ell$ in \CONGEST. For each of the lemmas used in the proof, the construction of the labeling is performed in \(\mathcal O(1)\)  rounds. For cases~2 and~3, we also need
to construct a skeleton, which can be done in \(\mathcal O(\log n)\) rounds. Thus, the overall round complexity is \(\mathcal O(\log n)\), as claimed.
\end{proof}

\subsubsection{Proof of \Cref{thm:thmFOlocal}} 

Let \(\varphi(x)\) be a \(r\)-local formula. By \Cref{prop:local_form}, it can be
rewritten in \(r\)-local form. Then, by \Cref{lem:qelim:induction}, for some \(p\in \N\)
\(\varphi\) can be reduced
in \(\cO(\log n)\) rounds to a \(p\)-\(\lca\)-reduced formula \(\widehat \varphi(x)\).
Using \Cref{prop:skeleton-construction}, we can build a \(p\)-skeleton \(S\) in
\(\cO(\log n)\) rounds. Then, from an input labeling \(\ell\), the reduced labeling
\(\widehat \ell\) can be computed in \(\cO(\log n)\) rounds. It satisfies:
\[
(G,\ell) \models \varphi(v) \iff (S, \widehat \ell) \models \widehat \varphi(v)
\]

Since \(\widehat \varphi(x)\) is quantifier free, each node \(v\) can check locally
without additional communication whether it satisfies \(\widehat\varphi\), which completes
the proof.

\section{Distributed Model Checking of General Formulas}
\label{sec:general-model-checking}

In this section, we establish \Cref{thm:thmFOexpCONG-informal}, generalized to labeled graphs, as formaly stated below. 

\begin{theorem}
\label{thm:thmFOexpCONG}
    Let $\Lambda$ be a finite set. For every FO formula $\varphi$ on $\Lambda$-labeled graphs, and for every class of graphs $\mathcal{G}$ of bounded expansion, there exists a distributed algorithm that,  for every $n$-node network $G=(V,E)\in \mathcal{G}$  of diameter~$D$, and every $\ell:V\to 2^\Lambda$, decides whether $(G,\ell)\models\varphi$ in $\cO(D+\log n)$ rounds under the \CONGEST model.
\end{theorem}

Actually, we show a result stronger than  \Cref{thm:thmFOexpCONG}. Given a formula $\varphi(x)$ with one free variable, we provide a CONGEST algorithm that marks all nodes of $n$-node graphs with bounded expansion that satisfy the formula, in $\mathcal{O}(D + \log n)$ rounds, where $D$ is the diameter of the input graph\footnote{This problem is more challenging than simply checking whether the graph satisfies a formula $\varphi$, i.e., stronger than the model checking problem as, given an FO formula $\varphi$ without free variables, one can define a logically equivalent formula $\varphi'(x)$ by adding a dummy free variable, so that the graph satisfies $\varphi$ if and only if it satisfies $\varphi'(v)$ at each of its nodes~$v$.}. 
To prove the result, we need to revisit the proof of \Cref{thm:thmFOlocal}, which involves quantifier elimination of existential formulas first over forests of bounded-depth, then over graphs of bounded tree-depth, and finally over graphs with bounded expansion. Fortunately, once we have resolved quantifier elimination over  forests bounded-depth, the extensions to the other two cases will follow from the arguments already developed in the previous section for local formulas. Therefore, it will be sufficient to revisit quantifier elimination for existential formulas on forests of bounded-depth, and then proceed by induction to cover the general case.

\subsection{Rooted Forests of Bounded Depth}

As established in the proof of \Cref{lem:elimination_forest}, an existential formula over bounded-depth forests corresponds to checking the presence of ``subpatterns'' (representing subforests) in which the free variables of the formula can be mapped onto. One of these free variables, quantified existentially, is to be eliminated. Depending on the position of this existentially quantified variable~$y$, we distinguish three cases. In the first two case, $y$~shares common ancestors with the other free variables --- these cases are the ``connected cases''. In the third case, $y$~belongs to a  connected component that is different from all the other variables in the formula. When dealing with local formulas, we only needed to handle the connected cases. However, for general formulas, we must also address the remaining case. This is for this ``disconnected case'' that our algorithm requires $\mathcal{O}(D)$ additional rounds.

\begin{lemma}\label{lem:generalforestcomputing}
Let $\varphi(\overline{x})= \exists \overline{y}\; \zeta(\overline{x},\overline{y})$ be a formula in $\FO[\Lambda]$. There exists \(\widehat\Lambda\), a \(\lca\)-reduced formula
\(\widehat\varphi(\overline x)\) and a distributed algorithm performing in $\cO(D)$ rounds in the
\CONGEST model satisfying the following. For every $\Lambda$-labeled graph \((G,\ell)\) of
diameter~$D$, and for every subgraph $F\in \mathcal{F}_d$ of $G$, 
\begin{itemize}
\item every node $v\in V(G)$ is given as input its labeling $\ell(v)$, whether it belongs to~$F$, and whether it is a root of a tree in~$F$ (if not, it also receives as input its depth and parent in \(F\));

\item every node $v\in V(G)$ outputs a set of labels $\widehat{\ell}(v)$, such that
\(\widehat\Lambda, \widehat\varphi, \widehat \ell\) form a reduction of \(\Lambda, \varphi, \ell\).
\end{itemize}
\end{lemma}

\begin{proof}
We follow the same guidelines as in the proof of \Cref{lem:local_forest_computing} for Cases~1 and~2. It is thus sufficient to consider Case~3, as described in \Cref{lem:elimination_forest}. We use the same notations as in these lemmas.
Recall that Case~3 assumes $$\delta(1,k+1)=-1.$$ The algorithm consists in four phases. 

\begin{enumerate}
    \item During  the first phase, the algorithm aims at identifying the  active roots. First, the candidate nodes are identified in the same way as in the previous cases~1 and~2. By broadcasting a ping message upward each tree from the candidate nodes for $h_y$ rounds, each root becomes aware of the presence of a candidate node in its tree. The roots that receive a ping message mark themselves as active roots. This phase takes $ \cO(d)$ rounds. 

    \item During the second phase, the algorithm aims at identifying the good children, by broadcasting a ping message downward from every active root. Every node $u$ that receives such ping message on this phase marks itself as good,  and adds the label $\good$ to $\widehat{\ell}^\psi(u)$. This phase takes $\cO(d)$ rounds. At this point, the algorithm constructs a BFS tree $T$ of $G$, rooted at  node~$r_1$. This procedure takes $\cO(D)$ rounds. 

    \item During the third phase, the root $r_1$ gets the number $\rho(F)$ of active roots in~$F$, by counting upward~$T$. More specifically, let us denote by $d_T = \cO(D)$ the depth of $T$. The counting proceeds in $d_T+1$ round. For each node $u$ in~$T$, let $a(u)$ be~$1$ if $u$ is an active root, and $0$ otherwise.  In the first round, each node $u$ at depth $d_T$ set $\val(u) = a(u)$, and sends $\val(u)$ to its parent in~$T$. Then, at round $i+1$ with $i\in [d_t]$, every node $u$ of $T$ at depth $d_T-i$ computes $\val(u) = a(u) + \sum_{v\in C(u)} \val(u)$, where $C(u)$ is the set of children of $u$ in~$T$, and sends $\val(u)$ to its parent. At the end of this phase,  we have $\val(r_1)=\rho(F)$.  Observe that, for every node~$u$,  $\val(u)\leq n$, and thus this value can be communicated in one round under the \CONGEST model. Therefore, the third  phase runs in $\cO(D) + d_T+1 = \cO(D)$ rounds. 

    \item Finally, during the fourth phase, the root $r_1$ broadcast $\rho(F)$ downward the tree~$T$, to all nodes in $G$. Then, each node $u$ in $F$ adds the label $\min(\rho(F),k+1)$ to $\widehat{F}^\psi(u)$. This phase takes $d_T = O(D)$ rounds.
\end{enumerate}

We conplete the proof by noticing that, for every $\psi \in I$, every node $u$ can compute $\ell^\psi(u)$ in $\cO(D)$ rounds. Therefore,  the total number of rounds sufficient for computing $\widehat{\ell}$ is $$|I|\cdot D \leq  c\cdot 2^{k|\Lambda|}(d+1)^{k^2} \cdot D,$$ where $c$ is a constant not depending on either~$d$, $\Lambda$, $\psi$, or~$D$.
\end{proof}

\subsection{Quantifier Elimination on Graphs of Bounded Treedepth}

We now deal with the quantifier elimination on graphs of bounded treedepth.  

\begin{lemma}
\label{lem:qelim:td-general}
Let \(\Lambda\) be a set of labels, and let 
\(\varphi(\overline x) = \exists \overline y, \zeta(\overline x, \overline y)\) be an existential
formula in \(\FO[\Lambda]\), where  \(\zeta\) is a quantifier free formula.
There exists a set of labels \(\widehat \Lambda\), and a \(\lca\)-reduced formula
\(\widehat \varphi(\overline x)\) such that, for every induced subgraph \((H,\ell)\) of treedepth~\(t\) 
of a \(\Lambda\)-labeled graph \((G,\ell)\), and for every decomposition forest \(F\) of~\(H\),
there exists a \(\widehat \Lambda\)-labeling \(\widehat \ell\) such that:
\[
\true(\varphi, H, \ell)
= \true(\widehat\varphi, F, \widehat \ell)
\]
Moreover, there exists an algorithm computing \(\widehat \ell\) in \(\mathcal O(D)\) rounds in the
\CONGEST\ model, each node \(u\) is given as input \(\ell(u)\), whether it belongs to \(H\)
and, if applicable, its depth and parent in~\(F\).
\end{lemma}

\begin{proof}
By \Cref{lem:treedepthtoforest}, there exists an \(\lca\)-reduced formula
\(\zeta_1(\overline x, \overline y)\) such
that, for every induced labeled subgraph \( (H,\ell)\) of treedepth at most~\(t\), for every
decomposition forest \(F\) of~\(H\), there exists a labeling \(\ell_1\) such that
\(\true(\zeta, H, \ell) = \true(\zeta_1, F, \ell_1)\), and \(\ell_1\) can be computed in
\(\mathcal O(1)\) rounds in the  \CONGEST\  model. It is thus sufficient to  merely  apply Lemma~\ref{lem:generalforestcomputing} to the formula \(\exists
\overline y, \zeta_1(\overline x, \overline y)\), which takes $\cO(D)$ \CONGEST  rounds. 
\end{proof}

\subsection{Graphs of Bounded Expansion}

We can now directly treat the case of graphs of bounded expansion. As we did in the local case, we consider first the case of existential formulas. Let $\mathcal{G}$ be a class of graphs, of bounded expansion. 

\begin{lemma}\label{lem:reductionofexistential} 
Let $\Lambda$ be a set of labels, and let $\varphi(\overline{x}) = \exists \overline y \,  \psi(\overline x, \overline y)$ be an existential formula in $ \FO[\Lambda]$, where $\psi$ is quantifier-free. For every $p\geq |\overline{x}| + |\overline y|$, there exists a set $\widehat{\Lambda}$ of labels, and a $p$-$\lca$-reduced formula $\widehat{\varphi}$ such that, for every $(G,\ell)\in \mathcal{G}[\Lambda]$, and for every $p$-skeleton $S$ of~$G$, there exists a $\widehat{\Lambda}$-labeling $\widehat{\ell}$ such that
\[
\true(\varphi, G,\ell) = \true(\widehat{\varphi}, S, \widehat{\ell}).
\]
Moreover, there exists a distributed algorithm that computes \(\widehat \ell\) in \(\mathcal O(D)\)
rounds in  \CONGEST, assuming \(\ell\) and \(S\) are given as input.
\end{lemma}

\begin{proof}
Let us fix an integer \(p \geq |\overline x|+|\overline y|\). By applying Lemma \ref{lem:qelim:td-general} to~\(\varphi\), we get that 
there exist a set of labels \(\Lambda_1\), and an \(lca\)-reduced formula
\(\varphi_1(\overline x)\) such that, for every labeled
graph \( (H, \ell)\) of treedepth at most \(p\), and for every decomposition forest \(F\)
of~\(H\), there exists a \(\Lambda_1\)-labeling \(\ell_1\) of \(F\) such that
\[
\true(\varphi, H, \ell) = \true(\varphi_1, F, \ell_1).
\]
Let us now consider a \(p\)-skeleton~\(S\) of $(G,\ell)$.
For each \(U\in \Col(p)\), the induced subgraph \(G[U]\) has treedepth at most~\(p\), and \(F^U\) is a decomposition forest of \(G[U]\). Thus, the corresponding labeling
\(\ell_1^U\) can be computed in \(\mathcal O(D)\) rounds. Let us consider the \(p\)-lca-reduced
formula 
\[
\widehat\varphi^U(\overline x)=\varphi_1^U(\overline x)\wedge \col_U(\overline x).
\]
Note that \(\varphi_1^U\) is the formula \(\varphi_1\) in which all
\(\lca\) predicates were replaced by \(\lca^U\) predicates.
We also define
\[
\widehat\varphi(\overline x) = \bigvee_{U \in \Col(p)}
\widehat\varphi^U(\overline x).
\]
Finally, for each node \(u\), we set 
\[
\widehat \ell(u) = \bigcup_{U \in \Col(p)} \ell_1^U(u),
\] 
where it is assumed, w.l.o.g., that the formulas
\(\varphi^U\) use disjoint  sets of labels.
Let us now prove that the formula $\widehat\varphi(\overline x)$ has the desired property. Let \( (G,\ell)\in
\mathcal G[\Lambda]\), let \(S\) be a \(p\)-skeleton of~\(G\), and let us prove that 
\[
(G,\ell)\models \varphi(\overline{v}) \iff (S,\widehat \ell) \models \widehat\varphi(\overline v). 
\]
For this purpose, let us assume  that \(
(G,\ell)\models \psi(\overline v, \overline w)\) for some \(\overline v,
\overline w \subseteq V(G)\). Since \(p \geq |v|+|w|\), there exists \(U
\in \Col(p)\) such that \(G[U]\) contains both $\overline v$ and $\overline w$. Therefore, \((G[U],\ell) \models \psi(\overline v, \overline
w) \). Furthermore,
\begin{align*}
 (G,\ell)\models \varphi(\overline v) 
&\iff \exists U\in\Col(p), \overline v \subseteq G[U], 
	(G[U], \ell) \models \varphi(\overline v) \\
&\iff \exists U\in\Col(p), \overline v \subseteq G[U],
	(F^U, \ell_1^U) \models \varphi_1(\overline v) \\
&\iff \exists U\in \Col(p),
	(S, \ell_1^U) \models \widehat\varphi^U(\overline v) \\
&\iff (S, \widehat \ell) \models \widehat\varphi(\overline v)
\end{align*}
where the second equivalence is thanks to Lemma \ref{lem:qelim:td-general}. 
\end{proof}

Next, we adapt \Cref{lem:qelim:remove-lca} to the case of general (i.e., non necessarily local) formulas. 

\begin{lemma}
\label{lem:qelim:remove-lca-general}
Let \(p\in \N\), let \(\Lambda\) be a set of labels, and let \(\varphi(\overline x)\) be a
\(p\)-\(lca\)-reduced formula. There exists a set of labels \(\widehat
\Lambda\), and  a quantifier free formula \(\psi(\overline x, \overline y)\) such
that, for every \(G\in \mathcal G\),  every \(p\)-skeleton \(S\) of~\(G\), and every 
\(\Lambda\)-labeling \(\ell\) of \(G\), there exists a labeling \(\widehat \ell\) of \(G\) for which
\[
\true(\varphi, S, \ell)
= \true(\widehat\varphi, G, \widehat \ell).
\]
Moreover, for every $p$, $\Lambda$ and $\varphi$, there exists a distributed algorithm that,
given $S$ and $\ell$ as inputs, computes $\widehat{\ell}$ without any communication.
\end{lemma}

\begin{proof}
This lemma is a direct consequence of the arguments in the proof of \Cref{lem:qelim:remove-lca}. Indeed, we follow the same steps as in that proof, until the predicates \(\lca_i^U(y,z)\) has to be replace by the similar existential formula
\begin{align*}
\exists (y_0, \dots, y_d) \; & \exists (z_0, \dots, z_d)\;
\bigvee_{(a,b) \in [i, 2^p-1]^2} \Biggl( (y_a = y) \wedge (z_b = z)
\\
& \wedge \Big(\bigwedge_{s\in [0,i]} y_s = z_s\Big)
\wedge \Big(\bigwedge_{s\in [i+1, \min(a,b)]} y_s \neq z_s\Big)
\wedge \depth_0^U(y_0) \wedge \depth_0^U(z_0)
\\
& \wedge \Big(\bigwedge_{s\in [1,a]} \depth_s^U(y_s) \wedge \adj(y_{s-1}, y_s)\Big)
\wedge \Big(\bigwedge_{s\in [1,b]} \depth_s^U(z_s) \wedge \adj(z_{s-1}, z_s)\Big)\Biggr)
\end{align*}
This completes the proof. 
\end{proof}

We are now ready to deal with quantifier elimination of general first-order formulas. 

\begin{lemma}\label{lem:reductiongeneral}
Let $\Lambda$ be a set of labels, and let $\varphi(\overline x)$ be a formula in $\FO[\Lambda]$. There exists $p\in \mathbb{N}$, a set of labels $\widehat{\ell}$, and a $p$-$\lca$-reduced formula $\widehat{\varphi}(\overline{x})$ such that, for every $(G,\ell)\in \mathcal{G}[\Lambda]$, and every $p$-skeleton $S$ of~$G$, there exists a $\widehat{\ell}$-labeling $\widehat{\ell}$ of~$G$ such that
\[ 
\true(\varphi,G,\ell) = \true(\widehat{\varphi},S,\widehat{\ell}).
\]
Moreover, there exists a distributed algorithm that computes the labeling $\widehat{\ell}$ in
$\mathcal{O}(D + \log n)$ rounds, assuming \(\ell\) and \(S\) are given as input.
\end{lemma}

\begin{proof}
We prove the lemma by designing the algorithm by induction on the quantifier depth $t$ of the formula. For the base case, \(\varphi(\overline x)\) is quantifier-free. Let  $p= |\overline{x}|$. Thanks to \Cref{prop:skeleton-construction}, the $p$-skeleton $S$ can be computed in $\cO(\log n)$ rounds.  We can then apply \Cref{lem:reductionofexistential} with \(|\overline y|=0\), and return the reduction of $(\Lambda, \varphi, \ell)$ on~$S$. 

Now let us suppose that the statement of the lemma holds for all formulas with quantifier depth~$t\geq 0$. Let $\varphi(\overline x) = Qy, \psi(\overline x, y)$, where $Q \in \{\exists, \forall\}$, and $\psi$ is a formula of quantifier depth~$t$. We study $Q=\exists$ and $Q=\forall$ separately.

\begin{description}
    \item[Existential quantifier.] By induction hypothesis, an integer \(p_1\), a set \(\Lambda_1\),  a \(p_1\)-\(\lca\)-reduced
formula \(\psi_1(\overline x, y)\), a set of labels $\ell_2$, and a $p_1$-skeleton $S$ such that \[\true(\psi, G, \ell) = \true(\psi_1, S, \ell_1)\] 
can be computed in $\cO(D+ \log n)$ rounds.
One would then like to apply Lemma \ref{lem:reductionofexistential} on $\exists y\, \psi_1(\overline{x}, y)$. However, $\psi_1$  is not defined over graphs of bounded expansion, but over $p_1$-skeletons. Therefore, we first apply Lemma~\ref{lem:qelim:remove-lca-general} on
\(p_1, \Lambda_1, \psi_1, G, S, \ell_1\), to compute  \(\Lambda_2, \psi_2(\overline x, y)\), and $\ell_2$ satisfying
\[
\true(\psi_1, S, \ell_1) =
\true( \psi_2, G, \ell_2)
\]
in $\cO(\log n)$ rounds, 
where $\psi_2$ is existential, of the form  \(\psi_2(\overline x, y) = \exists
\overline z, \rho(\overline x, y, \overline z)\), where $\rho \in \FO[\Lambda_2]$ is quantifier-free. 
Now we can define $\varphi_2(\overline{x}) = \exists y\, \psi_2(\overline{x}, y)$. Note that $\varphi_2$ is also existential,  and
\[
\true(\varphi_2, G, \ell_2) = \true(\varphi, G, \ell).\]
We can finally apply Lemma~\ref{lem:reductionofexistential} to compute a reduction
\((p, \widehat\Lambda, \widehat \varphi, \widehat\ell)\) of $(\Lambda_2, \varphi_2, \ell_2)$.

\item[Universal quantifier.]  In this case, we merely
rewrite \(\varphi\) as:
\[
\varphi(\overline x) = \neg\exists y\,\neg \psi(\overline x, y).
\]
By induction hypothesis, an integer \(p_1\), a set \(\Lambda_1\), a \(p_1\)-\(\lca\)-reduced
formula \(\psi_1(\overline x, y)\), a $\Lambda_1$-labeling $\ell_1$ of $G$, and a $p_1$-skeleton $S_1$ of $G$ can be computed in $\cO(D+ \log n)$ rounds, such that 
\[
\true(\psi, G, \ell) = \true(\psi_1, S_1, \ell_1).
\] 
The negation of \(\psi_1\) is still a \(p_1\)-\(\lca\)-reduced formula. Let us apply Lemma~\ref{lem:qelim:remove-lca-general} on $p_1$, $\Lambda_1$, $\neg \psi_1$, $G$, $S_1$, and $\ell_1$ to compute a set \(\Lambda_2\),  an existential formula \(\psi_2\), and a
\(\Lambda_2\)-labeling \(\ell_2\) of \(G\)  in $\cO(\log n)$ rounds, such that 
\[
\true(\neg\psi_1 , S_1, \ell_1) = \true(\psi_2, G,
\ell_2).
\] 
Now we can define $\rho(\overline{x}) = \exists y\, \psi_2(\overline{x}, y)$. Observe that this formula is existential, and it satisfies
\[
\true(\varphi, G, \ell) = \true(\neg \rho, G, \ell_2).
\]
We can eventually apply Lemma~\ref{lem:reductionofexistential} to compute the reduction $(p, \widehat{\Lambda}, \widehat{\rho}, \widehat{\ell})$ of $(\Lambda_2, \rho, \ell_2)$  in $\cO(D + \log n)$ additional rounds. In particular, $\widehat{\rho}$ is $p$-lca-reduced, as well as $\neg \widehat{\rho}$. Therefore $(p, \widehat{\Lambda}, \neg \widehat{\rho}, \widehat{\ell})$ is the reduction of $(\Lambda, \varphi, \ell)$.
\end{description}

We obtain that eliminating each quantifier of $\varphi$ can be done in $\cO(D + \log n)$ rounds.
\end{proof}

\subsection{Proof of \Cref{thm:thmFOexpCONG}} 

Let \(\varphi\) be a formula in $\FO[\Lambda]$. We define the formula $\varphi'(x)$ with  one (dummy) free variable by \[\varphi'(x) = \varphi \wedge (x=x).\] We do have that, for every labeled graph $(G,\ell)$,  and for every node $v\in V(G)$, 
\[
(G,\ell)\models\varphi'(v) \iff (G,\ell)\models \varphi.
\]
By \Cref{lem:reductiongeneral}, $\varphi'(x)$  can be reduced to a \(p\)-\(\lca\)-reduced formula \(\widehat \varphi(x)\) 
in \(\cO(D+\log n)\) rounds. Since \(\widehat \varphi(x)\) is quantifier free, each node \(v\) can check locally
(without additional communication) whether it satisfies \(\widehat\varphi\).
By definition of the reduction, the truth of \(\widehat \varphi\) is equivalent to
that of \(\varphi'\), which completes the proof.\qed

\section{Counting and Optimization Problems} 
\label{se:countopt} 

In this section, we prove Theorems~\ref{thm:thmFOexpCNT-informal} and~\ref{thm:thmFOexpOPT-informal}. We restate these two theorems here, in their full generality of labeled graphs. 

\begin{theorem}\label{thm:thmFOexpCNT}
Let $\Lambda$ be a finite set. For every FO formula $\varphi(x_1,\dots,x_k)$ on $\Lambda$-labeled graphs with $k\geq 1$ free variables $x_1,\dots,x_k$, and for every class of graphs $\mathcal{G}$ of bounded expansion, there exists a distributed algorithm that, for every $n$-node network $G=(V,E)\in \mathcal{G}$ of diameter~$D$ and every  $\ell:V\to 2^\Lambda$, performs global counting in ${O(D+\log n)}$ rounds in \CONGEST. 
If the formula $\varphi(x_1,\dots,x_k)$  is local then there exists a distributed algorithm that performs local counting in $\cO(\log n)$ rounds in \CONGEST. 
\end{theorem}

\begin{theorem}\label{thm:thmFOexpOPT}
Let $\Lambda$ be a finite set. For every FO formula $\varphi(x_1,\dots,x_k)$ on $\Lambda$-labeled graphs with $k\geq 1$ free variables $x_1,\dots,x_k$, and for every class graphs $\mathcal{G}$ of bounded expansion, there exists a distributed algorithm that, for every $n$-node network $G=(V,E)\in \mathcal{G}$ of diameter~$D$, every $\ell:V\to 2^\Lambda$, and  every weight function $\omega:V\to \mathbb{N}$, computes a $k$-tuple of vertices $(v_1,\dots,v_k)$ such that $(G,\ell)\models\varphi(v_1,\dots,v_k)$ and $\sum_{i=1}^k\omega(v_i)$ is maximum, in $\cO(D+\log n)$ rounds  under the \CONGEST model. (The same holds if replacing maximum by minimum.) If no tuples $(v_1,\dots,v_k)$ exist such that $(G,\ell)\models\varphi(v_1,\dots,v_k)$, then all nodes reject. 
\end{theorem}

To ease the notations, let us adopt the following notations. 
\begin{itemize}
    \item Let $\OPT_{\varphi}(G,L)$ denote the optimum value for $\varphi$, i.e., the maximum value of $\sum_{i=1}^k \omega(v_i)$ over all $k$-uples of vertices $(v_1,\dots,v_k)$ such that $G \models \varphi(v_1,\dots,v_k)$. 
    
    \item Let $\cnt_{\varphi}(G,L)$ denote the number of $k$-uples satisfying the formula. 

    \item Given some $(p,f(p))$-treedepth coloring of $G$, and a set of colors $U$, let $\OPT_{\varphi}(G[U],L)$ and $\cnt_{\varphi}(G[U],L)$ be the same quantities as above, but assuming that the formula is restricted to the subgraph $G[U]$ of $G$ induced by the colors in~$U$.
\end{itemize}
Let us prove the following lemma, which is restricted to
quantifier-free formulas.

\begin{lemma}\label{lem:QF-counting}
Let $\Lambda$ be a finite set. For every \emph{quantifier-free} FO formula $\varphi(x_1,\dots,x_k)$ on $\Lambda$-labeled graphs with $k\geq 1$ free variables $x_1,\dots,x_k$, and for every class of graphs $\mathcal{G}$ of bounded expansion,
there exists a distributed algorithm that, for every $n$-node network $G=(V,E)\in \mathcal{G}$
of diameter~$D$, and every  $\ell:V\to 2^\Lambda$, performs global counting in ${O(D+\log n)}$
rounds  under the \CONGEST model. If the formula is local, then there exists an
algorithm that performs local counting in \(\cO(\log n)\) rounds under the \CONGEST model.
\end{lemma}

\subsection{Proof of Lemma~\ref{lem:QF-counting} and Theorem~\ref{thm:thmFOexpOPT}}

We first consider~\Cref{lem:QF-counting}, on counting. In this case, the formula~$\varphi$ has exactly
$k$ variables, namely $x_1,\dots,x_k$. We now apply 
\Cref{th:congestltd} to compute, in $\cO(\log n)$ rounds, a $(k,f(k))$-treedepth coloring
of $G$. Since formula $\varphi$ has exactly $k$ variables, a $k$-uple of vertices $(v_1\dots,v_k)$
satisfies the formula in $(G,\ell)$ if and only if it satisfies it in some set $G[U]$. By definition,
$\cnt_{\varphi}(G[U],\ell)$ is the size of the set $\true(\varphi,G[U],\ell)$ of $k$-uples
$(v_1,\dots, v_k)$ such that $(G[U],\ell) \models \varphi(v_1,\dots,v_k)$. An easy but crucial
observation is that, for any two sets of colors $U_1, U_2$, we have that 
\[
\true(\varphi,G[U_1],\ell)
\cap \true(\varphi,G[U_2],\ell) = \true(\varphi,G[U_1\cap U_2],\ell),\]
because $\varphi(x_1,\dots,x_k)$
has no other variables than the free ones. The condition fails if one allows other variables,
which explains that the technique does not extend directly to arbitrary FO formulas;
this is the subject of \Cref{ss:counting-FO}. The equality above allows us to use the inclusion-exclusion
principle in order to count the solutions of $\varphi$, assuming that we were able to count the partial
solutions over subgraphs $G[U]$, where $U$ are sets of at most $k$ colors.

\begin{lemma}\label{lem:countie}
Let $\varphi(x_1,\dots,x_k) \in FO(\Lambda)$ be a quantifier-free formula. Let us consider a $(k,f(k))$-treedepth coloring of $(G,\ell)\in \mathcal{G}[\Lambda]$. Let $\kappa = \binom{f(k)}{k}$, and let $U_1\dots U_{\kappa}$ be all subsets of $[f(k)]$ of size~$k$. Then 
$$
\cnt_{\varphi}(G,\ell) = \sum_{h=1}^{\kappa}(-1)^{h-1} \sum_{1 \leq i_1 \leq \dots \leq i_{h} \leq \kappa} \cnt_{\varphi}(G[U_{i_1} \cap \dots \cap U_{i_{h}}],\ell).
$$
\end{lemma}
\begin{proof}
    Recall that for any $i \in [\kappa]$, $\cnt_{\varphi}(G[U_i],\ell)$ is the size of the set $\true(\varphi,G[U_i],\ell)$ of $k$-uples $(v_1,\dots, v_k)$ such that $(G[U_i],\ell) \models \varphi(v_1,\dots,v_k)$. Also, $\cnt_{\varphi}(G[U_{i_1} \cap \dots \cap U_{i_{h}}],\ell)$ is the size of the set $\true(\varphi,G[U_{i_1} \cap \dots \cap U_{i_{h}}],\ell)$, which is itself equal to  $\true(\varphi,G[U_{i_1}]) \cap \dots \cap \true(\varphi,G[U_{i_{h}}],\ell)$. 
    Eventually, $\true(\varphi,G,\ell) = \true(\varphi,G[U_1],\ell) \cup \dots \cup \true(\varphi,G[U_{\kappa}],\ell)$ and $\cnt_{\varphi}(G,\ell) = |\true(\varphi,G,\ell)|$, so the formula is a straightforward application of the inclusion-exclusion principle.
\end{proof}

Before explaining (in~\Cref{lem:FFMRText}) how to compute the values $\cnt_{\varphi}(G[U],\ell)$, we address  optimization as stated in~\Cref{thm:thmFOexpOPT}. 
Recall that, in general, the formula $\varphi$ may have other (quantified) variables than its free valiables. 
We first transform $\varphi(\overline{x})$ into an equivalent existential formula $\tilde{\varphi}(\overline{x})$ at the cost of enriching the set of labels, and the labeling of the input graph.

\begin{lemma}\label{lem:existentialOnly}
    Given a formula $\varphi(\overline{x}) \in FO[\Lambda]$, and a labeled graph $(G,\ell) \in \mathcal{G}[\Lambda]$, a tuple $(\tilde{\Lambda},\tilde{\varphi},\tilde{\ell})$ such that 
    \begin{enumerate}
        \item $\tilde{\ell}$ is a $\tilde{\Lambda}$-labeling of $G$,
        \item $\tilde{\varphi}(\overline{x})$ is an existential formula in $FO[\tilde{\Lambda}]$, i.e., $\tilde{\varphi}(\overline{x}) = \exists\overline{y}\, \psi(\overline{x},\overline{y})$ where $\psi \in \tilde{\Lambda}$ is quantifier-free, and
        \item $\true(\varphi,G,\ell) = \true(\tilde{\varphi},G,\tilde{\ell})$, 
    \end{enumerate}
    can be computed in $\cO(D + \log n)$ rounds  in \CONGEST.
\end{lemma}
\begin{proof}
    Thanks to~\Cref{lem:reductiongeneral} (see also~\Cref{def:reducibleonBE}), there exists some constant
    $p$ such that, for any $p$-skeleton $S$ of $G$
    (obtained by~\Cref{prop:skeleton-construction}), the tuple $(\Lambda,\varphi,\tau_p,\ell)$ can be
    reduced  to a tuple  $(p,\widehat{\Lambda},
    \widehat{\varphi},\widehat{\ell})$ in $\cO(D + \log n)$ rounds, where $\widehat{\varphi}(\overline{x})$ is a $p$-$\lca$-reduced
    (thus quantifier-free)
    formula in $FO(\widehat{\Lambda})$, $\widehat{\ell}$ is a $\widehat{\Lambda}$
    labeling of \(S\), and $\true(\varphi,G,\ell) = \true(\widehat{\varphi},S,
    \widehat{\ell})$.
    Now, given the tuple $(p,\widehat{\Lambda},\widehat{\varphi},\widehat{\ell})$,
    \Cref{lem:qelim:remove-lca-general} states that the desired
    tuple $(\tilde{\Lambda},\tilde{\varphi},\tilde{\ell})$ can be computed with no extra communication.
\end{proof}

Thanks to  \Cref{lem:existentialOnly}, we may now assume, w.l.o.g., that, in the proof of \Cref{thm:thmFOexpOPT}, the formula $\varphi(\overline{x}) = \exists \overline{y} \psi(\overline{x},\overline{y})$ is given in an existential form, where  $\psi \in FO[\Lambda]$ is quantifier-free.

Let $p = |\overline{x}| + |\overline{y}|$. We use again~\Cref{th:congestltd} to
compute a $(p,f(p))$-treedepth coloring of $G$ in $\cO(\log n)$ rounds. Since $\varphi$ has exactly $p$ variables, a $k$-uple of vertices $(v_1\dots,v_k)$ satisfies the
formula in $(G,\ell)$ if and only if it satisfies it in some subgraph $G[U]$, where $U$ can be viewed as obtained
by “guessing” the colors of $(v_1,\dots,v_k)$ as well as the colors of the quantified variables
$(y_1,\dots, y_{p-k}) = \overline{y}$. This is formally stated as follows, for further references. 

\begin{lemma}\label{lem:optcol}
Let $\varphi \in FO(\Lambda)$ be an existential formula, i.e., $\varphi(\overline{x}) = \exists \overline{y}, \psi(\overline{x},\overline{y})$, where  $\psi \in FO[\Lambda]$ is quantifier-free. Let $p = |\overline{x}| + |\overline{y}|$, and let us assume a given a $(p,f(p))$ coloring of $(G,\ell) \in \mathcal{G}[\Lambda]$. We have
$$
\OPT_{\varphi}(G,\ell) = \max_{U\subseteq [f(p)],\; |U|\leq p} \OPT_{\varphi}(G[U],\ell). 
$$ 
\end{lemma}

It remains to compute $\OPT_{\varphi}(G[U],\ell)$ and $\cnt_{\varphi}(G[U],\ell)$ for every set $U$ of at most $p$ colors, given some $(p,f(p))$-treedepth coloring of~$G$. This is partially done in~\Cref{thm:FFMRT24} thanks a previous result in~\cite{FominFMRT24} in which it is shown that  $\OPT_{\varphi}(G,\ell)$ and $\cnt_{\varphi}(G,\ell)$ can be computed  in a constant number of rounds in connected graphs of bounded treedepth.
We however need to extend their techniques to fit our purposes, i.e., to deal with \emph{disconnected} graphs of bounded treedepth (the non-necessarily connected subgraphs of $G$ induced by a set of colors~$U$). For global counting, the cost of communication among different components adds an  additive term $\cO(D)$ to the round complexity.

\begin{lemma}\label{lem:FFMRText}
    Let $\mathcal{G}[\Lambda]$ be a class of labeled graphs of bounded expansion, let $\varphi(\overline{x})$ be a $FO(\Lambda)$ formula, and $p\in\mathbb{N}$. Let us assume given  a $(p,f(p))$-treedepth coloring of $(G,\ell) \in \mathcal{G}[\Lambda]$, and let $U\subseteq [f(p)]$ be any set of at most $p$ colors. Then $\OPT_{\varphi}(G[U],\ell)$ and $\cnt_{\varphi}(G[U],\ell)$ can be computed in $\cO(D)$ rounds in \CONGEST.
    Moreover, if \(\varphi\) is local, then local counting in \(G[U]\) can be solved in \(\cO(1)\)  rounds.
\end{lemma}

\begin{proof}
Let us first consider the case of local counting, that is, we assume \(\varphi\) is local. In this case,  \(\cnt_\varphi(G[U], \ell)\) is equal to the sum of \(\cnt_\varphi\) over all
connected components of \(G[U]\), i.e., 
\[
\cnt_\varphi(G[U], \ell) = \sum_{H \in \mbox{CC}(G[U])} \cnt_\varphi(H, \ell),
\]
where \(\mbox{CC}(G[U])\) denotes the set of connected components of~\(G[U]\). Each component \(H\in CC(G[U])\) has treedepth at most \(|U|\), and thus, by \Cref{thm:FFMRT24},
\(\cnt_\varphi(H,\ell)\) can be computed in \(\cO(1)\) rounds. One arbitrary node \(v_H \in H\) outputs \(\nu(v_H) = \cnt_\varphi(H,\ell)\), whereas every other node~$u$
outputs \(\nu(u) = 0\).
Since \(H\) has bounded treedepth and thus bounded diameter, we can choose \(v_H\) to be for
example the node in \(H\) with maximal identifier.
It follows that
\[
\sum_{v \in V} \nu(v)
= \sum_{H \in \mbox{CC}(G[U])} \nu(v_H)
= \sum_{H \in \mbox{CC}(G[U])} \cnt_\varphi(H, \ell)
= \cnt_\varphi(G[U],\ell)
\]

Let us now focus on counting and optimization for an arbitrary (non-necessarily local) formula.

To prove the lemma, we generalize the techniques of~\cite{FominFMRT24} for proving \Cref{thm:FFMRT24}. 
We will only restate here the main features of the proof in~\cite{FominFMRT24}, based on tree decomposition and treewidth-related techniques. Recall that the tree decomposition of a graph $H=(V,E)$ is a pair $$(T = (I,F), \{B_i\}_{i \in I})$$ where $T$ is a tree, and, for every $i \in I$, $B_i$~is a subset of vertices of~$V$, referred to as the \emph{bag} of node~$i$ in~$T$, or merely ``bag~$i$'', such that
\begin{itemize}
    \item every vertex of $H$ is contained in at least one bag,
    \item for every edge of $H$, at least one bag contains its two endpoints, and
    \item for every vertex $v$ of $H$, the nodes of $T$ whose bags contain $v$ form a connected subgraph of~$T$.
\end{itemize}
The width of the decomposition is the maximum size of a bag, minus one. The treewidth of $H$ is the minimum width over all its tree decompositions.
Note that graphs of treedepth at most~$p$ have treewidth at most $2^p$~\cite{NesetrilOdM12}. The proof of \Cref{thm:FFMRT24} is based on three ingredients.

\begin{enumerate}
    \item The first ingredient is a dynamic programming centralized algorithm for decision, optimization and counting solutions of a given formula $\varphi(x) \in FO(\Lambda)$, over labeled graphs $(H,\ell)$ of bounded treewidth. Here we will assume that a tree decomposition of constant width is part of the input, and moreover the tree is rooted. The distributed algorithm used in~\cite{FominFMRT24} is based on the centralized algorithm by Borie, Parker and Tovey~\cite{BoPaTo92}. The nodes of the tree decomposition $(T = (I,F), \{B_i\}_{i \in I})$ are processed bottom-up, from the leaves to the root of $T$. When processing node $i$, let $j_1,\dots,j_q$ be its children in $T$. The essential observation due to~\cite{BoPaTo92} is that the dynamic programming tables at node $i$ are of constant size for decision, and of $\cO(\log n)$ size for counting and optimization. Moreover, for computing the information (encoding partial solutions) at node $i$, one only needs to know the subgraphs $H[B_i]$ and $H[B_{j_1}], \dots, H[B_{j_q}]$ induced by the bags at node $i$ and its children, together with the information already computed at the children nodes. In particular, when $i$ is a leaf, $H[B_i]$ suffices for computing the dynamic programming table at node $i$.

    \item The second ingredient of~\cite{FominFMRT24} is an implementation of the above centralized algorithm in the \CONGEST model, in a number of rounds equal to the depth of the tree decomposition $T$,
    under the condition that one can associate to each node $i$ of $T$ a vertex $h(i)$ of $H$, such that vertices $h(i)$ satisfy several properties in order to “simulate the behavior” of node $i$ in the dynamic programming. Vertex $h(i)$ must know the depth of tree $T$ and the depth of node $i$ in $T$, the bag $B_i$ and the subgraph $H[B_i]$, and moreover $h(i)$ knows the node $h(p_T(i))$ associated to the parent $p_T(i)$ of $i$ in $T$, and vertices $h(i), h(p_T(i))$ are equal or adjacent in $H$. Informally, the distributed algorithm implements the centralized one by making vertex $h(i)$ to help it computing the right moment (i.e., at the round corresponding to the depth of $T$ minus the depth of $i$ in $T$, plus one) the dynamic programming table at node $i$ is processed, and the result transmitted to $h(p_T(i))$. Observe that $h(i)$ has collected all the information from the children $j_1, \dots, j_q$ of $i$, thanks to the fact that  $h(j_1), \dots, h(j_q)$ have already computed and communicated it at the previous round. (To be complete, in~\cite{FominFMRT24}, one actually has $h(i)=i$ and $T$ is a subgraph of $H$, but our generalization is a simple observation.) This explains that the number of rounds of the distributed algorithm is the depth of the tree decomposition $T$.

    \item The third ingredient in~\cite{FominFMRT24} is an algorithm computing  the desired tree decomposition \linebreak$(T = (I,F), \{B_i\}_{i \in I})$ of graph $H$ of treedepth at most $p$ in a constant number of rounds, together with the information requested by the previous item. The width of the tree-decomposition is at most $p$, and the depth of $T$ is at most $2^p$.
\end{enumerate}

Let us now come back to our framework, that is when $H=G[U]$ is the subgraph of $G$ induced by some set $U$ of at most $p$ colors of some $(p,f(p))$-treedepth decomposition of $G$. In particular $H=G[U]$ has treedepth at most~$p$. Our goal is to replace the third item above by a distributed algorithm running in $\cO(D)$ rounds in \CONGEST for computing a suitable tree decomposition of $G[U]$. We compute a $p$-skeleton using \Cref{prop:skeleton-construction}. Note that this can be done in $\cO(2^{2p})$ rounds, hence in a constant number of rounds whenever the coloring is provided.
Let $F^U$ be the resulting decomposition forest spanning $G[U]$, of depth at most $2^d$. Recall that each node knows its parent and depth in the forest.

Generalizing~\cite{FominFMRT24}, we transform the decomposition forest $F^U$ into a tree decomposition of $G[U]$. This tree decomposition $(T = (I,F), \{B_i\}_{i \in I})$ has to be computed in time $\cO(D)$ and to satisfy all the conditions of the second ingredient above. The only difference is that, for each node $i$ of $T$, its “leader” $h(i)$ has to be a vertex of $G$ (not necessarily of $G[U]$), and, for each pair of nodes $i,j$ of $T$ such that $j = p_T(i)$, vertices $h(i)$ and $h(j)$ must be adjacent in $G$. In our construction, there may be up to two different nodes of $T$ having the same leader in~$G$.

Let us choose an arbitrary vertex $r$ of $G$, and compute a breadth-first search tree of $G$. From this rooted tree, remove every node $v$ that does not have a descendant which is a root in \(F^U\). Note that the construction of this tree can be performed in time $\cO(D)$ rounds. We call $TS$ the resulting tree, rooted at $r$. $TS$ can be viewed as a Steiner tree in $G$ spanning $r$ and all the vertices labeled $\root^U$. The depth of $TS$ is, by construction, at most~$D$, the diameter of $G$. 

From now on, let us focus on $TS$ alone, independently from the the fact that it is a subgraph of~$G$. Let us then construct a ``virtual'' tree $T$ by augmenting $TS$ a follows. For every vertex labeled $\root^U$ in $TS$, we root the  corresponding tree $T^U$ to that vertex. So, such a vertex has two types of children in $T$, those originally in $TS$, and those in the tree $T^U$ rooted at that vertex. Note that two different vertices $i$ and $j$ in $T$ may correspond to the same node $v$ in~$G$, as node~$i$ may be $v$ as a node in $TS$, and $j$ may be $v$ again, but as a node of~$T^U$. 
For every node $i$ of $T$, let $h(i)$ be the corresponding vertex of $G$.  Note that, by a simple distributed algorithm working in $\cO(D)$ rounds, each node $h(i)$ of $G$ can learn the depth of~$T$, the depth of $i$ in~$T$, and the identifier of $h(p_T(i))$. The depth of $T$ is at most $D + 2^p$ because the forest $F$ has depth at most $2^p$. By construction, $h(i)$ and $h(p_T(i))$ are either adjacent in~$G$ or equal, whenever $i$ is labeled $\root^U$.  

It remains to construct the bag $B_i$ for each node $i$ of~$T$, and to ensure that the node $h(i)$ of $G$ knows the subgraph induced by $B_i$ in $H=G[U]$. 
\begin{itemize}
    \item For every node $i$ of $T$ belonging to $TS$, we set $B_i = \emptyset$.

    \item For every node $i$ belonging to $F^U$, we act as in ~\cite{FominFMRT24}, i.e., we set $h(i) =i$, and $B_i$ is the set of ancestors of $i$ in $F^U$. As observed in~\cite{FominFMRT24}, each node $i$ can learn the subgraph induced by bag $B_i$ in a constant number of rounds. Moreover, for each tree $T^U$ of the forest $F^U$, the restriction of the decomposition $(T=(I,F),\{B_i\}_{i \in F})$ to $T^U$ is a tree decomposition of the subgraph of $G$ induced by the vertices of $T^U$, of width at most $2^p$. 
\end{itemize}

It is a matter of exercise to check that, by attaching to $TS$ the  tree-decompositions corresponding to the trees $T^U$ in the forest $F^U$, the resulting tree is a tree decomposition of $G[U]$ with width at most $2^p$.
Therefore we can 
apply the same dynamic programming scheme as above for computing values $\OPT_{\varphi}(G[U],\ell)$ and $\cnt_{\varphi}(G[U],\ell)$ in $D+2^p$ rounds. Together with the construction of the tree decomposition, the overall algorithm takes $\cO(D)$ rounds.
\end{proof}

By plugging the result of \Cref{lem:FFMRText} into \Cref{lem:optcol} and \Cref{lem:countie}
over all sets $U$ of at most $p$ colors for a suitable $p$ depending on $\varphi$ and
a $(p,f(p))$-treedepth coloring of $G$, we achieve the proof of \Cref{lem:QF-counting} and
\Cref{thm:thmFOexpOPT}. Note that in each case, we need to build a \((p,f(p))\)-treedepth coloring, which can be done
in \(\cO(\log n)\) \CONGEST rounds by \Cref{th:congestltd}.
\qed

\subsection{Proof of Theorem~\ref{thm:thmFOexpCNT}}
\label{ss:counting-FO}

\Cref{lem:QF-counting} shows how to solve the counting problem for \emph{quantifier-free}
FO formulas on bounded expansion classes. In order to prove the more general 
\Cref{thm:thmFOexpCNT}, we apply the quantifier elimination techniques developed in
\Cref{sec:local-model-checking,sec:general-model-checking}. \Cref{lem:reductiongeneral}
transforms a general FO formula into a \(p\)-\(\lca\)-reduced formula. However,
\Cref{lem:QF-counting} does not handle such formulas, as it only deals with usual adjacency and
equality predicates. In order to solve that issue, it is tempting to use
\Cref{lem:qelim:remove-lca-general}, which replaces \(\lca\) predicates by FO formulas.
However, this would result in an existential formula rather than a quantifier free
formula as needed for applying \Cref{lem:QF-counting}. The main idea of the proof of \Cref{thm:thmFOexpCNT}
is to maintain that the quantified variables introduced by \Cref{lem:qelim:remove-lca-general}
are quantified \emph{uniquely}, in the following sense.

\begin{definition}\label{def:lca-freeable}
Let \(p \in \N\), let \(\Lambda\) a  set of labels, and let \(\varphi(\overline
x) \in \FO[\Lambda]\) be a \(p\)-\(\lca\)-reduced formula. \(\varphi\)~is
\emph{\(\lca\)-freeable on \(\mathcal G\)} if there exists \(\widehat\Lambda\), and a quantifier free formula 
\(\psi(\overline x, \overline y)\in\FO[\widehat\Lambda]\)  such that,
for every \(G = (V,E) \in \mathcal G\), for
every \(\Lambda\)-labeled \(p\)-skeleton \( (S,\ell)\) of \(G\), there exists a
\(\widehat\Lambda\)-labeling \(\widehat \ell\) of \(G\) satisfying 
\begin{enumerate}
    \item $\true(\varphi, S, \ell) = \true\Big(\big(\exists \overline y\,  \psi(\overline x, \overline y)\big), G, \widehat \ell\Big)$, and
    \item for every \(\overline v \subseteq V\) satisfying \((S,\ell) \models \varphi(\overline v)\), there exists a unique \(\overline w \subseteq V\) such that \( (G, \widehat \ell) \models \psi(\overline v, \overline w)\).
\end{enumerate}

\end{definition}

\noindent The following lemma outlines the usefulness of \(\lca\)-freeability for counting.

\begin{lemma}\label{prop:lca-freeable-counting}
Let \(p\in \N\), and let \(\varphi(\overline x) \in \FO[\Lambda]\) be a \(p\)-\(\lca\)-reduced, and \(\lca\)-freeable formula. For every  \(\widehat\Lambda, \psi(\overline x, \overline y), G, S, \ell, \widehat \ell\) as in \Cref{def:lca-freeable}, 
\[
\cnt_\varphi(S,\ell) = \cnt_\psi(G,\widehat \ell).
\]
\end{lemma}

\begin{proof}
    \Cref{def:lca-freeable} specifies the existence of a one-to-one correspondence between \(\true(\varphi,S,\ell)\) and
    \(\true(\psi, G, \widehat \ell)\). These two set have thus the same cardinality.
\end{proof}

The following lemma is an analog of \Cref{lem:qelim:remove-lca-general}.

\begin{lemma}\label{lem:lca-freeable-lca}
Every formula \(\varphi(y,z) = \lca_i^U(y,z)\) is \(\lca\)-freeable on
\(\mathcal G\).
Moreover,
assuming \( (S,\ell)\) is given as input, every node \(u\) can compute \(\widehat\Lambda,
\psi, \widehat \ell(u)\) without any communication.
\end{lemma}

\begin{proof}
The proof follows the same guidelines as the one of \Cref{lem:qelim:remove-lca-general}.
Recall labels \(\col_k\) were added for each color \(k\), and the label \(\depth^U_j\) was added to mark the depth
of a node in each forest \(F^U\). Each \(\lca\) predicates was then replaced by an existential formula.
The following formula is very similar, except that we removed the quantifiers on \(y_0,\dots,y_{d-1}, z_0, \dots, z_{d-1}\), which are now free variables, and we added another constraint, specified on the last line of the formula below.

\begin{align*}
\psi(y,z,y_0, \dots, y_{d-1}, & z_0, \dots, z_{d-1})
 = \bigvee_{(a,b) \in [\max(0,i), d-1]^2} \Biggl( (y_a = y) \wedge (z_b = z)
\\
& \wedge \Big(\bigwedge_{s\in [0,i]} y_s = z_s\Big)
\wedge \Big(\bigwedge_{s\in [i+1, \min(a,b)]} y_s \neq z_s\Big)
\wedge \depth_0^U(y_0) \wedge \depth_0^U(z_0)
\\
& \wedge \Big(\bigwedge_{s\in [1,a]} \depth_s^U(y_s) \wedge \adj(y_{s-1}, y_s)\Big)
\wedge \Big(\bigwedge_{s\in [1,b]} \depth_s^U(z_s) \wedge \adj(z_{s-1},
z_s)\Big)
\\
& \wedge \Big(\bigwedge_{s \in [a+1, d-1]} y_s = y\Big)
\wedge \Big(\bigwedge_{s \in [b+1,d-1]} z_s = z \Big)
\Biggr).
\end{align*}
Recall that \(a\) (resp., \(b\))  represents here the depth of \(y\) (resp., \(z\)) in \(F^U\),
and \(y_0, \dots, y_a\) (resp., \(z_0, \dots, z_b\)) represent the path from the root of
\(y\)'s tree to \(y\) (resp., from the root of \(z\)'s tree to \(z\)). Since, in general,
variables \(y_{a+1},\dots,y_{d-1},z_{b+1}, \dots, z_{d-1}\) can take any value, the (new)
last line of the formula above is needed to enforce their uniqueness. For the same reasons as in \Cref{lem:qelim:remove-lca-general}, \(\psi\) satisfies
\[
\true(\varphi, S, \ell) =
\true(\exists \overline y, \exists \overline z, \psi(y,z,\overline y, \overline
z), G, \widehat \ell).
\]
Let us now prove that the quantified variables are unique.
Fix \((y,z) \in \true(\varphi, S, \ell)\) and let \(\overline y, \overline z\) such that
\((G,\widehat \ell) \models \psi(y,z,\overline y, \overline z)\). For some \(
(a,b)\in [i,d-1]^2\), the inside of the disjunction of \(\psi\) is true. This
implies that:
\begin{itemize}
\item \(y\) (resp. \(z\)) has depth \(a\) (resp. \(b\)),
\item for \(s\leq a\) (resp. \(s \leq b\)), \(y_s\) (resp. \(z_s\)) is the
(unique) ancestor of \(y\) (resp. \(z\)) at depth \(s\), and
\item for \(s > a\) (resp. \(s > b\)), \(y_s = y\) (resp. \(z_s = z\)), thus it
is also unique.
\end{itemize}
This completes the proof. 
\end{proof}

\begin{lemma}\label{lem:lca-freeable-conj}
A conjunction of \(\lca\)-freeable formulas is \(\lca\)-freeable.
\end{lemma}

\begin{proof}
Let \(\varphi_1(\overline x), \varphi_2(\overline x)\) be two \(\lca\)-freeable
formulas and \(\varphi(\overline x) = \varphi_1(\overline x) \wedge
\varphi_2(\overline x)\).
Let \(\widehat\Lambda_1, \widehat\Lambda_2\), \(\psi_1(\overline x,\overline y),
\psi_2(\overline x, \overline z)\) as in \Cref{def:lca-freeable}.
We set \(\widehat\Lambda = \widehat\Lambda_1 \cup \widehat\Lambda_2\), where we 
assume, w.l.o.g., that the sets of labels \(\widehat\Lambda_1, \widehat\Lambda_2\) are disjoint, and we set \(\psi(\overline x,\overline y,
\overline z) = \psi_1(\overline x, \overline y) \wedge \psi_2(\overline x,
\overline z)\). For \(G,S,\ell_1,\ell_2,\widehat \ell_1, \widehat
\ell_2\) as in \Cref{def:lca-freeable}, we define \(\widehat \ell(u) = \widehat
\ell_1(u)\cup\widehat \ell_2(u)\).
Let \(\overline v \in \true(\varphi, S, \ell)\).
Since \(\psi_1\) (resp. \(\psi_2\)) is \(\lca\)-freeable, there exists a unique
\(\overline w\) (resp. \(\overline u\))
such that \( (\overline v, \overline w) \in \true(\psi_1, G,
\widehat \ell_1)\) (resp. \( (\overline v, \overline u) \in \true(\psi_2, G,
\widehat \ell_2)\)). Then, \((G, \widehat \ell)\models \psi(\overline v, \overline w, \overline u)\).
Conversely, let \(\overline{w'}, \overline{u'}\) such that \( (G, \widehat \ell) \models
\psi(\overline v, \overline{w'}, \overline{u'})\). Then, \((G, \widehat \ell) \models \psi_1(\overline v, \overline{w'}\) (resp. \(\psi_2(\overline v, \overline{u'})\)), which implies
\(\overline{w'} = \overline w\) (resp. \(\overline{u'} = \overline u\)).
\end{proof}

At this stage, we can now prove the global counting part of \Cref{thm:thmFOexpCNT}.

\begin{proof}[Proof of \Cref{thm:thmFOexpCNT} (global counting)]
Let \(\Lambda\) be a (finite) set of labels and \(\varphi(\overline x) \in \FO[\Lambda]\).
By \Cref{lem:reductiongeneral}, there exists \(p\in \N\), \(\Lambda_1\) and
\(\varphi_1(\overline x) \in \FO[\Lambda_1]\) \(p\)-\(\lca\)-reduced such that, for every
\((G,\ell) \in \mathcal G[\Lambda]\), and for every \(S\) \(p\)-skeleton of \(G\), there exists a
\(\Lambda_1\)-labeling \(\ell_1\) of \(S\) such that \(\true(\varphi, G, \ell) = \true(\varphi_1, S, \ell_1)\) (and thus \(\cnt_\varphi(G,\ell) = \cnt_{\varphi_1}(S,\ell_1)\).
We compute a \(p\)-skeleton using \Cref{prop:skeleton-construction}, and the corresponding labeling \(\ell_1\),
in \(\cO(D+\log n)\) rounds. Let us write \(\varphi_1\) in disjunctive normal form:
\[
\varphi_1(\overline x) = \bigvee_{j=1}^k \psi_j(\overline x)
\]
where \(\psi_j\) is a conjunction of \(\lca^U_i\) predicates, label predicates, color predicates,
or their negations.
As in \Cref{lem:qelim:remove-lca}, note that \(\neg\lca^U_i(x,y)\) is equivalent to:
\[\neg\col_U(x) \vee \neg\col_U(y) \vee \bigvee_{j\in [-1,d] \setminus \{i\}} \lca^U_j(x,y).\]
We may thus assume that \(\varphi_1\) does not contain any negated \(\lca\) predicate.
We apply the inclusion/exclusion principle to \(\varphi_1\), that is, 
\[
\cnt_{\varphi_1}(S,\ell_1) = \sum_{h=1}^k (-1)^{h-1} \sum_{1\leq j_1 < \dots < j_h \leq k}
\cnt_{\bigwedge_{m=1}^h \psi_{j_m}}(S,\ell_1). 
\]
It only remains to compute the values of \[\cnt_{\bigwedge_{m=1}^h \psi_{j_m}}(S,\ell_1),\]
for some fixed indexes \(j_1,\dots,j_h\).
Since each \(\psi_{j_m}\) are conjunctions of \(\lca^U\)-predicates and (possibly negated) label and color predicates, it is also true of
\(\bigwedge_{m=1}^h \psi_{j_m}\). By \Cref{lem:lca-freeable-lca,lem:lca-freeable-conj}, this
formula is \lca-freeable. By \Cref{prop:lca-freeable-counting} there exists \(\widehat\Lambda\),
\(\widehat\psi \in \FO[\Lambda]\) quantifier-free, and \(\widehat \ell\) such that
\[\cnt_{\bigwedge_{m=1}^h \psi_{j_m}}(S,\ell_1) = \cnt_{\widehat\psi}(G, \widehat \ell).\]
Moreover, \(\widehat \Lambda, \widehat \psi,\widehat \ell(u)\) can be computed by any node $u$ without communication
(cf. \Cref{lem:lca-freeable-lca,lem:lca-freeable-conj}, and their proofs).
Finally, since \(\widehat\psi\) is quantifier free, it only remains to apply \Cref{lem:QF-counting} to complete the proof.
\end{proof}

In order to prove our result on \emph{local} counting, we need to adapt the notion of \emph{\lca-freeable
formulas} so that it preserves locality.

\begin{definition}
Let \(p \in \N\), let \(\Lambda\) be a (finite) set of labels,  and let  \(\varphi(\overline
x) \in \FO[\Lambda]\) be a \emph{local} \(p\)-\(\lca\)-reduced formula.
\(\varphi\) is
\emph{locally \(\lca\)-freeable on \(\mathcal G\)} if the following two conditions hold:
\begin{itemize}
    \item \(\varphi\) is \lca-freeable, and 
    \item the quantifier-free formula \(\psi\) from \Cref{def:lca-freeable} can be chosen
    to be local.
\end{itemize}
\end{definition}

\noindent The following lemma is the ``locality preserving version'' of \Cref{lem:lca-freeable-lca,lem:lca-freeable-conj}.

\begin{lemma}\label{lem:locally-lca-freable}
    Let \(\varphi(\overline x) \in \FO[\Lambda]\) be a \(p\)-\lca-reduced formula that  is supposed to be a
    conjunction of \(\lca^U\) predicates, and (possibly negated) label and color predicates.
    If \(\varphi\) is local then it is locally \lca-freeable.
    Moreover, assuming a \(p\)-skeleton \((S,\ell)\) is given as input, each node \(u\) can
    compute \(\widehat \Lambda, \psi, \widehat\ell(u)\) matching the requirement of \Cref{def:lca-freeable}
    without any communication.
\end{lemma}

Before proving \Cref{lem:locally-lca-freable}, we first prove a technical lemma.

\begin{lemma}\label{lem:lca-connected}
    Let \(\varphi(\overline x)\) be a local \(p\)-\lca-reduced formula that  is assumed to be a
    conjunction of \(\lca^U\) predicates, and (possibly negated) label and color predicates.
    Let \(k=|x|\), and let \(H(\varphi)\) be the graph with vertex set equal to the set of variables \(x_1, \dots, x_k\),
    and where there is an edge between \(x_i\) and \(x_j\) if and only if the formula
    \(\varphi\) contains a predicate of the form \(\lca^U_m(x_i,x_j)\) with $m \geq 0$.
    If there exists a labeled skeleton \((S,\ell)\), and a set of vertices \(\overline v\) such that
    \((S,\ell) \models \varphi(\overline v)\), then, \(H(\varphi)\) is connected.

\end{lemma}

\begin{proof}
    Let us fix a formula \(\varphi\), a labeled skeleton \((S, \ell)\), and a set of vertices \(\overline v\)
    as in the statement of the lemma.
    We proceed by contradiction. Let us assume that \(H(\varphi)\) is
    not connected. Up to a reordering of the variables \(x_2, \dots, x_k\), we may assume that
    \(x_1, \dots, x_s\) form a single connected component, $1\leq s<k$, and that the vertices \(x_{s+1}, \dots, x_k\)
    are in other connected components. Let us consider the labeled skeleton \((S',\ell')\) consisting of two connected components \(S^1\) and \(S^2\), each of which being a copy of \((S,\ell)\).
    For any vertex \(u \in S\), let us denote by \(u^i\) the copy of \(u\) in \(S^i\). Let \(\overline{w}\) where, for every $i=1,\dots,k$, 
    \[
        w_i = \begin{cases}
            v^1_i & \textrm{if } i\leq s \\
            v^2_i & \textrm{if } i>s
        \end{cases}   
    \]
    We have that \((S',\ell') \models \varphi(\overline w)\).
    However \(w_{s+1}, \dots, w_k\) are not in the connected component of \(w_1\). In particular,
    they are not at bounded distance from \(w_1\). This contradicts the fact that
    \(\varphi\) is local (see \Cref{def:r-local-general}).
\end{proof}

When $H(\varphi)$ is connected, we say that \(\varphi\) is \emph{\(\lca\)-connected}.

\begin{proof}[Proof of \Cref{lem:locally-lca-freable}]
Since \(\varphi\) is local, we get from \Cref{lem:lca-connected} that it is either \lca-connected, or always
false. Let us assume that it is not \lca-connected. Then one can trivially output the formula \(\bot\)
that is always false.
Thus, we now assume that \(\varphi\) is \lca-connected.
We then merely adapt the proofs of \Cref{lem:lca-freeable-lca,lem:lca-freeable-conj}.
Let us consider a labeled \(p\)-skeleton \((S, \ell)\). As in the proof of \Cref{lem:lca-freeable-lca},
we add new labels \(\depth^U_i\), and \(\col_k\), and we denote by \(\widehat\ell\) the resulting labeling.
We then replace each predicate \(\lca^U_i(y,z)\) of \(\varphi\) by
\begin{align*}
\bigvee_{(a,b) \in [\max(0,i), d-1]^2} & \Biggl( (y_a = y) \wedge (z_b = z)
 \wedge \Big(\bigwedge_{s\in [0,i]} y_s = z_s\Big)
\wedge \Big(\bigwedge_{s\in [i+1, \min(a,b)]} y_s \neq z_s\Big)\\
& \wedge \depth_0^U(y_0) \wedge \depth_0^U(z_0)
\\
& \wedge \Big(\bigwedge_{s\in [1,a]} \depth_s^U(y_s) \wedge \adj(y_{s-1}, y_s)\Big)
\wedge \Big(\bigwedge_{s\in [1,b]} \depth_s^U(z_s) \wedge \adj(z_{s-1},
z_s)\Big)
\\
& \wedge \Big(\bigwedge_{s \in [a+1, d-1]} y_s = y\Big)
\wedge \Big(\bigwedge_{s \in [b+1,d-1]} z_s = z \Big)
\Biggr),
\end{align*}
where \(y_1, \dots, y_d, z_1, \dots, z_d\) are new free variables (which are \emph{unique},
see the proof of \Cref{lem:lca-freeable-lca}).
Note that if the formula above is satisfied, then  each newly introduced variable \(y_i\)
(resp. \(z_i\)) is at distance at most \(d\) from \(y\) (resp.~\(z\)). Moreover, if \(i\ge 0\), then \(y\) and \(z\) are at distance at most \(2d\), as they are both
at distance at most \(d\) from \(y_0=z_0\).
Let us denote by \(\psi(\overline x, \overline t)\) the formula obtained by replacing each
\(\lca^U\)-predicate by the formula above (\(\overline t\) is the collection of newly introduced variables).
By the same argument as in \Cref{lem:lca-freeable-lca}, we have
\[
\true(\varphi, S, \ell) =
\true(\exists \overline t\; \psi(\overline x, \overline t), G, \widehat\ell), 
\]
and
\[
\cnt_\varphi(S, \ell) = \cnt_\psi(G, \widehat\ell).
\]
It only remains to show that \(\psi\) is local.
Since it is quantifier-free, this is equivalent to showing 
that, for some \(r'\), every labeled graph \((G',\ell')\) satisfies
\[
    (\overline v,\overline w) \in \true(\psi, G', \ell') \implies
    \overline v, \overline w \subseteq B_{G'}(v_1, r').
\]
Since, are underlined above, each newly introduced variable in \(t_i\in\bar t\)
is at distance at most \(d\) from at least one of the variables \(x_j\), it is sufficient to show that for some $r'$
\[
    (\overline v,\overline w) \in \true(\psi, G', \ell') \implies
    \overline v \subseteq B_{G'}(v_1, r'), 
\]
as $\overline v \subseteq B_{G'}(v_1, r')$ implies that \(\overline w \subseteq B_{G'}(v_1,r'+d)\).
Recall that \(\varphi(\overline x)\) is \lca-connected, and so let \(H(\varphi)\) be defined as in
\Cref{lem:lca-connected}. For each index \(1 \leq j \leq |x|\), since \(H(\varphi)\) is connected
there is a path \(x_{i_1}, \dots, x_{i_k}\) from \(x_1\) to \(x_j\) in \(H(\varphi)\).
Since \(H(\varphi)\) has \(|x|\) nodes, we may assume this path has length at most \(|x|-1\).
The formula \(\psi\) guarantees that, for every \(l<k\), \(x_{i_l}\) and
\(x_{i_{l+1}}\) are at distance at most \(2d\) in \(G'\). Finally, \(x_1\) and \(x_j\) are at
distance at most \(2d(|x|-1)\). This concludes the proof that \(\varphi\) is local.
\end{proof}

We are now ready to establish \Cref{thm:thmFOexpCNT}.

\begin{proof}[Proof of \Cref{thm:thmFOexpCNT} (local counting)]
    We follow the same guidelines as in the proof of the global counting result, but we make sure that
    locality is preserved throughout the process. Let us consider a labeled graph \((G,\ell)\).
    Let us apply the \emph{locality preserving} elimination technique (cf. \Cref{lem:qelim:induction})
    for reducing \(\varphi\) to a local \(p\)-\lca-reduced formula \(\varphi_1\).
    We then construct a \(p\)-skeleton~\(S\) (cf.  \Cref{prop:skeleton-construction}), and a new labeling
    \(\ell_1\), in \(\cO(\log n)\) rounds.
    Let us rewrite \(\varphi_1\)
    in disjunctive normal form:
    \[
    \varphi_1(\overline x) = \bigvee_{j=1}^k \psi_j(\overline x)
    \]
    where \(\psi_j\) is a conjunction of \(\lca^U\)-predicates, and (possibly negated) label
    and color predicates --- recall that we can assume that there are no negated \(\lca^U\)-predicates. The formula 
    \(\varphi_1\) is quantifier-free so, thanks to \Cref{lem:local-disjunction}, each term of the
    disjunction is local.
    Let us apply the inclusion/exclusion principle to \(\varphi_1\), i.e., 
    \[
    \cnt_{\varphi_1}(S,\ell_1) = \sum_{h=1}^k (-1)^{h-1} \sum_{1\leq j_1 < \dots < j_h \leq k}
    \cnt_{\bigwedge_{m=1}^h \psi_{j_m}}(S,\ell_1). 
    \]
    It thus remains solely to perform the local counting of 
    \[\cnt_{\bigwedge_{m=1}^h \psi_{j_m}}(S,\ell_1)\]
    for fixed indexes \(j_1, \dots, j_h\). Since each formula \(\psi_{j_m}\) is a \emph{local}
    conjunction of \(\lca\)-predicates (as well as possibly negated label and color predicates), then same holds for\(\bigwedge_{m=1}^h \psi_{j_m}\). By \Cref{lem:locally-lca-freable}, this latter formula 
    formula is locally \lca-freeable, and one can compute a local, quantifier free
    formula \(\widehat\psi\), and a labeling \(\widehat\ell\) of \(G\), without any communication,  such that
    \[\cnt_{\bigwedge_{m=1}^h \psi_{j_m}}(S,\ell_1) = \cnt_{\widehat\psi}(G, \widehat \ell).\]
    Since \(\widehat\psi\) is local and quantifier-free, it suffices to apply
    (the local case of) \Cref{lem:QF-counting} for completing the proof.
\end{proof}

\section{Distributed Certification}
\label{sec:proofofthmFOexpCER}

In this section we prove that every graph property expressible in FO can be certified with certificates on $\cO(\log n)$ bits in graphs of bounded expansion. 

\begin{theorem}
\label{thm:thmFOexpCER} 
Let $\Lambda$ be a finite set. For every FO formula $\varphi$ on $\Lambda$-labeled graphs, and for every class  of graphs $\mathcal{G}$ of bounded expansion, there exists a distributed certification scheme that, for every $n$-node network $G=(V,E) \in \mathcal{G}$ and every $\ell: V\rightarrow 2^{\Lambda}$, certifies $(G,\ell) \models \varphi$ using certificates on $\cO(\log n)$ bits.
\end{theorem}

For establishing \Cref{thm:thmFOexpCER}, we revisit the logical structures introduced and analyzed in the previous section, by mainly focusing on the points where theses structures must be certified instead of being computed. 

\subsection{Certifying Reduction of Trees with Bounded Depth}

We begin with a basic primitive, that we use to count values, to identify connected components, and for several other applications. 

\begin{proposition}[\cite{KormanKP10}]\label{lem:STcertification}
 Let us suppose that each node $u$ in a graph $G = (V,E)$ is given as an input a triple $(\root(u), \parent(u), \depth(u))$, where $\root(u)$ and $\parent(u)$ are node identifiers, and $\depth(u)$ is a positive integer. Let us consider the algorithm where each node $u$ checks:
 \begin{itemize}
     \item For each $v\in N(u)$, $\root(v)=\root(u) = \rho$.
     \item If $\id(u) = \rho$ then $\parent(u) = \id(u)$ and $\depth(u)=0$.
     \item If $\id(u)\neq \rho$, then there exists $v\in N(u)$ s.t. $\id(v) = \parent(u)$ and $\depth(v) = \depth(u)-1$. 
 \end{itemize}
 If all checks are passed at every node, then there exists a spanning tree $T$ of $G$ rooted at $\rho$ such that the function $\parent$ defines the parent relation in~$T$, and $\depth(u)$ is the depth of $u$ in~$T$. 
\end{proposition}

In the above proposition, if all checks are passed at every node, then the set $$\{(\root(u), \parent(u), \depth(u)) \mid u\in V\}$$ is called the \emph{encoding}  of the tree~$T$. 
Given a graph $G$ and a  forest $F\subseteq G$ with components $T_1, \dots, T_k$, we say that the nodes of $G$ are given the \emph{encoding}  of $F$ if, for every $i\in[k]$, each node $u\in T_i$ is given the encoding  of~$T_i$. 


\begin{lemma}\label{lem:bcongestToPLS}
     Let $\mathcal{A}$ be an algorithm such that, in every $n$-node network $G=(V,E)$,  every node $u\in V$ is given an $\cO(\log n)$-bit input $\textsf{In}(u)$,  and must produce an $\cO(\log n)$-bit output $\textsf{Out}(u)$. Let $R(n)$ be the round-complexity of $\mathcal{A}$  in the \BCONGEST\ model. Then, there is a certification scheme using certificates on $\cO(R(n)\log n)$ bits that certifies that a pair of functions $f:V\to \{0,1\}^*$ and $g:V\to \{0,1\}^*$ form an input-output pair for $\mathcal{A}$ (i.e., if, for every node $u\in V$, $\textsf{In}(u)=f(u)$, then, for every node $u\in V$, $\textsf{Out}(u)=g(u)$).
\end{lemma}

\begin{proof}
    On legal instances $(f,g)$ for network $G=(V,E)$, the prover provides each vertex $u\in V$ with the sequence $M(u) = (m_1(u), \dots, m_{R(n)}(u))$ where $m_i(u)$ is the message that $u$ broadcasts to its neighbors on input $f(u)$ during the $i$-th round of~$\mathcal{A}$. 
    
    The verification algorithm proceeds as follows. Every node collects the certificates of its neighbors, and simulates the $R(n)$ rounds of~$\mathcal{A}$. That is, $u$ checks first whether $m_1(u)$ is indeed the message that it would send in the first round of~$\mathcal{A}$. For every $i\in \{2,\dots,R(n)\}$, node $u$ assumes that the messages $m_1(v), \dots, m_{i-1}(v)$ of each of its neighbors $v\in N(u)$ were correct, and checks whether $m_i(u)$ is indeed the message that it would broadcast at round~$i$. If some of these tests are not passed then $u$ rejects. If all tests are passed, then $u$ finally checks whether $g(u)$ is the output that $u$ would produce by executing~$\mathcal{A}$, and accepts or rejects accordingly. 
    
    The completeness and soundness of the algorithm follow by construction. 
\end{proof}

Let us now prove a result analogous to \Cref{lem:generalforestcomputing}, but for distributed certification. 

\begin{lemma}\label{lem:forestReductionverification}
Let $(G,\ell)$ be a connected $\Lambda$-labeled graph, and let $F\in \mathcal{F}_d$ be a subgraph of $G$. Let $\varphi(\overline{x})= \exists \overline{y}\; \zeta(\overline{x},\overline{y})$ be a formula in $\FO[\Lambda]$, where $\zeta$ is an $\lca$-reduced formula. Let us assume that every node $u$ of $G$ is provided with the following values:  
\begin{itemize}
\item Two sets $\Lambda$ and $\widehat{\Lambda}$;
\item The formula $\varphi$, and a $\lca$-reduced formula $\widehat{\varphi}$;
\item A set of labels $\ell(u)\subseteq \Lambda$, and a set of labels $\widehat{\ell}(u)\subseteq \hat{\Lambda}$;
\item The encoding $(\root(u), \parent(u), \depth(u))$ of $F$;
\item The encoding $(\root^T(u), \parent^T(u), \depth^T(u))$ of a spanning tree $T$ of $G$.
\end{itemize}
There is a certification scheme  using $\cO(\log n)$-bit certificates that certifies  that $\true(\varphi,F,\ell) = \true(\widehat{\varphi},F,\tilde{\ell})$.
\end{lemma}

\begin{proof}
Let $k = |\overline{x}|$. As  in the proof of \Cref{lem:elimination_forest}, let us assume w.l.o.g. that $k\geq 1$ (as otherwise we can create a dummy free variable $x$ for $\varphi$). 
Moreover, let us assume that $\widehat{\Lambda}$ and $\widehat{\varphi}$ are the set and the formula defined in the proof of \Cref{lem:elimination_forest}. Since $\widehat{\Lambda}$ and $\widehat{\varphi}$ can be computed by every node without any communication, each node can reject in case they are not appropriately defined. Similarly, we can assume that $\ell$ and $\hat{\ell}$ are, respectively, a $\Lambda$-labeling and a $\hat{\Lambda}$-labeling. Thanks to \Cref{lem:STcertification}, $\cO(\log n)$-bits certificates are sufficient to certify that the encodings of $F$ and $T$ are correctly. The remaining of the proof consists in checking that $\widehat{\ell}$ is correctly set as specified the proof of \Cref{lem:elimination_forest}.

Let us first suppose that $|\overline{y}|=1$. We shall deal with the  case when $|y|> 1$ later in the proof. We start by proving the result for the case where $\zeta$ is a $\lca$-type $\type_{\gamma, \delta}\in \Type(k+1,d,\Lambda)$, where the variable $y$ is identified with variable $k+1$ in the definition of $\gamma$ and $\delta$. In a second stage, we'll define $s$, $h$, $h_s$ and $h_y$ in the same way as we did in \Cref{lem:elimination_forest}, and revisit the cases of \Cref{lem:elimination_forest}. 

In Case~1 and~2, observe that the algorithms provided in \Cref{lem:local_forest_computing} are algorithms that fit with the broadcast constraint of the \BCONGEST\ model, running in $\cO(d)$ rounds. On input $\ell$, these algorithms output $\widehat{\ell}$. By \Cref{lem:bcongestToPLS}, we can certify whether $\widehat{\ell}$ is well defined in $d\cdot \cO(\log n) = \cO(\log n)$ rounds.  

Now let us deal with Case~3. The algorithm $\mathcal{A}$ given in \Cref{lem:generalforestcomputing} for this case has three phases. The  first two phases can be implemented by an algorithm $\mathcal{B}$ under the \BCONGEST\ model, running in $\cO(d)$ rounds. The output of $\mathcal{B}$ consists in a marking of all active roots, as well as a marking of all good nodes. We can simulate these two first phases by giving a certificate to each node, interpreted as the markings given by $\mathcal{B}$, as well as a certificate of the execution of $\mathcal{B}$. According to \Cref{lem:bcongestToPLS}, this can be implemented using certificates on $\cO(\log n)$ bits. The last phase of Algorithm~$\mathcal{A}$ aims at providing each node with  the number of active roots. Each node $u$ receives a certificate consisting of:
\begin{itemize}
    \item an integer $\rho(u)$ equal the number of active roots in $F$;
    \item  a bit $a(u)$ indicating whether or not $u$ is an active root;
    \item an integer $\textsf{val}(u)$ equal to the sum of all bits $a(v)$ in the subtree of $T$ rooted at $u$. 
\end{itemize}
Observe that these certificates can be encoded in $\cO(\log n)$ bits. For a non-negative integer~$\nu$, let $$h(\nu) =\begin{cases} \nu &\textsf{ if } \nu \leq k,\\
\infty &\textsf{ otherwise.}
\end{cases}$$
The verifier checks the following conditions at every node~$u$:
\begin{itemize}
    \item[(1)] For every $v\in N(u)$, $\rho(u) = \rho(v)$;
    \item[(2)] $a(u) = 1$ if and only if $u$ is marked as active by $\mathcal{B}$;
    \item[(3)] The following condition holds 
    $$\textsf{val}(u) = a(u) + \sum_{v \in \textsf{Children}(u)} \textsf{val}(v),$$ where $\textsf{Children}(u)$ denotes the set of nodes $v$ such that $\parent^T(v) = \id(u)$;
    \item[(4)] If $\id(u) = \root^T(u)$, then $\textsf{val}(u) = \rho(u)$;
    \item[(5)] If $u\in F$ is marked as good by $\mathcal{B}$, then $\widehat{\ell}(u) = \ell(u) \cup \{\good, h(\rho(u))\}$;
    \item[(6)] If $u\in F$ is not marked as good by $\mathcal{B}$, then $\widehat{\ell}(u) = \ell(u) \cup \{h(\rho(u))\}$;
    \item[(7)] If $u\notin F$ then $\widehat{\ell}(u) = \ell(u)$.
\end{itemize}

We now check completeness and soundness of this certification scheme. Condition~(1) is satisfied if and only if every node received the same value $\rho$ for the number of active roots. Conditions~(2) and~(3) are both satisfied if and only if $\val(u)$ correspond to the number of active roots in the subgraph of $T$ rooted at~$u$. In particular $\val(r)$ is equal to the number of active roots. Then, conditions (1) to (4) are satisfied if and only if $\rho$ corresponds to the number of active roots in $F$. Finally, conditions (1) to (7) are satisfied if and only if $\widehat{\ell}$ is well defined according to the proof of \Cref{lem:elimination_forest}.

Let us denote by $\mathcal{A}^{\type}$ the verification algorithm, and by $\Cert^\zeta$ the certificates  described for each of the three cases when $\zeta$ is an $\lca$-type formula.  Let us now assume that $\zeta$ is an arbitrary $\lca$-reduced formula. From \Cref{prop:basicnormal}, there exists a set $I\subseteq~\Type(k+1, D, \Lambda)$ such that,
$$
\zeta(\overline{x},y) = \bigvee_{\psi \in I} \psi(\overline{x},y).
$$
For each $\psi \in I$, let us denote by $\widehat{\Lambda}^\psi$, $\widehat{\psi}$ and $\widehat{\ell}^\psi$ the reduction on $\mathcal{F}_d$ of $\Lambda$, $\psi$ and $\ell$, respectively. The prover provides each node $u$ of $G$ with the certificate $c^\psi(u)=\widehat{\ell}^\psi(u)$ for each $\psi$, in addition to the certificate $\Cert^\psi$.  

The verifier proceeds as follows. Each node $u$ checks  that it accepts when running algorithm $\mathcal{A}^{\type}$ with certificates $c^\psi(u)$, for every $\psi \in I$. Then, each node $u$  relabels the sets $\widehat{\ell}^\psi$ to make them disjoint, as we did in the proof of \Cref{lem:elimination_forest}), and checks whether $\widehat{\ell}(u) = \bigcup_{\psi \in I} \widehat{\ell}^\psi(u)$. The soundness and completeness of the whole scheme follows from the soundness and completeness of $\mathcal{A}^{\type}$. 

It remains to deal with the case where $q= |\overline{y}| > 1$. Let us denote by $\mathcal{A}$ the verification algorithm that we described for the case when $|\overline{y}| = 1$.  The prover provides the nodes with a set of triplets $(\zeta_i, \ell_i, \Cert_i)_{i\in [q]}$ such that, for each $i \in [q]$,  $\text{Cert}_i$ is a certificate used to check whether
$$
\true(\exists y_i, \zeta_i(\overline{x}, y_1, \dots, y_i), G, \ell_i) = \true(\zeta_{i-1}(\overline{x}, y_1,\dots, y_{i-1}), G, \ell_{i-1}),
$$
Where $\zeta_q = \zeta$, $\ell_q = \ell$,  $\zeta_0 = \widehat{\varphi}$,  and $\ell_0 = \widehat{\ell}$. Every node $u$  accepts if $\mathcal{A}$ accept all certificates at~$u$. From the soundness and completeness of $\mathcal{A}$, we obtain that every node accepts if and only if
\begin{align*}
\true(\varphi(\overline{x}), G, \ell) &= \true(\exists (y_1, \dots, y_{q})\,  \zeta(\overline{x},(y_1, \dots, y_{q})), G, \ell)\\
&= \true(\exists (y_1, \dots, y_{q})\,  \zeta_q(\overline{x},(y_1, \dots, y_{q})), G, \ell_q)\\
&= \true(\exists (y_1, \dots, y_{q-1})\, \zeta_{q-1}(\overline{x},(y_1, \dots, y_{q-1})), G, \ell_{q-1})\\
&\dots\\
&= \true(\exists y_1\, \zeta_{1}(\overline{x},y_1), G, \ell_{1}) \\
&= \true( \zeta_{0}(\overline{x}), G, \ell_{0}) \\
&= \true(\widehat{\varphi}(\overline{x}), G, \widehat{\ell}).
\end{align*}
Which completes the proof. 
\end{proof}

\subsection{Certifying Low-Treedepth Colorings and Skeletons}

Bousquet et al. \cite{FeuilloleyBP22} have shown that, for every MSO property $\mathcal{P}$, there is a certification scheme  using $\cO(\log n)$-bit certificates for certifying $\mathcal{P}$ on graphs of bounded treedepth. Since the mere property ``treedepth at most $k$'' is expressible in MSO, we obtain the following. 

\begin{proposition}[\cite{FeuilloleyBP22}]\label{prop:PLSfortreedepth}
For every integer $p\geq 1$, there exists a certification scheme  using $\cO(\log n)$-bit certificates for certifying ``treedepth at most $p$''. 
\end{proposition}

We show how to use this result for certifying a $(p,f(p))$-treedepth coloring. Moreover, we also extend this certification scheme to an encoding of a $p$-skeleton. Given a $(p,f(p))$-treedepth coloring $C_p$, and $p$-skeleton $S$,  the \emph{encoding} of $S$ is defined as the encoding of all decomposition forests $F^U$ of $G[U]$, for all $U\in \Col(p)$. More precisely, in the encoding of $S$, for every  $U\in \Col(p)$, every node $u$ has  a triple $\{(\root^U(u), \parent^U(u), \depth^U(u))\}$ corresponding to the subtree $T_u^U$ containing $u$ of  the decomposition forest $F^U$ of $G[U]$.

\begin{lemma}\label{lem:verificationofpskeleton}
Let $f:\mathbb{N}\to\mathbb{N}$, and 
let $\mathcal{G}$ be a class of graphs of expansion $f$. Let $p>0$, and let us suppose that each node $u$ is given (1)~a value $C_p(u)\in [f(p)]$, and (2)~an encoding of $S$.
There is a certification scheme using $\cO(\log n)$-bit certificates for certifying that $C_p$ is a $(p,f(p))$-treedepth coloring of $G$, and that the encoding of  $S$ is indeed defining a $p$-skeleton of $G$. 
\end{lemma}

\begin{proof}
Let us first certify that $C_p$ is a $(p,f(p))$-treedepth coloring of~$G$. For every set $U\in \Col(p)$, every node $u$ receives as a certificate: (1)~The encoding of a rooted spanning tree $T$ of the component of $G[U]$ containing~$u$ (where each component is identified with the node identifier of the root of~$T$); (2)~The certification of the property ``treedepth at most $p$'' for all the nodes in the component of $G[U]$ containing~$u$. The soundness and completeness of (1) and (2) follow from  \Cref{lem:STcertification} and \Cref{prop:PLSfortreedepth}. 

Next, let us certify that $S$ is a $p$-skeleton of~$G$. First, we use the verifier of \Cref{lem:STcertification} to check that the encoding of each tree $T_u^U$ is correct, for all $u\in V$ and all $U\in \Col(p)$. If no nodes reject at that point, we can safely deduce that if $u$ and $v$ are in the same connected component of $G[U]$, then the encoding of $T_u^U$ and of $T_v^U$ must induce the same tree.

It remains to certify that each $T^U_u$ forms a decomposition tree of the component of $G[U]$ containing $u$. More precisely, we must certify that all neighbors of $u$ in $G[U]$ are either descendants or ancestors of $u$ in $T_u^U$. Let us fix an arbitrary set $U\in \Col(p)$. Let $(\root^U(u), \parent^U(u), \depth^U(u))$ be the encoding of $T_u^U$ that has $u$ in its input (given in the encoding of $S$). Then node $u$ receives the set $P^U_u = \{w^u_1, \dots, w^u_d\}$ as certificate, corresponding to a sequence of at most $d \leq 2^p$ node identifiers representing the path from $\root(u)$ to~$u$. Then node $u$ accepts $P^U_u$ if the following holds:
\begin{itemize}
    \item[(i)] $d=\depth^U(u)$;
    \item[(ii)] $w^u_1 = \root^U(u)$ and $w^u_d = \id(u)$;
    \item[(iii)] if $\id(u)\neq \root^U(u)$ and  $v\in N(u)$ is such that $\id(v)=\parent^U(u)$, then $P^U_v$ is a sequence of length $d-1$ satisfying $w_v^i = w_u^i$ for every $i\in[d-1]$;
    \item[(iv)] if every $v \in N(u)$ is such that $C_p(v) \in U$, then either $\id(u)\in P^U_v$ or $\id(v)\in P^U_u$. 
\end{itemize}
Conditions (i) to (iii) are satisfied if and only if $P_u^U$ represents the path from the root of $T_u^U$ to $u$, for all~$u$. Assuming that (i) to (iii) are satisfied, we have that (iv) is satisfied if and only if  $T_u^U$ is a decomposition tree of $G[U]$, for all nodes $u$ of $G[U]$. Indeed, for every $u$ in $G[U]$, and for every $v\in N(u)$ that is also in $G[U]$, we have that either $u$ is an ancestor of $v$ (i.e. $\id(u)$ belongs to $P^U_v$), or $v$ is an ancestor of $u$ (i.e. $\id(v)$ belongs to $P^U_u$). 

Then, for every $U\in \Col(p)$, one can certify that $T^U_u$ is an decomposition tree of  $G[U]$ with a certificate containing at most $2^p$ node identifiers, which can be encoded on $\cO(2^p \cdot \log n) = \cO(\log n)$ bits. Furthermore, one can certify that the union over all nodes $u$ and all $U\in \Col(p)$, of $T^U_u$ is a $p$-skeleton of $G$ using certificates of size $\cO(|\Col(p)| \cdot 2^p \cdot \log n) =\cO(\log n)$ bits. 
\end{proof}

\subsection{Certifying Reductions on Graphs of Bounded Expansion}

Let $f:\mathbb{N}\to\mathbb{N}$, and let $\mathcal{G}$ be a class of connected graphs with expansion $f$. We first define a notion of reducibility that admits an efficient distributed certification scheme. 

\begin{definition}\label{def:verifiablereduction}
Let $\varphi\in \FO[\Lambda]$. For every labeled graph $(G,\ell)$ with $G\in \mathcal{G}$, let us assume that every node $u$ of $G$ is provided with the following values:
\begin{enumerate}
    \item The set $\Lambda$, the formula $\varphi$, and the label $\ell(u)\subseteq \Lambda$;
    \item A positive integer $p$; 
    \item A set $\widehat{\Lambda}$, a formula $\widehat{\varphi}$, and a label $\widehat{\ell}(u)\subseteq \widehat{\Lambda}$;
    \item The encoding of a $p$-skeleton $S$ of $G$.
\end{enumerate}
We say that $\varphi$ admits a \emph{distributedly certifiable reduction on $\mathcal{G}$} if there is a  certification scheme  using $\cO(\log n)$-bits certificates for certifying that $(p, \widehat{\Lambda}, \widehat{\varphi}, \widehat{\ell})$ is the reduction of $(\Lambda, \varphi, S, \ell)$ on~$G$.
\end{definition}

In a way similar to what was done in \Cref{sec:general-model-checking}, we show that every first-order formula over labeled graphs admits a  certifiable reduction. We fist establish this result for existential formulas. 

\begin{lemma}\label{lem:reductionofexistentialverif}
Every existential formula in $\FO[\Lambda]$ admits a distributed certifiable  reduction on $\mathcal{G}$.
\end{lemma}

\begin{proof}
Let $\varphi \in \FO[\Lambda]$ be an existential formula. Consider that each node is provided with the values listed in \Cref{def:verifiablereduction}. Let $\varphi( \overline{x}) = \exists \overline{y}\, \psi(\overline{x},\overline{y})$, where $\psi$ is quantifier free. Every node checks that $p=|\overline{x}|+|\overline{y}|$. The  scheme of \Cref{lem:verificationofpskeleton} enables to certify that the encoding $S$ correctly defines a $p$-skeleton of~$G$. Let $C_p$ be the $(p,f(p))$-treedepth coloring that defines $S$. The same construction as in  \Cref{lem:reductionofexistential} can then be used. More precisely, for every $U\in \Col(p)$, let $\widetilde{\Lambda}^U$, $\widetilde{\psi}^U$, $\widetilde{\ell}^U$, $\Lambda^U$, $\psi^U$, and $\ell^U$ be defined as in the proof of \Cref{lem:reductionofexistential}. Instead of computing these values, each node  receives them from the prover as certificates. The values of $\widehat{\Lambda}$, $\widehat{\varphi}$, $\widetilde{\Lambda}^U$, $\widetilde{\psi}^U$,  $\Lambda^U$, and  $\psi^U$ do not depend on the input graph, and can be computed by each node. The labeling $\tilde{\ell}^U$ can be computed from $\ell$ by a 1-round algorithm in the \BCONGEST\ model. Hence, using \Cref{lem:bcongestToPLS}, the correctness of each labeling $\tilde{\ell}^U$ can be certified with $\cO(\log n)$-bit certificates. The correctness of each labeling $\widehat{\ell}^U$ can be certified using \Cref{lem:forestReductionverification}, using $\cO(\log n)$-bit certificates. Finally, each node can check the correctness of $\widehat{\ell}$, and accept or reject accordingly. 
\end{proof}

We can now proceed with the general case. 

\begin{lemma}\label{lem:reductiongeneralverif}
Every formula in $\FO[\Lambda]$ admits a distributed certifiable reduction on $\mathcal{G}$. 
\end{lemma}

\begin{proof}
Let $\varphi \in \FO[\Lambda]$ be an existential formula. Let us assume that each node is provided with the values described in \Cref{def:verifiablereduction}, i.e. the tuples $(\Lambda, \varphi, \ell)$ and $(p, \hat{\Lambda}, \hat{\varphi}, \hat{\ell})$, and the encoding of a $p$-skeleton~$S$.  We follow the same induction as in \Cref{lem:reductionofexistential}, over the quantifier depth of~$\varphi$. The base case is trivial as the construction requires no communication. Furthermore, if $\varphi = \neg \psi$ where $\psi$ admits a distributed certifiable reduction, then $\varphi$ admits a distributed certifiable reduction as well,  as we described in \Cref{lem:reductionofexistential}. 

It remains to consieder the case $\varphi(\overline{x}) = \exists y\, \psi(\overline{x}, y)$. Let  $(p^\psi, \widehat{\Lambda}^\psi, \widehat{\psi}, \widehat{\ell}^\psi(u))$, $(\Lambda^\xi, \xi, \ell^\xi)$ as it was done in Case~1 in the proof of \Cref{lem:reductiongeneral}. Again, instead of computing these values, each node receives them as certificates, and checks them. The values of $p, p^\psi, \widehat{\Lambda}^\psi$, $\widehat{\psi}$, $\Lambda^\xi$, and $\xi$ do not depend on the input graph, and can be computed  by each node.  Then, each node can construct the input values (1) to (4) given in \Cref{def:verifiablereduction} with respect to to $\xi$ and $\psi$. By the induction hypothesis, the correctness of $\widehat{\ell}^\psi(u)$ at each node $u$ can be certified with $\cO(\log n)$-bit certificates.

Recall that, in \Cref{lem:reductiongeneral}, the value of $\ell^\xi(u)$ on each node $u$ was computed using \Cref{lem:qelim:remove-lca-general}, without any communication, under the assumption that each node knows a encoding a $p^\psi$-skeleton $S$ of $G$. The prover provides the encoding of a $p^\psi$-skeleton to each node as a $\cO(\log n)$-bit certificate, which can be checked using \Cref{lem:verificationofpskeleton}. Finally, the correctness of $\widehat{\ell}$ can be certified with $\cO(\log n)$-bit certificates, using the scheme of \Cref{lem:reductionofexistentialverif} with given values
  $(\Lambda^\xi, (\exists y, \xi(\overline{x},y)), \ell^\xi)$ and $(p, \widehat{\Lambda},\widehat{\varphi}, \widehat{\ell})$ and $S$.
\end{proof}

\subsection{Proof of \Cref{thm:thmFOexpCER}} 

We are now ready to prove     \Cref{thm:thmFOexpCER}. 

\begin{proof}
Let $\varphi \in \FO[\Lambda]$. We proceed the same way as in \Cref{thm:thmFOexpCONG}. Specifically, let us redefine $\varphi(x)$, as $\varphi$ with a dummy free-variable~$x$. The prover provides each node with the following certificate:
\begin{itemize}
    \item[(i)] a positive integer $p$, a set $\widehat{\Lambda}$, a $p$-$\lca$-reduced formula $\widehat{\varphi}$, and a label set~$\widehat{\ell}(u)$;
    \item[(ii)] the certificate of $u$ used to certify the $p$-skeleton $S$ of $G$;
    \item [(iii)] the certificate $\Cert(u)$ for certifying that $(p, \widehat{\Lambda}, \widehat{\varphi}, \widehat{\ell})$ is the reduction of $(\Lambda, \varphi, S, \ell)$. 
\end{itemize}
The part (i) of the certificate can be encoded in $\cO(1)$ bits, the part (ii) can be encoded in $\cO(\log n)$ bits, and \Cref{lem:reductiongeneralverif} asserts that a certificate for (iii) exists and, furthermore, it can be encoded in $\cO(\log n)$ bits. 
 
The verifier proceeds as follows. Each node checks that $(p, \widehat{\Lambda}, \widehat{\varphi}, \widehat{\ell})$ is the reduction of $(\Lambda, \varphi, S, \ell)$ using the certificates given in~(iii) using the  verifier $\mathcal{A}$ in \Cref{lem:reductiongeneralverif}, with certificates (i) and (ii) provided to the nodes. If a node rejects in $\mathcal{A}$, then the node rejects, otherwise it carries on the verification as described hereafter. 

Assuming that no nodes have rejected so far, we have that $\true(\varphi, G,\ell) = \true(\widehat{\varphi}, G, \widehat{\ell})$. Since $\varphi$ has only one dummy variable, the truth value of $\varphi(x)$, as well as the one of $\widehat{\varphi}(x)$, is the same for every node in~$G$. Moreover, the formula $\widehat{\varphi}(x)$ is a Boolean combination of label predicates of~$x$. Therefore, every node~$u$ can check $\widehat{\varphi}(u)$ using $\widehat{\ell}(u)$, and it accepts if $\widehat{\varphi}(u)$ holds, and rejects otherwise.

It remains to analyze soundness and completeness. 

For establishing completeness of the certification scheme, let us suppose that $G\models \varphi$. Then, by \Cref{lem:reductiongeneral},  we have that $\varphi$ is reducible in $\mathcal{G}$. Then there exist certificates for (i) and~(ii), and, by \Cref{lem:reductiongeneralverif}, there exist certificates for (iii) such that every node accepts when running verifier~$\mathcal{A}$. Since $G\models \varphi$ we have that $\true(\varphi, G) = \true(\widehat{\varphi}, G, \widehat{\ell}) = V(G)$. It follows that  $\widehat{\varphi}(u)$ is true at every node $u$,  and thus every node accepts. 

For establishing soundness,  let us suppose that $G \not\models \varphi$, and let us suppose that the nodes received certificates (i)-(iii) such that every node accepts when renning the verifier~$\mathcal{A}$. Then $\true(\varphi, G,\ell) = \true(\widehat{\varphi}, G, \widehat{\ell})$. However, in this case $\true(\varphi, G,\ell) = \emptyset$.  Therefore $\widehat{\varphi}(u)$ does not hold at any node $u$, and thus all nodes reject. 
\end{proof}

\section{Lower Bounds and Impossibility Results}
\label{se:lowerbounds} 

In this section we give evidence that our results somehow represent the limit of tractability of distributed model checking and distributed certification of properties on graphs of bounded expansion. 

All our reductions are based on classical results in communication complexity. Let $f: X \times Y \rightarrow \{0,1\}$ be a function. The \emph{deterministic communication complexity of $f$} $D(f)$ corresponds to the minimum number of bits that must be exchanged  in any deterministic protocol between a player Alice holding an input $x \in X$ and another player Bob $y\in Y$ in order to output the value $f(x,y)$. The \emph{non-deterministic communication complexity of $f$}, $N(f)$ is defined similarly, but for non-deterministic communication protocols. We refer to the book of Kushilevitz and Nissan for more detailed definitions \cite{kushilevitz1997communication}. 

In important problem in communication complexity is \emph{set disjointedness}. It is defined by function $\Disj_n$. For each $n$ the inputs of $\Disj_n$ are subsets of $[n]$. Given $A, B\subseteq [n]$ we have that $\Disj(A,B) = 1$ if and only if $A\cap B = \emptyset$.  
\begin{proposition}[\cite{kushilevitz1997communication}]\label{prop:lowerboundDISJ}
    $D(\Disj_n) = \Omega(n)$ and $N(\Disj_n) = \Omega(n)$. 
\end{proposition}

\subsection{No Extension to Monadic Second-Order Logic}

We first show that it is not possible to extend our results to the whole set of MSO properties on graphs of bounded expansion. We illustrate this fact by considering the property of being non-three-colorable.  We define $\ntcol$ the set of graphs that do not admit a proper three-coloring (their chromatic number is at least $4$). The membership in $\ntcol$  can be represented by the following MSO formula:
$$\forall C_1, C_2, C_3, \exists x,y,~
 \adj(x,y) \wedge \big((x\in C_1 \wedge y \in C_1)\vee (x\in C_2 \wedge y \in C_2) \vee (x\in C_3 \wedge y \in C_3)\big) $$

In the following theorem, we show that there are no efficient \CONGEST\ algorithms for checking nor compact certification scheme  for certifying the membership in $\ntcol$, even under the promise that the input graph belongs to a class of graphs of bounded expansion and of logarithmic diameter.

\begin{theorem}\label{thm:lowerboundMSOinEXP}
Let $\mathcal{G}$ the class of graphs of maximum degree $5$ and diameter $\cO(\log n)$.
Every \CONGEST\ algorithm deciding the membership in $\ntcol$ of a $n$-node input graph in $\mathcal{G}$ requires $\Omega(n/\log^2 n)$ rounds of communication. 

Every certification scheme  that certifies membership to $\ntcol$ of a $n$-node input graph in $\mathcal{G}$  requires certificates of size $\Omega(n/\log n)$.
\end{theorem}

\begin{proof}
  \cite{GoosS16} show that every locally checkable proof, and therefore every proof-labeling scheme  for $\ntcol$ requires certificates of size $\Omega(n^2/\log n)$. The same proof directly implies the same lowerbound for the number of rounds of a \CONGEST\ algorithm deciding the membership in $\ntcol$. We give an sketch of the proof of Göös and Suomella and then we show how to adapt it to the case when the input graph has maximum degree $4$.

  The proof of \cite{GoosS16} is consists in a reduction to $\Disj_{n^2}$. In the reduction each instance $(A,B)$ of $\Disj_{n^2}$ is associated with a $\cO(n)$-node graph $G_{A,B}$. The set of vertices of $G_{A,B}$ is divided in two sets of vertices, denoted $V_A$ and $V_B$, and satisfies the following conditions:
  \begin{itemize}
      \item[(C1)] The graph induced by $V_A$ only depends on $A$.
      \item[(C2)] The graph induced by $V_B$ only depends on $B$.
    \item[(C3)] The set of edges $E_{A,B}$ with one endpoint in $V_A$ and the other in $V_B$ contains at most $\cO(\log n)$ edges, and does not depend on $A$ and $B$, only on $n$. The nodes with endpoints in $E_{A,B}$ are denoted $V_{A,B}$.
      \item[(C4)] $A\cap B = \emptyset$ if and only if $G_{A,B} \in \ntcol$.
  \end{itemize}
Then, given a certification scheme  $\mathcal{P}$ for $\ntcol$ with certificates of size $c(n)$ on $n$-node graphs, we build the following non-deterministic protocol for $\Disj_{n^2}$. First, Alice and Bob build graph $G_{A,B}$, where Alice builds the subgraph depending on $A$, i.e., $V_A$ and all edges incident to $V_A$, and Bob the part depending on $B$. Then, non-deterministically each player generate the certificates of all the nodes of $G_{A,B}$, and communicate to the other player the certificates generated for the nodes in the cut. Then, the players run the verification algorithm of $\mathcal{P}$ and accept if all the nodes on their side accept. We have that both players accept if and only if $G_{A,B}\in \ntcol$, which holds if and only if $A\cap B = \emptyset$. The total communication is $\cO(\log n \cdot c(n))$. Using \Cref{prop:lowerboundDISJ} we obtain that $\cO(\log n \cdot c(n)) = N(\Disj_{n^2}) = \Omega(n^2)$. Therefore, $c(n) = \Omega(n^2/\log n)$. 

The construction is also useful for a lowerbound in the \CONGEST\ model. Let $\mathcal{A}$ be an algorithm deciding the membership in $\ntcol$ in $\mathcal{R}(n)$ rounds. Then Alice and Bob construct $G_{A,B}$ as above, and simulate the $\mathcal{R}(n)$ rounds of $\mathcal{A}$ on it, communicating on each round to other player the messages transmitted over the edges $E_{A,B}$. The players accept if all the nodes of their sides accept. On each communication round the players send $\cO(\log n)$ bits per edge of $E_{A,B}$. Hence, the nodes interchange $\cO(R(n)\cdot |E_{A,B}|\cdot \log n) = \cO(R(n)\cdot \log^2 n)$ bits in total. From \Cref{prop:lowerboundDISJ} we obtain that $\cO(R(n)\cdot \log^2 n) = D(\Disj_{n^2}) = \Omega(n^2)$. Therefore,  $R(n) = \Omega(n^2/\log^2 n)$.\\

We now use the classical reduction from Garey, Johnson and Stockmeyer \cite{garey1974some}, used to show that 3-colorability on planar graphs of maximum degree $4$ is $NP$-Complete.  Given an $n$-node graph $G$, the reduction  picks an arbitrary graph $G$ and constructs a graph $\widetilde{G}$ of maximum degree $4$, such that $G\in \ntcol$ if and only if $\widetilde{G} \in \ntcol$. Roughly, the we can construct $\widetilde{G}$ replacing each node $v$ of degree greater than $4$ by a node gadget, as the one represented in \Cref{fig:Thm5-2}. Let us denote $d_v$ the degree of $v$. The gadget consists in a set $S_v$ of $7(d_v-2) + 1$ nodes of degree at most $4$. In $S_v$ there are $d_v$ nodes $v_1, \dots, v_{d_v}$ satisfying that on every proper three-coloring of $\widetilde{G}$, these nodes must be colored the same. Furthermore, every coloring that assigns $v_1, \dots, v_{d_v}$  the same color can be extended to a proper $3$-coloring of the whole gadget. For each node in $G$, consider an arbitrary ordering of its neighbors. Let us pick an edge $\{u,v\}$  of $G$, such that $v$ is the $i$-th neighbor of $u$ and $v$ is the $j$-th neighbor of $v$. Then, in $\widetilde{G}$ we add the edge $\{u_i, v_j\}$. By construction satisfies that $\widetilde{G}$ has the desired properties (c.f. \cite{garey1974some} for more details). Observe that in the worst case  $\widetilde{G}$ has $\sum_{v\in V(G)}{7(d_v-2) + 1} = \Theta(n^2)$ nodes.

\begin{figure}[h]
    \centering
    \scalebox{0.75}{

\tikzset{every picture/.style={line width=0.75pt}} 

\begin{tikzpicture}[x=0.75pt,y=0.75pt,yscale=-1,xscale=1]

\draw  [line width=1.5]   (409.5, 300) circle [x radius= 15.91, y radius= 15.91]   ;
\draw  (409.5, 300) node     {$\ v\ $};
\draw  [line width=1.5]   (340.5, 289.5) circle [x radius= 16.99, y radius= 16.99]   ;
\draw (340.5, 289.5) node    {$w_{1}$};
\draw  [line width=1.5]   (357.5, 239.5) circle [x radius= 16.99, y radius= 16.99]   ;
\draw (357.5, 239.5) node     {$w_{2}$};
\draw  [line width=1.5]   (408.5, 220.5) circle [x radius= 16.99, y radius= 16.99]   ;
\draw (408.5, 220.5) node     {$w_{3}$};
\draw  [line width=1.5]   (457.5, 239.5) circle [x radius= 16.99, y radius= 16.99]   ;
\draw  (457.5, 239.5) node    {$w_{4}$};
\draw  [line width=1.5]   (478.5, 290.5) circle [x radius= 16.99, y radius= 16.99]   ;
\draw (478.5, 290.5) node     {$w_{5}$};
\draw  [line width=1.5]   (81.5, 579.5) circle [x radius= 16.99, y radius= 16.99]   ;
\draw (81.5, 579.5) node     {$w_{1}$};
\draw  [line width=1.5]   (136, 580.5) circle [x radius= 15.91, y radius= 15.91]   ;
\draw (136, 580.5) node     {$v_{1}$};
\draw  [line width=1.5]   (191.5, 579) circle [x radius= 13.9, y radius= 13.9]   ;
\draw (185,571.4) node [anchor=north west][inner sep=0.75pt]    {$$};
\draw  [line width=1.5]   (247.5, 580) circle [x radius= 13.9, y radius= 13.9]   ;
\draw (241,572.4) node [anchor=north west][inner sep=0.75pt]    {$$};
\draw  [line width=1.5]   (307.5, 579) circle [x radius= 13.9, y radius= 13.9]   ;
\draw (301,571.4) node [anchor=north west][inner sep=0.75pt]    {$$};
\draw  [line width=1.5]   (167.5, 510) circle [x radius= 13.9, y radius= 13.9]   ;
\draw (161,502.4) node [anchor=north west][inner sep=0.75pt]    {$$};
\draw  [line width=1.5]   (223.5, 511) circle [x radius= 13.9, y radius= 13.9]   ;
\draw (217,503.4) node [anchor=north west][inner sep=0.75pt]    {$$};
\draw  [line width=1.5]   (277.5, 510) circle [x radius= 13.9, y radius= 13.9]   ;
\draw (271,502.4) node [anchor=north west][inner sep=0.75pt]    {$$};
\draw  [line width=1.5]   (226, 460.5) circle [x radius= 15.91, y radius= 15.91]   ;
\draw  (226, 460.5)  node {$v_{2}$};
\draw  [line width=1.5]   (224.5, 410.5) circle [x radius= 16.99, y radius= 16.99]   ;
\draw  (224.5, 410.5)  node     {$w_{2}$};
\draw  [line width=1.5]   (363.5, 579) circle [x radius= 13.9, y radius= 13.9]   ;
\draw (357,571.4) node [anchor=north west][inner sep=0.75pt]    {$$};
\draw  [line width=1.5]   (417.5, 581) circle [x radius= 13.9, y radius= 13.9]   ;
\draw (411,573.4) node [anchor=north west][inner sep=0.75pt]    {$$};
\draw  [line width=1.5]   (473.5, 581) circle [x radius= 13.9, y radius= 13.9]   ;
\draw (467,573.4) node [anchor=north west][inner sep=0.75pt]    {$$};
\draw  [line width=1.5]   (337.5, 511) circle [x radius= 13.9, y radius= 13.9]   ;
\draw (331,503.4) node [anchor=north west][inner sep=0.75pt]    {$$};
\draw  [line width=1.5]   (393.5, 512) circle [x radius= 13.9, y radius= 13.9]   ;
\draw (387,504.4) node [anchor=north west][inner sep=0.75pt]    {$$};
\draw  [line width=1.5]   (447.5, 511) circle [x radius= 13.9, y radius= 13.9]   ;
\draw (441,503.4) node [anchor=north west][inner sep=0.75pt]    {$$};
\draw  [line width=1.5]   (396, 461.5) circle [x radius= 15.91, y radius= 15.91]   ;
\draw (396, 461.5) node     {$v_{3}$};
\draw  [line width=1.5]   (394.5, 410.5) circle [x radius= 16.99, y radius= 16.99]   ;
\draw (394.5, 410.5) node     {$w_{3}$};
\draw  [line width=1.5]   (533.5, 580) circle [x radius= 13.9, y radius= 13.9]   ;
\draw (527,572.4) node [anchor=north west][inner sep=0.75pt]    {$$};
\draw  [line width=1.5]   (587.5, 581) circle [x radius= 13.9, y radius= 13.9]   ;
\draw (581,573.4) node [anchor=north west][inner sep=0.75pt]    {$$};
\draw  [line width=1.5]   (646, 581.5) circle [x radius= 15.91, y radius= 15.91]   ;
\draw (646, 581.5) node    {$v_{5}$};
\draw  [line width=1.5]   (507.5, 511) circle [x radius= 13.9, y radius= 13.9]   ;
\draw (501,503.4) node [anchor=north west][inner sep=0.75pt]    {$$};
\draw  [line width=1.5]   (563.5, 512) circle [x radius= 13.9, y radius= 13.9]   ;
\draw (557,504.4) node [anchor=north west][inner sep=0.75pt]    {$$};
\draw  [line width=1.5]   (617.5, 511) circle [x radius= 13.9, y radius= 13.9]   ;
\draw (611,503.4) node [anchor=north west][inner sep=0.75pt]    {$$};
\draw  [line width=1.5]   (566, 461.5) circle [x radius= 15.91, y radius= 15.91]   ;
\draw (566, 461.5) node     {$v_{4}$};
\draw  [line width=1.5]   (566.5, 411.5) circle [x radius= 16.99, y radius= 16.99]   ;
\draw (566.5, 411.5) node     {$w_{4}$};
\draw  [line width=1.5]   (694.5, 580.5) circle [x radius= 16.99, y radius= 16.99]   ;
\draw (694.5, 580.5) node     {$w_{5}$};
\draw [line width=1.5]    (408.71,237.48) -- (409.3,284.09) ;
\draw [line width=1.5]    (368.57,252.38) -- (399.13,287.93) ;
\draw [line width=1.5]    (357.29,292.06) -- (393.77,297.61) ;
\draw [line width=1.5]    (446.94,252.81) -- (419.39,287.53) ;
\draw [line width=1.5]    (461.67,292.82) -- (425.27,297.83) ;
\draw [line width=1.5]    (98.48,579.81) -- (120.09,580.21) ;
\draw [line width=1.5]    (151.91,580.07) -- (177.6,579.38) ;
\draw [line width=1.5]    (205.4,579.25) -- (233.6,579.75) ;
\draw [line width=1.5]    (261.4,579.77) -- (293.6,579.23) ;
\draw [line width=1.5]    (161.83,522.7) -- (142.49,565.97) ;
\draw [line width=1.5]    (172.07,523.13) -- (186.93,565.87) ;
\draw [line width=1.5]    (181.4,510.25) -- (209.6,510.75) ;
\draw [line width=1.5]    (217.58,523.58) -- (197.42,566.42) ;
\draw [line width=1.5]    (228.07,524.13) -- (242.93,566.87) ;
\draw [line width=1.5]    (263.6,510.26) -- (237.4,510.74) ;
\draw [line width=1.5]    (283.04,522.75) -- (301.96,566.25) ;
\draw [line width=1.5]    (272.02,522.78) -- (252.98,567.22) ;
\draw [line width=1.5]    (213.85,470.78) -- (178.11,501.02) ;
\draw [line width=1.5]    (237.47,471.53) -- (267.48,500.37) ;
\draw [line width=1.5]    (225.01,427.48) -- (225.52,444.59) ;
\draw [line width=1.5]    (377.39,579.51) -- (403.61,580.49) ;
\draw [line width=1.5]    (431.4,581) -- (459.6,581) ;
\draw [line width=1.5]    (342.47,523.99) -- (358.53,566.01) ;
\draw [line width=1.5]    (351.4,511.25) -- (379.6,511.75) ;
\draw [line width=1.5]    (387.82,524.69) -- (369.18,566.31) ;
\draw [line width=1.5]    (398.07,525.13) -- (412.93,567.87) ;
\draw [line width=1.5]    (433.6,511.26) -- (407.4,511.74) ;
\draw [line width=1.5]    (452.34,524.04) -- (468.66,567.96) ;
\draw [line width=1.5]    (442.02,523.78) -- (422.98,568.22) ;
\draw [line width=1.5]    (383.85,471.78) -- (348.11,502.02) ;
\draw [line width=1.5]    (407.47,472.53) -- (437.48,501.37) ;
\draw [line width=1.5]    (395,427.48) -- (395.53,445.59) ;
\draw [line width=1.5]    (331.89,523.72) -- (313.11,566.28) ;
\draw [line width=1.5]    (349.6,579) -- (321.4,579) ;
\draw [line width=1.5]    (547.4,580.26) -- (573.6,580.74) ;
\draw [line width=1.5]    (601.4,581.12) -- (630.09,581.36) ;
\draw [line width=1.5]    (512.4,524.01) -- (528.6,566.99) ;
\draw [line width=1.5]    (521.4,511.25) -- (549.6,511.75) ;
\draw [line width=1.5]    (557.89,524.72) -- (539.11,567.28) ;
\draw [line width=1.5]    (568.07,525.13) -- (582.93,567.87) ;
\draw [line width=1.5]    (603.6,511.26) -- (577.4,511.74) ;
\draw [line width=1.5]    (622.71,523.89) -- (640.03,566.74) ;
\draw [line width=1.5]    (612.02,523.78) -- (592.98,568.22) ;
\draw [line width=1.5]    (553.85,471.78) -- (518.11,502.02) ;
\draw [line width=1.5]    (577.47,472.53) -- (607.48,501.37) ;
\draw [line width=1.5]    (566.33,428.48) -- (566.16,445.59) ;
\draw [line width=1.5]    (501.43,523.51) -- (479.57,568.49) ;
\draw [line width=1.5]    (519.6,580.23) -- (487.4,580.77) ;
\draw [line width=1.5]    (677.52,580.85) -- (661.91,581.17) ;

\end{tikzpicture}}
    
    \caption{Node gadget.  In this figure is depicted the gadget for a node of degree $5$. Notice that in every proper three-coloring of the gadget the nodes $v_1, \dots, v_5$  receive the same colors. Conversely, if we color the nodes $v_1, \dots, v_5$ with the same color, we can extend the coloring into a three-coloring of all the nodes of the gadget.  }
    \label{fig:Thm5-2}
\end{figure}
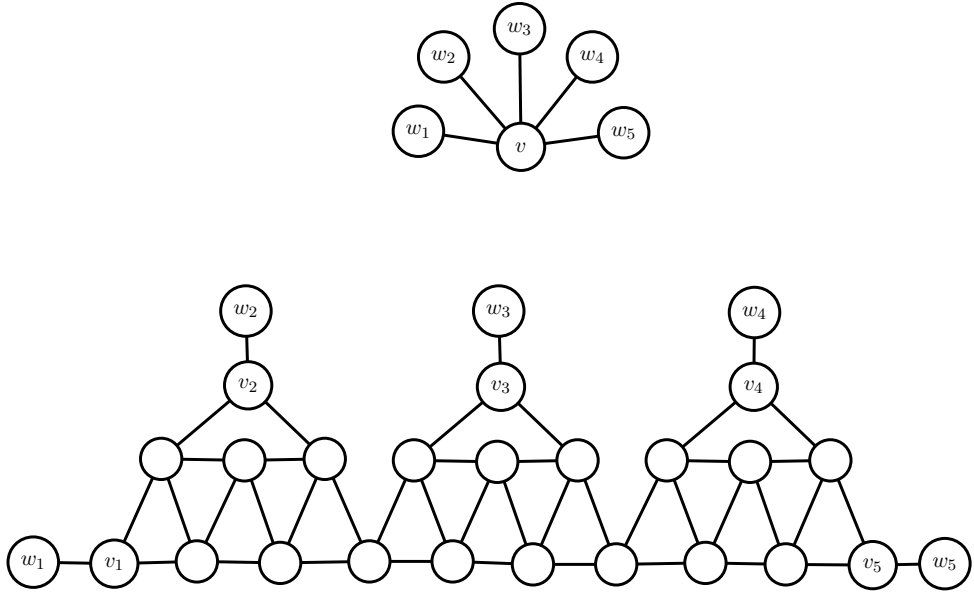

 The statement of the theorem is obtained by plugging the reduction construction of Garey, Johnson and Stockmeyer to the construction of Göös and Suomella. More precisely, given an instance $(A,B)$ of $\Disj_n$, we build graph $G_{A,B}$. Then, we apply the construction of Garey, Johnson and Stockmeyer on $G_{A,B}$ to obtain a $\widetilde{G}_{A,B}$ of maximum degree $4$. Observe that  $\widetilde{G}_{A,B}\in \ntcol$ if and only if $A\cap B = \emptyset$.  We partition the nodes of $\widetilde{G}_{A,B}$ into $\widetilde{V}_A$ and $\widetilde{V}_B$. All the nodes inside gadgets created for nodes in $V_A$ (respectively $V_B$) belong to $\widetilde{V}_A$ (respectively $\widetilde{V}_B$).  We obtain that $\widetilde{G}_{A,B}$ satisfies conditions (C1)-(C4) and has maximum degree $4$. However, $\widetilde{G}_{A,B}$ might have a large diameter. 
 
We now reduce the diameter constructing a graph $\widehat{G}_{A,B}$ obtained from $\widetilde{G}_{A,B}$ by plugging to it two rooted binary trees $T_A$ and $T_B$, where each edge of $T_A$ and $T_B$ is subdivided. Tree $T_A$ has depth $\cO(\log |\widetilde{V_A}|)$, and each leaf corresponds to a node of $\widetilde{V_A}$. The tree $T_B$ is defined analogously for $\widetilde{V}_B$. Finally, add an edge between the root of $T_A$ and the root of $T_B$. Let us define $\widehat{V}_A = \widetilde{V}_A \cup V(T_A)$ and $\widetilde{V}_B = \widetilde{V}_B \cup V(T_B)$.  We obtain that $\widehat{G}_{A,B}$ also satisfies conditions (C1)-(C4) and belongs to $\mathcal{G}$. Conditions C1 and C2 follow from the construction of $\widehat{V}_A$ and $\widehat{V}_B$. Condition C3 holds because in $\widehat{G}_{A,B}$ there is only one more edge in the cut than $\widetilde{G}_{A,B}$ consisting in the edge between the roots. Finally, for condition (C4) we observe that any $3$-coloring of the nodes in $\widetilde{G}_{A,B}$ can be extended to a three coloring of the trees $T_A$ and $T_B$. First, pick two different colors for the roots of $T_A$ and $T_B$. Then, color arbitrarily the other nodes of degree $3$ in $T_A$ and $T_B$. Then, color the nodes of degree two (the subdivided edges) with a color different than their endpoints. We deduce that $\widehat{G}_{A,B}$ satisfies conditions (C1)-(C4) and belongs to $\mathcal{G}$.

Let $\mathcal{P}$ be a certification scheme  (resp. let $\mathcal{A}$ be a \CONGEST\ algorithm)  for $\ntcol$ with certificates of size $c(n)$ (resp. running in $R(n)$ rounds) under the promise that the input graph is a $n$-node graph contained in $\mathcal{G}$.  Let $A,B\subseteq[n]$ be an instance of $\Disj_{n^2}$. Consider the communication protocol where Alice and Bob construct $\widehat{G}_{A,B}$ and simulate $\mathcal{P}$ (resp. $\mathcal{A}$) in the same way that we explained for $G_{A,B}$. Since $\widehat{G}_{A,B}$ is a graph with $\Theta(n^2)$ nodes, we obtain that $\cO(c(n^2)\cdot \log n) = N(\Disj_{n^2}) = \Omega(n^2)$, hence $c(n^2) = \Omega(n^2/\log n)$, and then $c(n) = \Omega(n/\log n)$. Similarly, we obtain that $R(n) = \Omega(\sqrt{n}/\log^2(n))$.
\end{proof}

\paragraph{Remark. }

\Cref{thm:lowerboundMSOinEXP} implies that the difficulty of certifying MSO properties over graphs of max degree 5 is (up to logarithmic factors) as hard as the hardest problem over that family.  Indeed, a graph of max degree $5$ can be encoded using $\cO(n \log n)$ bits, e.g. by adjacency lists of all nodes. We can use such encoding to design a certification scheme  that certifies any property on graphs of maximum degree $5$: every node interprets its certificate as an encoding of a graph of maximum degree $5$, and checks that every neighbor received the same encoding, and that its local view corresponds with the one in the encoding. If the conditions are verified, the nodes can compute any property of the input graph without any further communication. 

\medskip

We can also show that we cannot extend the certification of first-order properties to monadic second-order properties for planar graphs. However, in this case we obtain a weaker lowerbound. 

\begin{theorem}\label{thm:lowerboundMSOinPlanar}
Every certification scheme  for membership to $\ntcol$ under the promise that the input graph is planar requires certificates of size $\Omega(\sqrt{n}/\log n)$ bits.

\end{theorem}
Similar to the proof of \Cref{thm:lowerboundMSOinEXP} we use the construction of Göös and Suomella and the reduction of Garey, Johnson and Stockmeyer for planar graphs.  Given an $n$-node graph $G$, the reduction of \cite{garey1974some} for planar graphs first pick an arbitrary embedding of $G$ in the plane. In this embedding, there might be a number $\cross(G)$ of edge crossings (they must exist if $G$ is non-planar). Then, a planar graph $\widetilde{G}$ is constructed from $G$ by replacing on point where a pair of edges cross, by the \emph{ crossing gadget} shown in \Cref{fig:Thm5-1}. This gadget allows to transmit the coloring of one of the endpoints of the edge in a planar way. For each crossing, we add 13 new vertices. In the worst case, $\cross(G) = \Theta(n^4)$. Hence, the obtained graph has $\widetilde{G}$ has $\Theta(n^4)$ nodes. 

\begin{proof}
 \begin{figure}
    \centering
    \scalebox{0.75}{

\tikzset{every picture/.style={line width=0.75pt}} 

\begin{tikzpicture}[x=0.75pt,y=0.75pt,yscale=-1,xscale=1]

\draw  [line width=1.5]   (570.5, 171) circle [x radius= 15.91, y radius= 15.91]   ;
\draw (570.5, 171) node  {$\ x_2\ $};
\draw  [line width=1.5]   (360.5, 170) circle [x radius= 15.91, y radius= 15.91]   ;
\draw (360.5, 170) node     {$\ x_1\ $};
\draw  [line width=1.5]   (465.5, 70) circle [x radius= 15.91, y radius= 15.91]   ;
\draw (465.5, 70) node     {$\ y_1\ $};
\draw  [line width=1.5]   (465.5, 273) circle [x radius= 15.91, y radius= 15.91]   ;
\draw (465.5, 273)  node    {$\ y_2\ $};
\draw  [line width=1.5]   (465.5, 171) circle [x radius= 15.91, y radius= 15.91]   ;
\draw (456,163.4) node [anchor=north west][inner sep=0.75pt]    {$\ \ \ \ $};
\draw  [line width=1.5]   (414.5, 171) circle [x radius= 15.91, y radius= 15.91]   ;
\draw (405,163.4) node [anchor=north west][inner sep=0.75pt]    {$\ \ \ \ $};
\draw  [line width=1.5]   (413.5, 121) circle [x radius= 15.91, y radius= 15.91]   ;
\draw (404,113.4) node [anchor=north west][inner sep=0.75pt]    {$\ \ \ \ $};
\draw  [line width=1.5]   (465.5, 121) circle [x radius= 15.91, y radius= 15.91]   ;
\draw (456,113.4) node [anchor=north west][inner sep=0.75pt]    {$\ \ \ \ $};
\draw  [line width=1.5]   (414.5, 221) circle [x radius= 15.91, y radius= 15.91]   ;
\draw (405,213.4) node [anchor=north west][inner sep=0.75pt]    {$\ \ \ \ $};
\draw  [line width=1.5]   (465.5, 222) circle [x radius= 15.91, y radius= 15.91]   ;
\draw (456,214.4) node [anchor=north west][inner sep=0.75pt]    {$\ \ \ \ $};
\draw  [line width=1.5]   (516.5, 222) circle [x radius= 15.91, y radius= 15.91]   ;
\draw (507,214.4) node [anchor=north west][inner sep=0.75pt]    {$\ \ \ \ $};
\draw  [line width=1.5]   (516.5, 172) circle [x radius= 15.91, y radius= 15.91]   ;
\draw (507,164.4) node [anchor=north west][inner sep=0.75pt]    {$\ \ \ \ $};
\draw  [line width=1.5]   (516.5, 121) circle [x radius= 15.91, y radius= 15.91]   ;
\draw (507,113.4) node [anchor=north west][inner sep=0.75pt]    {$\ \ \ \ $};
\draw [line width=1.5]    (376.41,170.29) -- (398.59,170.71) ;
\draw [line width=1.5]    (430.41,171) -- (449.59,171) ;
\draw [line width=1.5]    (481.41,171.31) -- (500.59,171.69) ;
\draw [line width=1.5]    (532.41,171.71) -- (554.59,171.29) ;
\draw [line width=1.5]    (465.5,136.91) -- (465.5,155.09) ;
\draw [line width=1.5]    (465.5,186.91) -- (465.5,206.09) ;
\draw [line width=1.5]    (465.5,237.91) -- (465.5,257.09) ;
\draw [line width=1.5]    (465.5,85.91) -- (465.5,105.09) ;
\draw [line width=1.5]    (454.14,81.14) -- (424.86,109.86) ;
\draw [line width=1.5]    (401.81,131.8) -- (372.19,159.2) ;
\draw [line width=1.5]    (372.07,180.93) -- (402.93,210.07) ;
\draw [line width=1.5]    (425.64,232.36) -- (454.36,261.64) ;
\draw [line width=1.5]    (476.75,81.25) -- (505.25,109.75) ;
\draw [line width=1.5]    (528.18,131.81) -- (558.82,160.19) ;
\draw [line width=1.5]    (476.75,261.75) -- (505.25,233.25) ;
\draw [line width=1.5]    (528.07,211.07) -- (558.93,181.93) ;
\draw [line width=1.5]    (425.86,159.86) -- (454.14,132.14) ;
\draw [line width=1.5]    (476.75,132.25) -- (505.25,160.75) ;
\draw [line width=1.5]    (505.14,183.14) -- (476.86,210.86) ;
\draw [line width=1.5]    (454.25,210.75) -- (425.75,182.25) ;
\draw [line width=1.5]    (414.5,186.91) -- (414.5,205.09) ;
\draw [line width=1.5]    (481.41,222) -- (500.59,222) ;
\draw [line width=1.5]    (516.5,156.09) -- (516.5,136.91) ;
\draw [line width=1.5]    (449.59,121) -- (429.41,121) ;

\end{tikzpicture}}
    
    \caption{Crossing gadget. In every proper three-coloring of the gadget node label $x_1$ has the same color of node labeled $x_2$, as well as node labeled $y_1$ has the same color of $y_2$.  Conversely, every  three-coloring of $x_1, x_2, y_1$ and $y_2$ that assigns the same color to $x_1$ and $x_2$, and the same color to $y_1$ and $y_2$ can be extended into a three coloring of the gadget.}
    \label{fig:Thm5-1}
\end{figure}
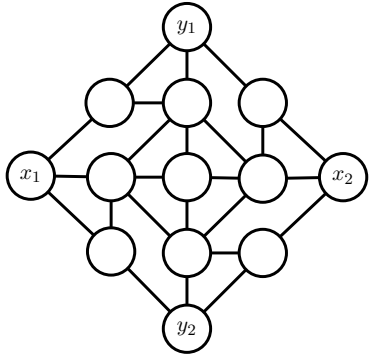

Let $\mathcal{G}$ be the class of planar graphs. As we did for \Cref{thm:lowerboundMSOinEXP}, we plug the reduction construction of Garey, Johnson and Stockmeyer to the construction of Göös and Suomella. More precisely, given an instance $(A,B)$ of $\Disj_n$, we build graph $G_{A,B}$ (c.f. proof of \Cref{thm:lowerboundMSOinEXP}). Then, using the crossing gadget we obtain a planar graph $\widetilde{G}_{A,B}$. We partition the nodes of $\widetilde{G}_{A,B}$ into $\widetilde{V}_A$ and $\widetilde{V}_B$. All the nodes inside gadgets created for edges with both endpoint in $V_A$ (respectively $V_B$) belong to $\widetilde{V}_A$ (respectively $\widetilde{V}_B$).  We also consider that the nodes inside the gadgets used for the edges in the set $E_{A,B}$ of $G_{A,B}$ belong to $V_A$.  We obtain that $\widetilde{G}_{A,B}$ satisfies conditions (C1)-(C4). 

The remaining of the proof is analogous to \Cref{thm:lowerboundMSOinEXP}. Let $\mathcal{P}$ be a certification scheme   for $\ntcol$ with certificates of size $c(n)$ under the promise that the input graph is planar. On input $A,B\subseteq[n]$, Alice and Bob construct $\widetilde{G}_{A,B}$ and simulate $\mathcal{P}$ in the same way that we explained for $G_{A,B}$. Since $\widetilde{G}_{A,B}$ is a graph with $\Theta(n^4)$ nodes, we obtain that $\cO(c(n^4)\cdot \log n) = N(\Disj_{n^2}) = \Omega(n^2)$, hence $c(n^4) = \Omega(n^2/\log n)$, and then $c(n) = \Omega(\sqrt{n}/\log n)$.
\end{proof}

\subsection{No Extension to First Order Logic on Graphs of Bounded Degeneracy}

We now show that we cannot extend the efficient checking nor the certification of first order properties into graphs of degeneracy $2$ and diameter $7$. Consider the set of graphs $\csixfree$ containing all graphs not having a cycle of six nodes as a subgraph. 

\begin{theorem}\label{thm:lowerboundC6degeneracy}
Let $\mathcal{G}$ the class of degeneracy $2$ and diameter at most $6$.
Every \CONGEST\ algorithm deciding the membership in $\csixfree$  under the promise that the input graph belongs to $\mathcal{G}$ requires $\tilde{\Omega}(\sqrt{n})$ rounds of communication. 

Every certification scheme  that certifies membership to $\csixfree$  under the promise that the input graph belongs to $\mathcal{G}$ requires certificates of size $\tilde{\Omega}(\sqrt{n})$.
\end{theorem}

\begin{proof}
The result for \CONGEST\ algorithms and graph of degeneracy $2$ was already shown in Korhonen et al. \cite{korhonen_et_al:LIPIcs:2018:8625},  which in turn use a construction of Druker et al. \cite{DruckerKO13}. Roughly, the lower-bound is based on a reduction to $\Disj_m$, for $m = \Theta(n^{1+1/2})$. Each instance $A,B \in [m]$  is associated with a $\cO(n)$-node graph $G_{A,B}$ of degeneracy $2$. The set of vertices of $G_{A,B}$ are divided in two sets of vertices, denoted $V_A$ (and called \emph{the nodes owned by Alice}) and $V_B$ (\emph{the nodes owned by Bob}), satisfying the following conditions:
  \begin{itemize}
      \item[(B1)] The graph induced by $V_A$ only depends on $A$.
      \item[(B2)] The graph induced by $V_B$ only depends on $B$.
    \item[(B3)] The set of edges $E_{A,B}$ with one endpoint in $V_A$ and the other in $V_B$ contains at most $\cO(n)$ edges, and does not depend on $A$ and $B$. 
      \item[(B4)] $A\cap B = \emptyset$ if and only if $G_{A,B} \in \csixfree$.
  \end{itemize}
We need to deal with one detail. Graph $G_{A,B}$ is not necessarily of bounded diameter. To fix this issue, we build graph $\widehat{G}_{A,B}$ from $G_{A,B}$ by adding a node $r$, connecting it to every node in $V_A \cup V_B$, and then subdividing all the new edges two times. We obtain that every node in $V_A \cup V_B$ is at a distance of 3 from $r$, and then $\widehat{G}_{A,B}$ has a  diameter of at most 6 and still degeneracy $2$. Observe also that $r$ cannot belong to any $C_6$, hence $\widehat{G}_{A,B}$ belongs to $\csixfree$ if and only if $G_{A,B}$ does. We define $\widehat{V}_A$ as $V_A$ plus all the new nodes of $\widehat{G}_{A,B}$, and $\widehat{V}_B$ simply as $V_B$. We obtain that $\widehat{G}_{A,B} \in \mathcal{G}$ and satisfies conditions (B1)-(B4).

Then, analogously to the proof of \Cref{thm:lowerboundMSOinEXP}, we can use $\widehat{G}_{A,B}$ to transform any \CONGEST\ algorithm into a protocol for $\Disj_m$. More precisely, let $\mathcal{A}$ be an $R(n)$-round algorithm that decides the membership in $\csixfree$ of a graph in $\mathcal{G}$. Given an instance $A, B \in [m]$, the players construct $G_{A,B}$, simulate $\mathcal{A}$ in the nodes they own, and in each round communicate to the other player the messages sent from nodes on their side through $E_{A,B}$; in particular, the information communicated in one round is of size $\cO(n \log n)$. We deduce that $\cO(n \log n \cdot R(n)) = D(\Disj_m) = \Omega(n^{1+1/2})$. Therefore, $R(n) = \tilde{\Omega}(\sqrt{n})$. The same construction directly implies a lower bound of $\tilde{\Omega}(n^{1/2})$ for the size of a certificate of any certification scheme  for $\csixfree$.
\end{proof}

\section{Conclusion}

In centralized model checking, the fixed-parameter tractability of FO was extended from graphs of bounded expansion to the more general class of \emph{nowhere dense} graphs in~\cite{GroheKS14}. Notably, nowhere dense graphs represent the ultimate frontier in the program of designing fixed-parameter tractable algorithms for model checking FO formulas. Indeed, as shown in~\cite{DvorakKT13}, for any graph class $\mathcal{G}$ that is not nowhere dense (i.e., $\mathcal{G}$ is \emph{somewhere dense}) and closed under taking subgraphs, FO model checking on $\mathcal{G}$ is not fixed-parameter tractable (unless $\FPT = \W{1}$).
The tractability frontier for FO model checking in the \CONGEST model remains unknown. Moreover, there is no reason to expect that the tractability landscape in the distributed setting should mirror the one of the centralized setting. This raises several concrete and intriguing questions, discussed hereafter.

\subsection{Distributed Computing in Nowhere Dense Graphs}

Is it possible to design model-checking algorithms for FO formulas running in $\cO(D+\polylog(n))$ rounds in \CONGEST on nowhere dense graphs of diameter~$D$? From a theoretical perspective, a negative answer to this question would be even more insightful than a positive one.
More broadly, what is the precise boundary for $\cO(D+\polylog(n))$-round FO model checking in the \CONGEST model? Does it coincide with the class of bounded expansion graphs, extend to nowhere dense graphs, or even include a larger family of graphs? There are several difficulties in extending our approach to nowhere dense graphs. First, it is unclear how to implement the combinatorial arguments from~\cite{GroheKS14}, particularly the constructions of sparse neighborhood covers, in a distributed model. Second, the proof of \cite{GroheKS14} builds on a weak variant of Gaifman's local form, and not a quantifier elimination procedure. While this is sufficient for fixed-parameter tractability of FO model-checking, we do not know whether it could be adapted in the distributed setting.

\subsection{Distributed Computing for Extensions of First-Order Logic}
 
 Another question is whether a meta-theorem with $\cO(D+\polylog(n))$ round complexity can be obtained for some extensions of FO logic on planar graphs. \Cref{thm:lowerboundMSOinEXP} shows that we cannot hope to extend \Cref{thm:thmFOexpCONG} to MSO, even on planar graphs. On the other hand, in centralized computing, the fixed-parameter tractability of logics lying strictly between FO and MSO is an active area of research.
For example, the extension of FO known as \emph{separator logic}, denoted $\textsf{FOL+conn}$, was introduced independently by Schirrmacher, Siebertz, and Vigny~\cite{schirrmacher2022first}, and by Bojańczyk and Pilipczuk~\cite{bojanczyk2016definability}. This logic extends FO, for every $k \geq 1$, with the general predicate $\mathsf{conn}_k(x, y, z_1, \ldots, z_k)$, which evaluates to true on a graph $G$ if the nodes corresponding to $x$ and $y$ are connected by a path that avoids the nodes corresponding to $\{z_1, \ldots, z_k\}$. 
Another example is \emph{disjoint paths logic}, denoted $\textsf{FOL+DP}$, which extends FO with the atomic predicate $\mathsf{dp}_k(x_1, y_1, \ldots, x_k, y_k)$, expressing the existence of $k$ pairwise internally vertex-disjoint paths between $x_i$ and $y_i$ for $i \in \{1, \ldots, k\}$.
The main obstacle to developing distributed meta-theorems for these logics lies in the absence of distributed counterparts of the core algorithmic tools used in the centralized setting---such as tree decompositions 
which underlies tractability results for $\textsf{FOL+conn}$~\cite{pilipczuk2022algorithms}, or the irrelevant vertex technique used for $\textsf{FOL+DP}$~\cite{golovach2023model}.

\bibliographystyle{siam}
\bibliography{biblio}


\end{document}